\documentclass[reqno,12pt]{amsart}

\IfFileExists{mymtpro2.sty}{%
  \usepackage[subscriptcorrection]{mymtpro2}
}{}

\usepackage{a4,amssymb,physics}

\usepackage{appendix}

{

\topmargin -1cm
\textheight23.4cm
\textwidth15.7cm
\oddsidemargin 0.5cm
\evensidemargin 0.5cm
\parindent0.4cm
\newtheorem{theorem}{Theorem}[section]
\newtheorem{proposition}[theorem]{Proposition}
\newtheorem{lemma}[theorem]{Lemma}
\newtheorem{follow}[theorem]{Corollary}
\newtheorem{assumption}[theorem]{Assumption}

\newtheorem{defn}[theorem]{Definition}
\newtheorem{remark}[theorem]{Remark}
\newtheorem{remarks}[theorem]{Remarks}
\newcommand{\bel}{\begin{equation} \label}
\newcommand{\ee}{\end{equation}}

\newcommand{\beq}{\begin{eqnarray*}}
\newcommand{\eeq}{\end{eqnarray*}}
\newcommand{\beal}{\begin{eqnarray}\label}
\newcommand{\eeal}{\end{eqnarray}}
\newsavebox{\toy}
\savebox{\toy}{\framebox[0.65em]{\rule{0cm}{1ex}}}
\newcommand{\QED}{\usebox{\toy}}

\newcommand{\C}{{\mathbb C}}
\newcommand{\R}{{\mathbb R}}
\newcommand{\N}{{\mathbb N}}
\newcommand{\E}{{\mathbb E}}
\newcommand{\Z}{{\mathbb Z}}
\newcommand{\Pp}{{\mathbb P}}

\newcommand{\bea}{\begin{eqnarray}}
\newcommand{\eea}{\end{eqnarray}}
\newcommand{\beas}{\begin{eqnarray*}}
\newcommand{\eeas}{\end{eqnarray*}}

\newcommand{\Schr}{Schr\"odinger }
\newcommand{\Holder}{H\"older }
\newcommand{\labelnummer}{\mbox{\normalfont (\roman{numcount})}}%

\makeatletter

  {\let\curlabelspeicher\@currentlabel%
    \begin{list}{\labelnummer}%
      {\usecounter{numcount}\leftmargin0pt%
        \topsep0.5ex\partopsep2ex\parsep0pt\itemsep0ex\@plus1\p@%
        \labelwidth2.5em\itemindent3.5em\labelsep1em%
      }%
    \let\saveitem\item%
    \def\item{\saveitem%
      \def\@currentlabel{{\upshape\curlabelspeicher}$\,$\labelnummer}}%
    \let\savelabel\label%
    \def\label##1{\savelabel{##1}%
      \@bsphack%
       \ifmmode\else%
          \protected@write\@auxout{}%
          {\string\newlabel{##1item}{{\labelnummer}{\thepage}}}%
        \fi%
      \@esphack%
    }%
  }{\end{list}}%

\renewcommand{\appendix}{\def\thesection{\textsc{Appendix}}}
\begin{document}

\title[Eigenvalue statistics for random polymer models]{Eigenvalue statistics for random polymer models: localization and delocalization}

\author[P.\ D.\ Hislop]{Peter D.\ Hislop}
\address{Department of Mathematics,
    University of Kentucky,
    Lexington, Kentucky  40506-0027, USA}
\email{peter.hislop@uky.edu}

\author[F.\ Nakano]{Fumihiko Nakano}
\address{Mathematical Institute, Tohohu University, Sendai 980-8578, Japan}
\email{fumihiko.nakano.e4@tohoku.ac.jp}

%\thanks{PDH is partially supported by Simons Foundation Collaboration Grant for Mathematicians No.\ 843327.

\begin{abstract}
We study the local eigenvalue statistics (LES) associated with one-dimensional lattice models of random polymers. We consider models constructed from two polymers. Each polymer is a finite interval of lattice points with a finite potential. These polymers are distributed along $\Z$ according to a Bernoulli distribution. The deterministic spectrum for these models is dense pure point, and is known to contain finitely-many critical energies. In this paper, we prove that the LES centered at these critical energies is described by a uniform clock process, and that the LES for the unfolded eigenvalues, centered at any other energy in the deterministic spectrum, is a Poisson point process. 
 These results add to our understanding of these models that exhibit dynamical localization in any energy interval avoiding the critical energies \cite{DSS,dBg2}, and nontrivial transport for wave packets with initial states supported at an integer point \cite{jss}. We show that the projection of these initial states onto spectral subspaces associated with any energy interval that contains all of the critical energies exhibit nontrival transport, refining the connection between nontrivial transport and the critical energies.  Finally, we also prove that the transition in the unfolded LES is sharp at the critical energies.
\end{abstract}

\maketitle \thispagestyle{empty}

\tableofcontents

%\today

%%%%%%%%%%%%%%%%%%%%%%%%%%%%%%%%%%%%%%%%%%%%%%%%%%%%%%%%%%%%%%%%%%%%%%%%%%%%%%%%%%%%%%%%%%%%%%%%%%%%%%%%%%%

\section{Introduction and statement of the problem}\label{sec:introduction}
\setcounter{equation}{0}

Random \Schr operators (RSO) in dimension $d \geq 3$ are expected to exhibit a phase transition with respect to the energy for fixed, nonzero disorder. At low energy, or at band-edge energies, the models are expected to exhibit localization, where at high energy, or near the center of the band, there should be delocalized behavior. This general conjecture remains unproven but several ergodic models have been proven to exhibit both localization and delocalization. In one-dimension, the random dimer model on the lattice and, more generally, random polymer models, have been proven to have both properties \cite{dBg2, jss}. 

This article presents the proof of a sharp localization-delocalization transition for the unfolded local eigenvalue statistics for the random polymer model. Although these models have dense pure point spectrum almost surely, with exponentially decaying eigenvectors, they also have a finite number of critical energies. These are energies at which the Lyapunov exponent vanishes. We prove that the unfolded LES centered at a critical energy is a uniform clock process, and that the unfolded LES centered at any other energy in the deterministic spectrum $\Sigma$ is a Poisson point process.

%%%%%%%%%%%%%%%%%%%%%%%%%%%%%%%%%\subsection{Acknowledgement} PDH is partially supported by Simons Foundation Collaboration Grant for Mathematicians No.\ 843327. FN is partially supported by JSPS KAKENHI Grant Number 20K03659.
%\end{acknowledgement}
%%%%%%%%%%%%%%%%%%%%%%%%%%%%%%%%%%%%%%%%%%

%%%%%%%%%%%%%%%%%%%%%%%%%%%%%%%%%%%%%%%%%%%%%%%%%%%%%%%%%%%%
\subsection{The random dimer model}

The one-dimensional random dimer model \cite{dBg2} is a discrete random \Schr operator on $\ell^2(\Z)$ of the form $H_\omega = L + V_\omega$, where $L$ is the nearest-neighbor, finite-difference Laplacian. The random potential $V_\omega$ is defined by a sequence of random variables $V_\omega(n)$, for $n \in \Z$, with 
\bel{eq:dimer1}
V_\omega(2n) = V_\omega(2n+1), ~~~ n \in \Z .
\ee
A dimer consists of two neighboring sites, say $2n$ and $2n+1$, and a single Bernoulli random variable, say $V_\omega (2n)$, taking values in $\{ -V, V\}$, $V \in (0,1]$.
%, with probability $p$ and $1-p$, $p \in (0,1)$, respectively, associated with the two sites. 
The random variables $\{ V_\omega (2n) \}$ are independent and identically distributed (iid). The Anderson-Bernoulli random dimer model (AB RDM) is one for which the probability measure associated with $V_\omega(0)$, and consequently all $V_\omega (2n)$, is a Bernoulli measure $d\mu(s) = (p \delta(s- V) + (1-p) \delta (s+V)) ~ds$, with $p \in (0,1)$, and $V \in (0,1]$ .  
%The random potential $V_\omega$ is defined by
%\bel{eq:dimer1}
%V_\omega(2n) = V_\omega(2n+1), ~~~ n \in \N .
%\ee
Another way to describe the model is that it is formed from two dimers, one having value $V$ and the other having value $-V$. They are arranged randomly along $\Z$, the occurrence of one or the other controlled by a Bernoulli random variable.  

%taking value equal length polymers, each consisting of two sites. 
%For example, one polymer consists of sites, $2n$ and $2n+1$, with the same random variable. 
%We are interested in the case when the random variables $\{ V_\omega (2n) \}$ are independent and identically distributes (iid). In particular,

The AB RDM, in contrast to the AB random \Schr operator (RSO), is characterized by the existence of a pair of critical energies at which the Lyapunov exponent vanishes. For the AB RDM described above, the two critical energies occur at $E_c = \pm V$. The standard AB RSO does not have critical energies. The single polymer (a single site)  transfer matrices for two different sites never commute almost surely, whereas for the RDM, the two different polymer transfer matrices (products of the single-site transfer matrix over one polymer)  commute at the critical energies. 

An important consequence of this commutation is that Lyapunov exponent $L(E)$ vanishes at the critical energies. This lead to the conjecture that there are extended states associated with critical energies. In their important study, De Bi\`evre and Germinet \cite{dBg2} prove that there is dynamical localization on every energy interval not containing a critical energy $E_c$. After this work, Jitomirskaya, Schultz-Baldes, and Stolz \cite{jss} proved that these models (and the more general random polymer models (RPM)) also exhibit dynamical delocalization in the form of nontrivial transport. The form of dynamical localization proved in \cite{jss} is as follows. For any $\alpha > 0$, there is a finite constant $C_\alpha > 0$ so that:
\bel{deloc1}
M_q(T) := \frac{1}{T} \int_0^\infty e^{- \frac{t}{T}} \langle \delta_0,  e^{i H_\omega t}|X|^q e^{- i H_\omega t} \delta_0 \rangle ~dt \geq C_\alpha T^{q - \frac{1}{2} - \alpha}, a.s.,
\ee
where $\delta_0$ is the function equal to one at $0 \in \Z$, and zero everywhere else. 
%for any $\alpha > 0$, almost surely. 

Such anomalous transport associated with critical energies had been predicted by physicists (see \cite{dBg2} for references). We note that this estimate may be localized in energy intervals containing all the finitely-many critical energies since, in the complementary set, dynamical localization holds (see Appendix \ref{app:deloc1}).   
Hence, the AB RDM exhibits both almost sure pure point spectrum throughout its spectrum almost surely, and dynamical delocalization. 

%%%%%%%%%%%%%%%%%%%%%%%%%%%%%%%%%%%%%%%%%%%

In this paper, we study the unfolded local eigenvalue statistics for the more general AB RPM, as defined in \cite{jss},  and prove that there is an energy-dependent transition from Poisson statistics to uniform clock statistics at the critical energies. 
%We prove that the more general random polymer model, as defined in \cite{jss}, exhibits an energy dependent transition in the unfolded local eigenvalue statistics (LES). We prove that the unfolded LES centered at an energy in the deterministic spectrum that is not a critical energy is a Poisson point process, whereas if the energy is a critical energy, the LES is a clock process. 
We also prove that the rescaled eigenvalue spacing exhibits strong clock behavior \cite{LS} in this case. Consequently, the transition between these two phases as a function of energy is sharp. These models are the first nontrivial applications of unfolded eigenvalue statistics whose use is necessitated due to the fact that the integrated density of states is not Lipschitz continuous.

%%%%%%%%%%%%%%%%%%%%%%%%%%%%%%%%%%%%%%%

\subsection{Random polymer models} \label{subsec:rpm1}

In this section, we describe the random polymer model (RPM). 
Following \cite[section 2]{jss}, a polymer of length $L$ with fixed potential $v$ is a 
finite real sequence $\{ v(j) ,v(j+1), \ldots, v(j+L-1) \}$. We refer to this as polymer $L$.
A deterministic polymer model on $\Z$ is constructed from 
two polymers, with possibly different lengths $L_\pm$ and different potentials $v_\pm$, laid sequentially along $\Z$. 
%Each polymer $L_\pm$ has a different (finite length) potential designated by $v_\pm(j)$. 
By randomizing the choice of the polymer in this construction, we obtain the random polymer model (RPM). The randomization is given by Bernoulli random variables: polymer $L_+$ occurs with probability $p_+ := p$ and polymer $L_-$ occurs with probability $p_- := (1-p)$, with $p \in (0,1)$. We refer to this as the Anderson Bernoulli random polymer model (AB RPM). We also include a random hopping coefficient $t_\omega(n)$, as in \cite{jss}. These models generalize the RDM described above. 

%%%%%%%%%%%%%%%%%%%%%%%%%%%%%%%%%%%%%%%%
%The transfer matrices associated with RPM are the basis of arguments in \cite{jss}. A definition of critical energies is given, following \cite{jss}. 
%%%%%%%%%%%%%%%%%%%%%%%%%%%%%%%%%%%%%
%%\subsection{Definition of random polymer models}
%We study the LES for the general, random polymer models constructed in \cite{jss} and recall that construction here. 
%%%%%%%%%%%%%%%%%%%%%%%%%%%%%%%%%%%
%Briefly, the model consists of two finite subsets of vertices of lengths $L_\pm$. Each subset represents one of the two polymers. Each subset is associated with an independent, identically associated (iid) random variable.  random variable. We also include a random hopping coefficient $t_\omega(n)$ as in \cite{jss}.
%%%%%%%%%%%%%%%%%%%%%%%%%%%%%%%%%%%%%%

In more detail, we associate with each polymer of length $L_{\pm} \in {\bf N}$, the associated $L_\pm$-tuples of hopping terms and potentials. We designate these by
\bea
\hat{t}_{\pm}
&:=&
 \left(
\hat{t}_{\pm} (0), \hat{t}_{\pm} (1), \cdots, \hat{t}_{\pm}(L_{\pm}-1)
\right), 
\quad
\hat{t}_{\pm} (\ell) > 0, 
\;
\ell = 0, 1, \cdots, L_{\pm}-1
\\
\hat{v}_{\pm}
&:=&
\left(
\hat{v}_{\pm} (0), \hat{v}_{\pm} (1), \cdots, \hat{v}_{\pm}(L_{\pm}-1)
\right)
\eea
We consider laying these polymers head-to-tail along the lattice $\Z$ in a random order to create a random juxtaposition of the two polymers. Each configuration of polymers is labeled by a two-sided sequence $\omega := (\omega_\ell)_{\ell \in \Z}$ of $\pm$-signs: $\omega_\ell = \pm$. . 
%An infinite configuration $\omega$ is labeled  by a sequence of $\pm$-signs: $\omega := \{ \omega_{\ell} \}_{\ell \in {\bf Z}}$, with $\omega_{\ell} = + , -$.
We define sequences $t_\omega$ and $v_\omega$ by
\bea
t_{\omega}
 &:=&
( t_{\omega_\ell} )_{\ell \in {\bf Z}}
=
( \cdots, \hat{t}_{\omega_0}, \hat{t}_{\omega_1},  \cdots)
\\
v_{\omega}
 &:=&
( v_{\omega_\ell} )_{\ell \in {\bf Z}}
=
( \cdots, \hat{v}_{\omega_0}, \hat{v}_{\omega_1},  \cdots) .
\eea
Given these definitions, the discrete \Schr operator for a RPM is defined by 
\bea\label{eq:abrpm1}
( H_{\omega} \psi )(n)
&:=&
- t_\omega (n+1) \psi(n+1) - t_\omega (n) \psi(n-1)
+
v_{\omega}(n) \psi(n) ,
\eea
acting on $\ell^2 (\Z)$. 
As mentioned above, the random variables $\omega_\ell$ are iid and distributed according to a Bernoulli distribution with the probability given by 
${\Pp} (\omega_{\ell} = \pm ) = p_{\pm}$, 
where $p_{\pm}$ satisfy
$p_{\pm} \in (0, 1)$,  and $p_+ + p_- = 1$.

The complete construction of the probability space  $(\Pp, \Omega)$ requires a randomization of the values at the origin, $\hat{v}_{\omega_0}(0)$ and $\hat{t}_{\omega_0}(0)$. The details are given in \cite[section 4.1]{jss} and are not essential for the proofs on unfolded LES, 
but from which 
standard results known for ergodic \Schr operators are also valid for 
$H_{\omega}$. 
In particular, 
it follows that there is a closed set 
$\Sigma (\subset \R)$ 
such that 
$\sigma (H_{\omega}) = \Sigma$, 
a.s. 
%

%%%%%%%%%%%%%%%%%%%%%%%%%%%%%%%%%%%%%%

\subsection{Transfer matrices for RPM}

Transfer matrices allow an alternative to the \Schr equation 
$H \psi = E \psi$ in one-dimension. The \Schr equation on $\ell^2(\Z)$ has the form. 
\bea\label{eq:schr1}
t(n+1) \psi(n+1)
&=&
(v(n) - E) \psi (n) - t(n) \psi(n-1) \\
&=&
\frac {1}{ t(n) }
\left\{ (v(n) - E) t(n) \psi(n) - t(n)^2 \psi(n-1)
\right\} .
\eea
This implies an equivalent formulation to \eqref{eq:schr1}:
\bea\label{eq:transfer1}
\left( \begin{array}{c}
t(n+1) \psi(n+1) \\ \psi(n)
\end{array}
\right)
&=&
\frac {1}{t(n)} \left(
\begin{array}{cc}
v(n) - E & - t(n)^2 \\
1 & 0
\end{array} 
\right)
\left(
\begin{array}{c}
t(n) \psi(n) \\ \psi(n-1)
\end{array}
\right)
\\
&=&
T_{v(n)- E, t(n)}
\left(
\begin{array}{c}
t(n) \psi(n) \\ \psi(n-1)
\end{array} 
\right) ,
\eea
where the transfer matrix $T_{v, t}$ is defined by
\bea\label{eq:transfer2}
T_{v, t}
&:=&
\frac 1t
\left(
\begin{array}{cc}
v & -t^2 \\
1 & 0
\end{array}
\right).
\eea

Crucial for our discussion are the two polymer transfer matrices associated with each one of the polymers $L_\pm$. We denote these single polymer transfer matrices by $T^E_\pm$. They are given by:
%The one-block transfer matrices are defined by: 
\bea\label{eq:transfer3}
T^E_{\pm}
&:=&
T_{ \hat{v}_{\pm} (L_{\pm}-1)-E,\, \hat{t}_{\pm}(L_{\pm}-1) }
T_{ \hat{v}_{\pm} (L_{\pm}-2)-E,\, \hat{t}_{\pm}(L_{\pm}-2) }
\cdots
T_{ \hat{v}_{\pm} (0)-E, \,\hat{t}_{\pm}(0) }
\eea
%Using these matrices, we find that the transfer matrix from 
%$m^{th}$ block to $k^{th}$  is given by: 
%\bea
%T_{\omega}^E (k, m)
%&:=&
%T^E_{\omega_{k-1}}
%T^E_{\omega_{k-2}}
%\cdots
%T^E_{\omega_{m}}, 
%\quad
%k > m
%\\
%%
%T^E_{\omega} (k,m)
%&:=&
%T^E_{\omega} (m, k)^{-1}
%\quad
%k < m
%\eea
%
%
%%
%The Lyapunov exponent
%is given by 
%%
%\bea
%\gamma(E)
%&:=&
%\lim_{k \to \infty}
%\frac {1}{ k \langle L_{\pm} \rangle }
%\log
%\| T^E_{\omega} (k, 0) \|, 
%\qquad
%\mbox{ where }
%\langle c_{\pm} \rangle
%:=
%p_+ c_+ + p_- c_-
%\eea
%%
%%
%%
%%%%%%%%%%
%\begin{itembox}[l]
\begin{defn}\label{defn:critical_energy1}
An energy $E_c\in {\R}$ is \textbf{critical} for the AB RPM
$H_{\omega}$ in \eqref{eq:abrpm1} if the single polymer transfer matrices $T^E_\pm$, defined in \eqref{eq:transfer3}, satisfy: 
\begin{enumerate}
\item  $| Tr (T^{E_c}_{\pm}) | < 2$, or $T^{E_c}_{\pm} = \pm I$,
\medskip

\item $[ T^{E_c}_{+}, T^{E_c}_{-} ] = 0$.
\end{enumerate}

\end{defn}
%
%\end{itembox}\\
%%%
%

%
%The first condition says that 
%$E_c$ lies in the interior of the spectrum
%$\sigma (H_{\pm})$ of the corresponding periodic system. 
%
The transfer matrix over several polymers is given by the product of the individual polymer transfer matrices. For a configuration $\omega$, we write  
\bel{eq:klock_trans_matrix1}
T_\omega^E (k, m) = T^E_{\omega_{k-1}} \ldots T^E_{\omega_m}, ~~~ k > m .
\ee
If $m > k$, we have $T^E_\omega(k,m) = T^E_\omega(m,k)^{-1}$, and $T^E_\omega (k,k) = Id$. 
The Lyapunov exponent is given by 
\bea  \label{eq:lyapunov0}
L (E) &:=& \lim_{k \to \infty}
\frac {1}{ k \langle L_{\pm} \rangle }
\log \| T^E_{\omega} (k, 0) \|, 
\qquad
\mbox{ where }
\langle c_{\pm} \rangle
:= p_+ c_+ + p_- c_-.
\eea
An important property of a critical energy $E_c$ is that the Lyapunov exponent vanishes at that critical energy: $L (E_c) = 0$. Since the Lyapunov exponent is positive on energies in the localization regime, the vanishing of the Lyapunov exponent at a critical energy might imply nontrivial transport for wave packets with initial conditions containing a critical energy.

%%%%%%%%%%%%%%%%%%%%%%%%%%%%%%%%%%%%%%%
\subsection{Main results on eigenvalue statistics for RPM}\label{subsec:main1}

Our main results concern the local eigenvalue statistics (LES) for the discrete \Schr operator describing random polymer models on $\ell^2 (\Z)$. 
%\begin{defn}
%Let $E_0 \in \R$ be fixed. The LES $\xi_{E_0}$ at $E_0$ is \textbf{clock} if there exists a random variable $\theta$ with probability measure $\mu$ so that the characteristic function satisfies 
%\beq\label{eq:clock_defn}
%\E  \{ e^{-f(\xi_{E_0})} \}  = \E_\mu \left\{  e^{\sum_{n \in \Z} f( n - \theta} \right\} ,
%\eeq 
%for all $f \in C_0(\R)$.
%\end{defn}
%%%%%%%%%%%%%%%%%%%%%%%%%%%%%%%%%%%%%%
%\subsection{Local operators}\label{subsec:local_op1}
%%%%%%%%%%%%%%%%%%%%%%%%%%%%%%%%%%%%%

\subsubsection{Local operators and the density of states}\label{subsubsec:local_op1}

In order to describe the LES, we define local \Schr operators associated with the interval $[0, L-1] \cap \Z$. We let $H_\omega^L$ is the restriction of $H_\omega$ to 
$\{ 0, 1, \cdots, L-1 \}$ with Dirichlet boundary conditions at 
$n = -1$ and at $n = L$. We label the eigenvalues of $H_\omega^L$ by $\{ E_j (L) \}_{j=1}^L$, in increasing order $E_1(L) < E_2(L) < \cdots < E_L(L)$.

%eigenvalues of 
%$H_L$
%in increasing order.\\
%%
The density of states measure (DOSm) $\mu$ is defined for any bounded Borel subset $B \subset \R$ as the limit
\bel{eq:dos1}
\mu (B) := \lim_{L \rightarrow \infty} \frac{1}{L} Tr( P_{H_\omega^L} (B) ), 
\ee
where 
$P_H (B)$
is the spectral projection of an operator 
$H$
associated to the Borel set 
$B$. 
This limit exists weakly due to the ergodicity of the RSO. 
The integrated density of states ${N}(E)$ is the cumulative distribution function of $\mu$ and is given by
\bel{eq:ids1}
{N}(E) = \int_{-\infty}^E ~d \mu(s) .
\ee
We will abuse the notation and write $N(I) := \mu (I)$ for the DOS measure of a measurable subset $I \subset \R$.  

It is known that ${N}(E)$ is \Holder continuous for the AB model \cite{CKM}. However, we do not know if the density of states measure (DOSm) is absolutely continuous with respect to Lebesgue measure. In the case that the single-site probability measure is absolutely continuous with respect to Lebesgue measure, the density of states function (DOSf) exists and is positive. The DOSf plays a role in the LES in the localization regime.
%$\rho(E)$ : 
%density of states at 
%$E$
%(Radon-Nykodym derivative of IDS
%${N}(E)$.
%\\
%%
%$E_c \in {\bf R}$ : 
%critical energy
%(at which the transfer matrices 
%$T_{\pm}(E)$ 
%commute)
%\\
%%
For this reason, we consider the following point process $\xi_\omega^L$  on 
${\R}$ centered at $E_0 \in \Sigma$ formed from the \textbf{unfolded eigenvalues} of $H_\omega^L$: 
\bel{eq:les1}
\xi_L := 
\xi_L^{E_0}
:=
\sum_{j = 1}^L \delta_{ L ({N}(E_j(L)) -{N}(E_0) )}  . 
\ee
We refer to this random point process, and its limit 
$\xi^{E_0}
:=
\lim_{L \to \infty}
\xi_L^{E_0}$
as $L \rightarrow \infty$, if it exists, as the local eigenvalue statistics (LES) of $H_\omega$ centered at $E_0$.
$E_0$
is called the center energy or the reference energy.
For a point process $\xi$ and any $f \in C_0(\R)$, we write $\xi (f) := \int f(s) d\xi (s)$. 
For the AB RPM, we will consider the LES in two mutually-exclusive cases: 1) when $E_0 \in \mathcal{C}$, the set of critical energies, and 2) when $E_0 \in \Sigma \backslash \mathcal{C}$. 
 
 %it follows from \cite{jss} that the IDS is differentiable with density $\rho(E_j^c)$. Consequently, we consider: 
%\beq\label{eq:les2}
%\xi_\omega^L (ds) := \sum_{j = 1}^L \delta ( L (E_j(L) - E_j^c) - s) ~ds . 
%\eeq

We are interested in two random point processes.

\begin{defn}
Let $E_0 \in \R$ be fixed. The LES $\xi^{E_0}$ centered at $E_0$ is \textbf{clock} if there exists a random variable $\theta$ 
on the unit interval
$[0,1)$
with distribution 
$\mu$ 
such that the characteristic function satisfies
\bel{eq:clock_defn}
\E  \{ e^{- \xi^{E_0}(f)} \}  
= 
\int_0^1 d \mu (\theta)
e^{- \sum_{n \in \Z} f( n - \theta)}, 
%\E_\mu \left\{  e^{- \sum_{n \in \Z} f( n - \theta)} \right\} ,
\ee 
for all $f \in C_0(\R)$. 
We also say
$\xi^{E_0}$
obeys 
clock $(\mu)$.
We call the process a \textbf{uniform clock process} on the interval $[0,1)$ if $\theta$ is uniformly distributed on $[0,1)$ and we write clock $unif [0,1)$.
\end{defn}

\medskip

\begin{defn}
Let $E_0 \in \R$ be fixed. The LES $\xi^{E_0}$ centered at $E_0$ is \textbf{Poisson} with L\'evy measure $\nu$ if the characteristic function satisfies 
\bel{eq:poisson_defn}
\E  \{ e^{- \xi^{E_0}(f)} \}  =   e^{ \int_\R ( e^{-f(x)} - 1) ~d \nu(x)} , 
\ee
for all $0 \leq f \in C_0(\R)$. 
\end{defn}

\medskip

The main result of this article is to characterize the unfolded LES for each $E_0 \in \Sigma$. Our results are summarized in the following theorem. Assumptions and precise statements may be found in Theorem \ref{thm:clock_1} and Theorem \ref{thm:les_loc1}.   
\medskip

\begin{theorem}\label{thm:main1}
Let $H_\omega$ be the random polymer model with Bernoulli single-site probability measure. Let $\mathcal{C} := \{ E^c_j ~|~ j=1, \ldots , N \}$ be the finite set of critical energies. 

\medskip

\begin{enumerate}
\item \textbf{Delocalization:} The unfolded LES centered at a critical energy $E^c_j$ is a uniform clock process.

\medskip

\item \textbf{Localization}: The unfolded LES centered at an energy $E_0  \in \Sigma \backslash \mathcal{C}$ is a Poisson point process with intensity measure given by the Lebesgue measure on $\R$. 
\end{enumerate}
\end{theorem}

%We mention a corollary for the one-dimensional discrete AB RSO model.
%
%\begin{follow}\label{cor:ABmodel}
%Let $H_\omega$ be the random 
%Anderson model with Bernoulli single-site probability measure. 
%The LES centered at an energy $E_0  \in \Sigma$ is a Poisson point process with intensity measure given by the Lebesgue measure on $\R$. 
%\end{follow}

There 
are only few cases in which a Hamiltonian shows more than two completely different LES at different reference energies.
The 
only example known so far is the 1D \Schr operator with critically decaying random potential\cite{KVV, kn2, nakano2}, 
where LES is the Sine$_{\beta}$ process where  
$\beta = \beta (E_0)$
is a non-constant smooth function of the reference energy 
$E_0 \in \Sigma$ 
which takes all values in 
$(0, \infty)$ \cite{nakano2}.
Theorem \ref{thm:main1} now implies that AB RPM is another example.

We obtain another characterization of the limited behavior of the eigenvalues near a critical energy $E_j^c \in \mathcal{C}$. We re-number the eigenvalues of $H_\omega^L$ as
\bel{eq:order1}
\ldots E'_{-2}(L) < E'_{-1}(L) < E_j^c \leq E'_0(L) < E'_1(L) < \ldots.
\ee
We are interested in the rescaled, unfolded eigenvalue spacing around $E_j^c$:
\bel{eq:spacing0}
{N}(E'_{m+1}(L)) - {N}(E'_m(L)) .
\ee
This can be written in a simplified manner. 
It follows from \cite[Theorem 2]{jss} that the IDS is differentiable at $E_j^c$, with the derivative $n(E_j^c) := {{N}}^\prime(E_j^c) > 0$. Consequently, the 
unfolded eigenvalues in \eqref{eq:spacing0} may be expressed as 
\bel{eq:ids_c1} 
{N}(E'_{m+1}(L)) - {N}(E'_m (L)) :=  n(E_j^c) ( E'_{m+1}(L) - E'_m (L))  + o(L^{-1}) .
\ee

\begin{theorem}\label{thm:spacing01}
The local eigenvalue spacing near any critical energy $E_k^c \in \mathcal{C}$ 
satisfies
\bel{eq:spacing1}
\lim_{L \rightarrow \infty} n(E_k^c) L ( E'_{j+1}(L) - E'_j(L) ) = 1, 
~
\text{for any }j \in \Z, ~ a. s. 
\ee\end{theorem}

%The existence of $\rho{E_c} := {{N}}^\prime(E)$, and its positivity, at the critical energies $E_j^c$ was proven in \cite[Theorem 2]{jss}.

\begin{remark}
(1) 
Theorem \ref{thm:spacing01}
implies that eigenvalue spacings asymptotically show a repulsion, and such behavior is called the clock behavior.
There are 
a couple of examples showing clock behavior : 
(i)
The ergodic 
Jacobi matrices, if 
$E_0$
lies in absolutely continuous spectrum \cite{ALS}, 
(ii)
1D \Schr operators with supercritical decaying random or more general potentials \cite{kn2, nakano2, CN}, and 
(iii)
1D \Schr operators with sparse random potentials where 
$E_0$
lies in singular continuous spectrum \cite{BW}.

Our result is 
the first example of clock LES for a random \Schr operator with pure point spectrum. 
\\
(2)
We believe that 
if the LES is clock$(\mu)$, then 
$\mu$
is generically uniform for translation invariant systems. 
Note that 
it is not the case for 1D \Schr operator 
with supercritical random decaying potential \cite{kn2}.
%

%It follows from \cite{jss} that the IDS is differentiable at $E_c$ with density $\rho(E_j^c)$. Consequently, we consider: 
%\beq\label{eq:les2}
%\xi_\omega^L (ds) := \sum_{j = 1}^L \delta (\rho(E_c) L (E_j(L) - E_j^c) - s) ~ds . 
%\eeq
\end{remark}

The method 
of proof of part (2) of Theorem \ref{thm:main1} also applies to the standard Anderson-Bernouilli model.
So, 
we mention a corollary for the one-dimensional discrete AB RSO.

\begin{follow}\label{cor:ABmodel}
Let $H_\omega$ be the random 
Anderson model with Bernoulli single-site probability measure. 
Assume that 
$H_{\omega}$
satisfies Assumption \ref{assump:dos1}.
Then the unfolded LES centered at any energy $E_0  \in \Sigma$ is a Poisson point process with intensity measure given by the Lebesgue measure on $\R$. 
\end{follow}

%%%%%%%%%%%%%%%%%%%%%%%%%%%%
\subsection{Outline} \label{subsec:outline1}

The main results of this paper apply to the Anderson-Bernoulli (AB) random polymer model (RPM) as introduced by \cite{jss}, generalizing the AB random dimer model studied by  \cite{dBg2}. 
In section \ref{sec:clock1}, 
we prove the first part of Theorem \ref{thm:main1} on uniform clock statistics at critical energies for AB RPM. We also prove Theorem \ref{thm:spacing01} on the strong clock behavior of the eigenvalue spacing. 
To prove these theorems, 
we basically follow the strategy in \cite{kn2}.
The Poisson characterization of the unfolded LES for the AB RPM at non-critical energies is presented in section \ref{sec:poisson1}, roughly following ideas in \cite{nakano1, GKlo1} on the eigenvalue statistics.  Because the single-site probability measure is Bernoulli, we prove a Minami estimate partially using the ideas in \cite{bourgain1, klopp1}. We also use bootstrap multiscale analysis (MSA) of Germinet and Klein \cite{gk_boot1} in order to obtain localization bounds. 
Several appendices summarize these results on bootstrap and results of \cite{jss} used in the main proofs.

%%%%%%%%%%%%%%%%%%%%%%%%%%
\subsection{Notation} We use $C > 0$ to denote a finite, nonnegative constant whose value may change from line-to-line, but is independent of the parameters of the problem like $L$ and $E$.

%%%%%%%%%%%%%%%%%%%%%%%%%%%%%%%
\subsection{Acknowledgement} We thank Xiaolin Zeng for discussions on the topics of this paper. PDH is partially supported by Simons Foundation Collaboration Grant for Mathematicians No.\ 843327. FN is partially supported by JSPS KAKENHI Grant Number 20K03659. 

%%%%%%%%%%%%%%%%%%%%%%%%%%%%%%%%%%%
%\section{Random polymer models}\label{sec:models}
%\setcounter{equation}{0}
%\end{document}
%%%%%%%%%%%%%%%%%%%%%%%%%%%%%%%%%%%%%%%%%%%%%%%%%%%%%%%%%%%%

\section{Delocalization: LES is uniform clock at critical energies}\label{sec:clock1}
\setcounter{equation}{0}

This section is devoted to the proof that the LES are given by a uniform clock process
when the approximating processes are centered at a critical energy, by following the strategy in \cite{kn2}. 
We compute the limit as $L \rightarrow \infty$ of the characteristic function 
$\E \{ e^{- \xi_L(f) } \}$, where the process $\xi_L$ is defined by the unfolded eigenvalues relative to a critical energy $E_k^c \in \mathcal{C}$:
\bel{eq:les_process1}
\xi_L :=  \sum_{j=1}^L \delta_{n(E_k^c) L ( E_j(L) - E_k^c ) } .
\ee
Our calculations hold for any critical energy in $\mathcal{C}$. Consequently we write $E_c$ for any $E_k^c$.  We also prove the strong clock characterization that the limit as $L \to \infty$ of the rescaled eigenvalue spacing is a constant. 

We recall that as the DOS function $n(E) := {N}'(E)$ exists at $E = E_c$, 
so that by (\ref{eq:ids_c1}), 
the limit of 
$\xi_L$ 
in (\ref{eq:les_process1}) coincides with that of 
(\ref{eq:les1}).
We also recall that for any random variable
$c = (c_{\pm})$ 
depending on $\omega_{\pm}$, we define the average $\langle c_\pm \rangle := p c_+  + (1-p) c_-$. 
 
%%%%%%%%%%%%%%%%%%%%%%%%%%%%%%%%%%%%%%%

\subsection{The main result on the LES at critical energies}

We recall that the basic transfer matrices for the RPM are defined in \eqref{eq:transfer1} - \eqref{eq:transfer3}. In particular, the one-polymer transfer matrices are:
\bea\label{eq:one_poly_trans1}
T^E_{\pm}
&:=&
T_{ \hat{v}_{\pm} (L_{\pm}-1)-E,\, \hat{t}_{\pm}(L_{\pm}-1) }
T_{ \hat{v}_{\pm} (L_{\pm}-2)-E,\, \hat{t}_{\pm}(L_{\pm}-2) }
\cdots
T_{ \hat{v}_{\pm} (0)-E, \,\hat{t}_{\pm}(0) }  .
\eea
We will also need transfer matrices that map blocks of polymers to other blocks. The block transfer matrix from 
the $m^{th}$-block to the $k^{th}$-block is given by: 
\bea
T_{\omega}^E (k, m)
&:=&
T^E_{\omega_{k-1}}
T^E_{\omega_{k-2}}
\cdots
T^E_{\omega_{m}}, 
\quad
k > m , 
\\
T^E_{\omega} (k,m)
&:=&
T^E_{\omega} (m, k)^{-1},
\quad
k < m.
\eea

We recall the definition of a critical energy in Definition \ref{defn:critical_energy1}. The 
first condition says that 
$E_c$
lies in the interior of the spectrum
$\sigma (H_{\pm})$
of the corresponding periodic \Schr operators constructed from a single polymer $L_+$ or $L_-$. 

By definition, the $L_\pm$-polymer transfer matrices $T^{E_c}_{\pm}$
can be diagonalized simultaneously. Since $| tr (T^{E_c}_{\pm}) | < 2$, 
there exist  $\eta_{\pm} \in {\mathbb{R}}$, and an invertible matrix $M$, 
so that
\bel{eq:diagonalize1}
M T_{\pm}^{E_c} M^{-1}
= \left(\begin{array}{cc}
\cos \eta_{\pm} & - \sin \eta_{\pm} \\
\sin \eta_{\pm} & \cos \eta_{\pm}
\end{array}
\right).
\ee
In other words, the phases 
$e^{ i \eta_{\pm} }$ and $e^{- i \eta_{\pm} }$ 
are the eigenvalues of $T_{\pm}^{E_c}$,
respectively.
We prove that the LES $\xi_L$ converges to clock and that the unfolded eigenvalues satisfy clock behavior. This is the content of the following theorem based on two assumptions: 1) positivity of the IDS at a critical energy, and 2) an assumption on the behavior of the eigenvalues of  $T_{\pm}^{E_c}$. 

%We will then establish the assumptions in section \ref{subsec:verify_clock1}. 

\medskip

\begin{theorem}\label{thm:clock_1}
We assume that: 
%the eigenvalues $e^{i \eta_\pm}, e^{-i \eta_\pm}$ of the polymer transfer matrices $T_\pm^{E_c}$, at a critical energy $E_c$, satisfy
\begin{enumerate}
\item The eigenvalues $e^{i \eta_\pm}, e^{-i \eta_\pm}$ of the polymer transfer matrices $T_\pm^{E_c}$, at a critical energy $E_c$, satisfy the condition: 
\bel{eq:irr1}  
| \langle e^{ i k\eta_{\pm} } \rangle | < 1, 
\;
\mbox{ for all }
k \in {\Z \backslash \{0\}} ; 
\ee

\medskip

\item  The integrated density of states at a critical energy is strictly positive:
\beq
%&& \quad
0 < {{N}}(E_c) < 1.
\eeq
\end{enumerate}
Then, the LES $\xi_L$ at 
$E_c$, defined in \eqref{eq:les_process1}, satisfies
\begin{enumerate}
\item $\xi_L 
\stackrel{d}{\to}
clock (unif ( [0, 1) )$, where $unif (I)$ denotes the uniform distribution on the interval $I \subset \R$.
\\
\medskip

\item The rescaled eigenvalues, centered around $E_c$,  exhibit strong clock behavior : for any $j \in {\Z}$, we have 
$$
n(E_c) L \left(  E'_{j+1}(L) - E'_j (L)  \right)  \to 1, 
\;
a.s. 
$$
\end{enumerate}

\end{theorem}

\medskip

We make some remarks about the two assumptions of Theorem \ref{thm:clock_1}. 

\begin{remark}\label{rmk:ev_cond1} 

\begin{enumerate}
\item Let  $\triangle \eta := \eta_+ - \eta_-$. 
To understand condition \eqref{eq:irr1}, we compute
\bea\label{eq:phase2}
| \langle  e^{i k\eta_{\pm}}   \rangle   |^2
& =  &
 | p e^{i k \eta_+} + (1 - p) e^{i k\eta_-}  |^2
= 1 +   2p (1 - p) 
\left(
-1 +
\cos (k \triangle \eta)
\right).
\eea
Thus, the assumption \eqref{eq:irr1} is satisfied if for any $k \in \Z \backslash \{0\}$, we have $k \triangle \eta \not\equiv 0 \,  ( \rm{mod} ~ 2 \pi )$. For the dimer model, $L_\pm =2$, direct calculation shows that this condition is satisfied for certain $0 < V < 1$. 

\medskip

\item 
The positivity assumption on the IDS
may be replaced by a much weaker one that the number of eigenvalues in 
$(-\infty,  E_c)$, $(E_c, \infty)$ goes to infinity as $L \to \infty$, 
of which 
the positivity of IDS is a sufficient condition.
But 
the latter one can explicitly be examined for AB RPM.
In fact, by \cite[(14)]{jss}, the IDS at a critical energy satisfies, 
\bea
{N}(E_c) &=& \frac {\langle  L_{\pm} {N}_{\pm} (E_c)
\rangle  }{ \langle  L_{\pm}  \rangle  }
=  \frac { \langle  \eta_{\pm} / \pi  \rangle  }{ \langle L_{\pm} \rangle  }.
\eea
The dimer model, $L_\pm =2$, with $V \leq 1$, satisfies $\eta_+ = \pi$ and $\eta_- = f(V)$, with $f(V)$ explicitly computable, so that the condition is satisfied: $0 < {N} (E_c)  = \frac{1}{2} < 1$.\\

\end{enumerate}

\end{remark}

In the next subsections, we present the main technical tools necessary for the proof of Theorem \ref{thm:clock_1}. 
%The proof of Theorem \ref{thm:clock1} is presented in 
%subsection \ref{subsec:}-\ref{subsec:}. 

%%%%%%%%%%%%%%%%%%%%%%%%%%%%%%%%%%%%%%%

\medskip

%%%%%%%%%%%%%%%%%%%%%%%%%%%%%%%%%%%%
\subsection{Modified Pr\"ufer variables}\label{subsec:prufer1}

It will be convenient to use modified Pr\"ufer variables, as in \cite{jss}, to describe the LES. 
We consider the solution 
$u$
to the equation 
$Hu = Eu$
with initial condition
\beq
\left(
\begin{array}{c}
t(0) u(0) \\ u(-1)
\end{array}
\right)
=
\left(
\begin{array}{c}
\cos \theta_0 \\ \sin \theta_0
\end{array}
\right), 
\quad
\theta_0 \in [0, 2 \pi). 
\eeq
The \emph{modified Pr\"ufer variables} $(R_n (E), \theta_n (E))$ are defined by 
\bel{eq:prufer_defn1}
R_n (E)
e_{\theta_n (E)}
: =
M \left(
\begin{array}{c}
t(n) u(n) \\ u(n-1)
\end{array}
\right), 
\quad
e_{\theta}
:=
\left(
\begin{array}{c}
\cos \theta \\ \sin \theta
\end{array}
\right)  ,
\ee 
where $M$ is the matrix transforming the  polymer transfer matrices at a critical energy to rotations as in \eqref{eq:diagonalize1}.

We define functions $r:\R \to \R^+$, and $m:\R \to \R$, by
\bel{eq:m_fncs1}
Me_\theta = r(\theta) e_{m(\theta)} . 
\ee
 Since $e_{\theta + \pi} =   - e_\theta$, we can take $m(\theta+ \pi) = m ( \theta) + \pi$. As $r(\theta + \pi) = \| M e_{\theta + \pi}\|$, it follows that  $r(\theta+ \pi) = r(\theta)$.  
Finally, these functions are smooth in $\theta$ since $\theta \mapsto e_\theta$ is smooth.

For any $x \in \R$, we define the fractional part $(x)_{\pi {\Z}}$, and the integer part $ [x]_{\pi {\Z}}$, of $x$ modulo $\pi \Z$, by 
\[ 
(x)_{\pi {\Z}} := x - [x]_{\pi {\Z}}, 
\quad
[x]_{\pi {\Z}} 
:= \max \left\{ y \in \pi   {\Z}  \, \middle| \,  y  \le x  \right\}  .
\]

For $E \in \R$, let $\phi(E, L) := (\theta_L(E))_{\pi {\Z}}  \in [0, \pi)$, be the fractional part, and $m(E, L) \pi := [\theta_L (E)]_{\pi {\Z}} \in \pi \Z$, be the integer part, of the Pr\"ufer phase $\theta_L (E)$, modulo $\pi \Z$, respectively, so that 
\bel{eq:decomp1}
\theta_L (E)   =   m(E, L) \pi  + \phi(E, L) .
\ee
%$m(E_c, L)\in {\bf Z}$
%modulo
%$\pi ({\bf Z}+\frac 12)$. 
%These are defined by
%\bea
%\phi(E, L) &:=& (\theta_L (E))_{\pi ({\Z}+\frac 12)}
%\in [0, \pi)  \nonumber \\
%\theta_L (E_c)  & : =&  m(E_c, L) \pi + \phi(E_c, L)  ,
%\eea
%where
%\[ 
%(x)_{\pi ({\Z}+\frac 12)} := x - [x]_{\pi ({\Z}+\frac 12)}, 
%\quad
%[x]_{\pi ({\Z}+\frac 12)} 
%:= \max \left\{ y \in \pi   {\Z}  \, \middle| \,  y + \frac {\pi}{2} \le x  \right\}  .
%\]
%%
Following \cite{kn2}, 
we also define a \emph{relative Pr\"ufer angle,} divided by $\pi$, by 
\bel{eq:rel_prufer2}
\Psi_L (x) := \frac {1}{\pi}
\left\{\theta_L \left(
E_c+\frac {x}{ n(E_c) L}
\right)- \theta_L(E_c)
\right\}, 
\quad
x \in {\R}.
\ee
%
%be the 
%``relative Prufer angle"
%divided by 
%$\pi$.\\

%%%%%%%%%%%%%%%%%%%%%%%%%%%%%%%%%%%%%%

\subsection{The characteristic function of $\xi_L$}\label{subsec:charact1} 

% Proof of Theorem \ref{thm:clock_1}}

%We prove Theorem \ref{thm:clock_1} under the following assumptions. We will then %establish the assumptions in section \ref{subsec:verify_clock1}. 

We first obtain a representation of the characteristic function of $\xi_L$ in terms of the Prufer phase. 

\medskip

%\begin{itembox}[l]
\begin{lemma}\label{lemma:clock2}
The characteristic function of the process $\xi_L$ \eqref{eq:les_process1}, associated with the finite-length RPM \Schr operator $H_L$, satisfies
\bel{eq:clock_L1}
{\E}[e^{- \xi_L(f)}]
=
{\E}\left[
\exp \left\{
-\sum_{n =1 - m(E_c, L)}^{L - m(E_c, L)}
f
\left( 
\Psi_L^{-1} 
\left(n  - 
\frac {\phi(E_c, L) - m ( \frac{\pi}{2})}{\pi} 
\right)
\right)
\right\}
\right], 
\quad
f \in C_c({\R})
\ee
where
$\Psi_L^{-1}$
is the inverse of the function
$x \mapsto \Psi_L(x)$ defined in \eqref{eq:rel_prufer2}, and $m(\theta)$ is defined in 
\eqref{eq:m_fncs1}.
\end{lemma} 
%\end{itembox}
%
%
\begin{proof}
Let 
$x_n(L)$
be the 
$n^{th}$ atom of 
$\xi_L$ : 
\[
x_n (L) :=
n(E_c) L ( E_n(L) - E_c ), 
\quad
n\ge 1.
\]
Then, 
by Sturm oscillation theory, 
the modified Pr\"ufer angle at an eigenvalue satisfies 
%by a property of the Prufer angle, 
%
\begin{eqnarray*}
\theta_L(E_n(L))
&=&
\theta_L
\left(
E_c+\frac {x_n (L)}{n(E_c) L}
\right) 
= 
m \left( n \pi + \frac {\pi}{2} \right)
=
 n \pi  + m \left(  \frac{\pi}{2} \right), 
\end{eqnarray*}%
and at a critical energy, the modified Pr\"ufer phase satisfies
\begin{eqnarray*}
\theta_L(E_c) 
&=& 
 m(E_c, L) \pi + \phi(E_c, L).
\end{eqnarray*}
We refer to \eqref{eq:free_mod1} for the relation between the free Prufer angle and the modified Prufer angle. 
From the definition \eqref{eq:rel_prufer2},  we then have  
\begin{eqnarray*}
\Psi_L (x_n(L))
&=&
%\frac{1}{\pi}
%\left\{
%\theta_L
%\left( E_c+\frac {x}{n(E_c)L} \right)
%- \theta_L(E_c) \right\} 
%=
n-m(E_c, L) +  \frac{ m(\frac{\pi}{2}) - \phi(E_c,L)}{\pi}  .
\end{eqnarray*}
%Fro the definition of $m(\theta)$, we have YYY.
Since $x \mapsto \Psi_L(x)$ is an increasing function (see Appendix \ref{app:incr_psi1}), we obtain 
\bea\label{eq:modified1}
\{ x_n (L) \}_{n =1}^{L}
&=  & \left\{ \Psi_L^{-1} 
\left( n-m(E_c, L) - \frac {\phi(E_c,L) -  m(\frac{\pi}{2})}{\pi}
\right) \right\}_{ n=1 }^{ L }   \nonumber \\
&=  & \left\{	 \Psi_L^{-1} 
\left( n- \frac {\phi(E_c,L) -  m(\frac{\pi}{2})}{\pi} \right) 
\right\}_{n =1 - m(E_c, L)}^{L - m(E_c, L)} .
\eea
As $\xi_L (f) = \sum_{n=1}^L f( x_n(L))$, formula \eqref{eq:modified1} establishes the result \eqref{eq:clock_L1}. 
\end{proof}

\medskip
%
%\end{document}
%
The following lemma from \cite[Lemma 3.1]{kn1} is necessary in order to consider the convergence of the characteristic function. % It appears in \cite[Prop XX]{YY}. 
%%%%%
\medskip

%\begin{itembox}[l]
\begin{lemma}\label{lemma:clock_pp2}
Let $\{ \Psi_L(x)\}_L$
be a sequence of non-decreasing functions satisfying 
$\Psi_L(x) \to \Psi(x)$
pointwise, and assume that 
$\Psi$ is an increasing function. 
If $x_L \to x$, then
\[
\Psi_L^{-1}(x_L) 
\stackrel{L \to \infty}{\to}
\Psi^{-1}(x).
\]
\end{lemma}
%
%\end{itembox}
%%%%%
%
\begin{proof}
If the statement is false, we can find 
$\delta > 0$, a point $x$, and a subsequence 
$L_n$ such that 
\[
|\Psi_{L_n}^{-1}(x_{L_n}) - \Psi^{-1}(x) | > \delta.
\]
Taking a subsequence, if necessary, we can assume that 
\[
\Psi_{L_n}^{-1}(x_{L_n}) <   \Psi^{-1}(x) - \delta.
\]
Since 
$\Psi_{L_n}$
is non-decreasing, 
\[
x_{L_n}
\le
\Psi_{L_n} \left( \Psi^{-1}(x) - \delta \right)
\]
Letting 
$n \to \infty$, 
we have
\[
x \le 
\Psi \left(
\Psi^{-1}(x) - \delta 
\right)
<
\Psi ( \Psi^{-1}(x)) = x
\]
leading to a contradiction. 
%
%\QED
\end{proof}
%%%%%%%%%%%%%%%%%%%%%%%%%%%%%%%%%%%%%%%%%%%%%%%%%%%%%%%%%%%%%%%%%%%%%%%%%%%%%%

\subsection{First convergence result for $\xi_L$} %Characteristic function to a clock process}\label{subsec:charac_limit1}

\medskip

Given the form of the characteristic function $\xi_L$ in \eqref{eq:clock_L1}, we compute the limit as $L \to \infty$ and obtain a first result on the clock process. We make three assumptions whose validity is discussed in Remarks \ref{rmk:clock_assump1}.

 \medskip
 
 \begin{assumption}\label{assumpt:prufer1}
%Assumption 1
\begin{enumerate}
\item The function $\Psi_L$ satisfies $\lim_{L \to \infty}
\Psi_L (x) =  x$, 
a.s.
\medskip

\item There exists a random variable 
$\phi_c \in [0, \pi)$ such that the fractional part $\phi(E_c, L)$ of the Pr\"ufer phase
$\theta_L (E_c)$
converges \emph{in distribution} to $\phi_c$.

\medskip

\item The integer part $m(E_c, L)$ of the Pr\"ufer phase 
$\theta_L (E_c)$
at a critical energy satisfies 
$$
\lim_{L \to \infty}  m(E_c, L)  =\infty,  
$$ 
and 
$$
\lim_{L \to \infty}   \left( L - m(E_c, L)   \right)    = \infty
$$
almost surely. 
\end{enumerate}
\end{assumption}

\begin{remarks}\label{rmk:clock_assump1} 
\noindent
\begin{enumerate}
\item Part 1, 2 is proven for RPM in section \ref{subsec:verify_clock1}, 
under some conditions on 
$\eta_{\pm}$.

%\item Part 2 holds by compactness for subsequences. 
%We use a weaker version guaranteeing almost sure convergence in Proposition \ref{prop:clock1} but the stronger version of convergence in distribution is necessary for the characterization of the limit distribution as uniform in Proposition  \ref{prop:clock_unif1}.  

\item Part 3 follows if the number of eigenvalues of $H_\omega^L$ in 
$(-\infty,  E_c)$ and in $(E_c, \infty)$ goes to infinity as $L \to \infty$. As discussed in Remark \ref{rmk:ev_cond1}, if the IDS at the critical energy satisfies
$0 < N(E_c) < 1$, 
then part 3 follows.

\item
As can be seen later in the proof, 
the most important condition in Assumption \ref{assumpt:prufer1} is part (1).
In fact, 
without part (2), 
we can show that the accumulation point of 
$\xi_L$
is a clock process with some 
$\mu$
(Corollary \ref{rmk:clock_assump2}), 
and wituout part (3), the sum
$\sum_{ n \in {\bf Z}}$
in the definition of the clock process may be replaced by 
$\sum_{n \in E}$
for some subset 
$E (\subset \R)$.
%

%the IDS at a critical energy is between $0$ and $1$ so must increase to $1$ as $E \to  \infty$. This means there are an infinite number of eigenvalues in this region.  
%This follows from the positivity of the IDS ${N}$ at the critical energy. 
%In \cite[eqn.\ 14]{jss}, they prove that 
%\bea
%{N}(E_c)  &=&  \frac {\langle 
%L_{\pm} {{N}}_{\pm} (E_c)
%\rangle}{\langle  L_{\pm}  \rangle  }
%= \frac {  \langle  \eta_{\pm} / \pi \rangle }{ \langle L_{\pm} \rangle  }
%\quad
%\left(   {{N}}_{\pm} (E_c)  =  \frac {  \eta_{\pm} }{  \pi L_{\pm} }  \right)
%\eea
%%
%so that there should be some examples such that 
%$0 < {{N}} (E_c) < 1$.\\
%%
\end{enumerate}
\end{remarks}

We first present an intermediate result asserting the convergence of the characteristic function to a clock process with probability measure $\mu$. 
%
%%%%%
%\begin{itembox}[l]
\begin{theorem} (Convergence to a clock process)\label{thm:clock1}
Under Assumption \ref{assumpt:prufer1}, the LES  $\xi_L$ centered at a critical energy satisfies $\xi_L \stackrel{d}{\to}
clock(\mu_{frac})$, where $\mu_{frac}$ is the distribution of the fractional part of the random variable 
$\frac{1}{\pi}\left( \phi_c - m(\frac{\pi}{2})\right)$.
\end{theorem}
%%%%%
%
\begin{proof}
%
%By part (2) of Assumption 1, 
%the fractional part of the Prufer phase 
By Skorohod's Theorem \cite{bill1}, we may assume that  
$\phi(E_c, L) \to \phi_c$, a.s.
We write
$
\frac {1}{\pi}
\left(
\phi_c - m \left( \frac {\pi}{2} \right)
\right)
=
\left[
\frac {1}{\pi}
\left(
\phi_c - m \left( \frac {\pi}{2} \right)
\right)
\right]
+
\left\{
\frac {1}{\pi}
\left(
\phi_c - m \left( \frac {\pi}{2} \right)
\right)
\right\}
$, 
where 
$\left[
\frac {1}{\pi}
\left(
\phi_c - m \left( \frac {\pi}{2} \right)
\right)
\right]
\in {\bf Z}$
and
$
\left\{
\frac {1}{\pi}
\left(
\phi_c - m \left( \frac {\pi}{2} \right)
\right)
\right\}
\in 
[0 ,1)$
are the integer and fractional part of 
$
\frac {1}{\pi}
\left(
\phi_c - m \left( \frac {\pi}{2} \right)
\right)
$
respectively.
We denote by 
$\mu_{frac}$
the distribution of this fractional part.
Lemma \ref{lemma:clock2} provides a representation of the expected value of  
$e^{-\xi_L (f)}$, for any $f \in C_0(\R)$:
\bel{eq:clock5}
{\E}[ e^{- \xi_L (f)} ]
= {\E}\left[ \exp \left(-\sum_{n=1 - m(E_c, L)}^{L - m(E_c, L)} 
f \left (\Psi_L^{-1} \left(n  - \frac{\phi(E_c, L) - m(\frac{\pi}{2})}{\pi} \right)
\right) \right) \right]   .
\ee
By Assumption \ref{assumpt:prufer1} and Lemma \ref{lemma:clock_pp2}, we find that the limit $L \rightarrow \infty$ of the right side of \eqref{eq:clock5} is
\bea\label{eq:clock4}
%\begin{eqnarray*}
{\E}[ e^{- \xi(f)} ]
%&\stackrel{lem.1.1}{=}&
%= {\E}\left[
%\exp \left(
%- \sum_{n \in {\Z}} f( n - \frac{\phi(E_c,L)}{\pi} )
%\right)
%\right]
 & = &
%\int_0^1 d \mu_{frac} (\theta (E_c))
\E \left[ \exp \left[
- \sum_{n \in {\Z}} f \left( n  -  \left[
\frac {1}{\pi}
\left(
\phi - m \left( \frac {\pi}{2} \right)
\right)
\right]
-
\left\{
\frac {1}{\pi}
\left(
\phi - m \left( \frac {\pi}{2} \right)
\right)
\right\} \right) 
%\theta (E_c ) - \frac{1}{\pi} m(\frac{\pi}{2}) \right)
\right] \right] \nonumber \\  
 & = & \int_0^1 d \mu_{frac}(\theta)
\exp \left[ - \sum_{n \in {\Z}} f \left( n  - \theta \right)
\right] , 
\eea
where we write $\theta := \frac{1}{\pi} ( \phi_c - m(\frac{\pi}{2}) ).$
%where we set $\widetilde{\theta}(E_c) := \theta(E_c) - \frac{1}{\pi} m (\frac{\pi}{2})$. 
%establishing the theorem. 
\end{proof}

\medskip

The proof of 
Theorem \ref{thm:clock1}
implies the following corollary, stating that any accumulation point of 
$\xi_L$
is always a clock process for some 
$\mu$.

\begin{follow}\label{rmk:clock_assump2} 
If we only assume Assumption \ref{assumpt:prufer1}, parts 1 and 3,  then any accumulation point  $\xi$  of  $\xi_L$
is a clock process with some distribution  
$\mu$
on 
$[0, 1)$.
\end{follow}
\begin{proof}
For $f \in C_c ({\bf R})$, Lemma \ref{lemma:clock2} yields
\bea\label{eq:star}
\E[ e^{- \xi_L(f)} ]
&=&
\E
\left[
\exp
\left\{
-
\sum_{n \in {\bf Z}}
f(n + \theta_L)
\right\}
\right].
%\quad\cdots (*)
\eea
for sufficiently large
$L$, 
where we set
\beq
\theta_L
&:=&
\left\{
\frac {1}{\pi}
\left(
m \left(
\frac {\pi}{2}
\right)
%}
-\phi(E_c, L)
\right)
\right\}.
\eeq
Furthermore, if $f \in C_c ([0, 1))$, since $\theta_L \in (0, 1]$, this yields
\begin{equation}
\E[ e^{- \xi_L(f)} ]
=
\E
\left[
\exp
\left\{
-
f(\theta_L)
\right\}
\right].
\label{theta}
\end{equation}
Let 
$\xi$
be an accumulation point of 
$\xi_L$.
Then, we can find a subsequence 
$\{ L_k \}_k$
such that 
$\xi_{L_k} \stackrel{d}{\to} \xi$.
Thus, as $k \to \infty$, the limit
%$\lim_{k \to \infty}$
of the LHS of (\ref{theta}) exists and 
\beq
\E[ e^{- \xi(f)} ]
&=&
\lim_{k \to \infty}
\E
\left[
\exp
\left\{
-
f(\theta_{L_k})
\right\}
\right]
\eeq
which implies that 
$\{ \theta_{L_k} \}$
has a limit in distribution: $\lim_{k \to \infty}
\theta_{L_k} \stackrel{d}{=:} \theta$, with distribution $\mu$ supported in $[0,1]$.  
%\stackrel{d}{:=}
%\exists
%\lim_{k \to \infty}
%\theta_{L_k}$.
%
Taking 
$L_k \to \infty$
in \eqref{eq:star}
%$(*)$
now yields
\beq
\E[ e^{- \xi(f)} ]
&=&
\int_0^1 ~d \mu (\theta) ~ \left[
\exp
\left\{
-
\sum_{n \in {\bf Z}}
f(n + \theta)
\right\}
\right].
%\quad\cdots (*)
\eeq
Therefore, the limiting process 
$\xi$  is clock $(\mu)$.  
%where 
%$\mu$
%is the distribution of 
%$\theta$.
%
%
\end{proof}

\medskip

In the next section, we prove part 2 of Assumption \ref{assumpt:prufer1} in Proposition \ref{prop:clock_unif1}. This will establish the first part of Theorem \ref{thm:clock1} on the uniform nature of the distribution of the limiting random variable $\phi_c$. .

%
%\begin{remark}\label{rmk:ev_cond1} Let  $\triangle \eta := \eta_+ - \eta_-$. 
%To understand condition \eqref{eq:phase1}, we compute
%\bea\label{eq:phase2}
%| \langle  e^{i k\eta_{\pm}}   \rangle   |^2
%& =  &
% | p e^{i \eta_+} + (1 - p) e^{i \eta_-}  |^2
%= p^2 + (1- p)^2   +  2 p (1 - p)
%\cos (k \triangle \eta)   \nonumber \\ 
% &=&   1 +   2p (1 - p) 
%\left(
%-1 +
%\cos (k \triangle \eta)
%\right)   .
%\eea
%%
%Thus, the assumption \eqref{eq:phase2} is satisfied if for any $k \in \Z$,
%$k \triangle \eta \not\equiv 0 \,  ( \rm{mod} ~ 2 \pi )$.
%\end{remark}

%%%%%%%%%%%%%%%%%%%%%%%%%%%%%%%%%%%%%%%%%%%%%%%%%%%%%%%%%%%%%%%%%%%%%%%%%%%
\subsection{Strong clock behavior}\label{subsec:strong_clock1}

In the clock process problem, 
we need to study the fractional part of 
$\frac{1}{\pi} \theta_L (E_c)$ which might not reflect the true nature of the eigenvalue statistics. 
%which could be artificial. 
%
To further clarify the nature of the LES,  we will consider the spacing between consecutive eigenvalues.
We let $\{ E'_j (L) \}$  denote the re-indexed eigenvalues of $H_L$ centered at a critical energy $E_c$: 
\bel{eq:reindex_ev1} 
 \cdots < E'_{-2}(L) < E'_{-1}(L) < E_c \le E'_0(L) < E'_1(L) < \cdots
\ee
To show 
the strong clock behavior, we need slightly stronger condition than that for the convergence to the clock process : 
\begin{assumption}\label{assumpt:prufer2}
\begin{enumerate}
\item
$\lim_{L \to \infty}\Psi_L (x) = x$, a.s., 
locally uniformly.

\item
$\lim_{L \to \infty}
m(E_c, L)
=\infty$, 
$\lim_{L \to \infty}
\left(
L - m(E_c, L)
\right)
=\infty$.
\end{enumerate}
\end{assumption}

%\medskip

\begin{theorem}\label{thm:spacing1}(Clock behavior: eigenvalue spacing)
Under 
Assumption \ref{assumpt:prufer2}, the local eigenvalue spacing around a critical energy almost surely satisfies  
\bel{eq:ev_space1} 
\lim_{L \rightarrow \infty} n(E_c) L   \left(   E'_{j+1}(L) - E'_j (L)  \right)  = 1, 
\ee
for any 
$j \in \Z$. 
This establishes part 2 of Theorem \ref{thm:clock_1}.
\end{theorem}
%%%%%
%
\begin{proof}
By part (2) of Assumption 
\ref{assumpt:prufer2}, 
for any
$j \in {\bf Z}$, 
$E'_j (L)$
is well defined for sufficiently large 
$L$.
Set  
$x'_j (L) := {n(E_c) L}( E'_j (L) - E_c)$.
Then  
\beq
n(E_c) L 
\left(
E'_{j+1}(L) - E'_j (L)
\right)
&=&
x'_{j+1}(L) - x'_j(L)
\eeq
so it suffices to show that RHS of the equation above  converges to 
$1$. 
On the other hand, since 
$E'_j(L), E'_{j+1}(L)$
are consecutive eigenvalues, we can find 
$n_j(L) \in {\bf Z}$
such that
\beq
&&
\theta_L
\left(
E_c + \dfrac {x'_j(L)}{\rho(E_c) L}
\right)
=
n_j(L) \pi + 
%\textcolor{red}{
m \left(
\frac {\pi}{2}
\right)
%}
\\
&&
\theta_L
\left(
E_c + \dfrac {x'_{j+1}(L)}{\rho(E_c) L}
\right)
=
\left(
n_{j}(L) + 1 
\right) \pi
+ 
%\textcolor{red}{
m \left(
\frac {\pi}{2}
\right).
%}
\eeq
Therefore
\beq
%\stackrel{def.}{\leadsto}
%\quad
&&
\Psi_L(x'_{j+1}(L)) - \Psi_L(x'_j(L))
= 1.
\eeq
By 
Lemma \ref{lemma:x_bdd1} below, 
%part (1) of Assumption \ref{assumpt:prufer2}, 
we can find an interval
$I = [a,b]$
such that
$x'_j (L), x'_{j+1}(L) \in I$
for sufficiently large
$L$.
Then
\beq
&&
x'_{j+1}(L) - x'_j(L)  - 1
\\
&=&
x'_{j+1}(L) - \Psi_L (x'_{j+1}(L))
+
\Psi_L (x'_{j+1}(L))
-
\Psi_L (x'_{j}(L)) - 1
+
\Psi_L (x'_{j}(L))
-
x'_{j}(L)
\\
&=&
x'_{j+1}(L) - \Psi_L (x'_{j+1}(L))
+
\Psi_L (x'_{j}(L))
-
x'_{j}(L)
\eeq
Since 
$\Psi_L (x) 
\stackrel{L \to \infty}{\to} x$
uniformly on 
$I$, 
we are done. 
\end{proof}
%

%%%%%%%%%%%%%%%%%%%%%%%%%%%%%%%%
\begin{lemma}\label{lemma:x_bdd1}
For any fixed $j$, the rescaled eigenvalues $x'_j (L) := {n(E_c) L}( E'_j (L) - E_c)$ are bounded with respect to 
$L$:
\beq
\sup_{L \ge 1}
| x_j'(L) | 
<
\infty .
%\quad
%a.s.
\eeq
\end{lemma}

\begin{proof}
The argument given here is deterministic.  We compute the modified Pr\"ufer phase.
By Sturm oscillation theory, 
\beq
\theta_L \left( E_c + \frac { x'_j (L) }{n(E_c) L}  \right)
&=& m(E_c, L) \pi  +  m \left( \frac {\pi}{2} \right)
+   j \pi  
\eeq
On the other hand, 
by the definition of 
$\Psi_L$, 
\beq
\theta_L
\left(
E_c + \frac { x'_j (L) }{\rho L}
\right)
&=&
\theta_L (E_c)
+
\pi
\Psi_L (x'_j(L))    \\
  & =  & 
m(E_c, L) \pi + \phi (E_c, L)
+
\pi
\Psi_L (x'_j(L))  .
\eeq
These two equations yield the expression
\\
%
%\leadsto
%\quad
\beq
\Psi_L (x'_j(L))
&=& j  +  \frac{1}{\pi} \left( m\left( \frac {\pi}{2} \right)  -  \phi (E_c, L) \right). 
\eeq
Since
$\phi (E_c, L) \in [0, \pi)$, 
we have
\bel{eq:xj_bd1} 
\sup_L
\Psi_L (x'_j (L)) 
\le
C_j, 
\quad
C_j 
:=
j + \frac{1}{\pi} m \left( \frac {\pi}{2} \right) . 
%\quad\cdots (0)
\ee
Suppose that $\lim_{L \to \infty} x'_j (L) = \infty$.
For any $\epsilon > 0$, we have 
%
%Then 
%
\bel{eq:xj_bd2}
C_j + 2 \epsilon  \le  x'_j (L), 
\quad
%L \gg 1  .
%\quad
%\quad\cdots (1)
\ee
for sufficietly large
$L$.
On the other hand, since
$\Psi_L (x) \to x$ 
pointwise, 
%locally uniformly, 
%
\bel{eq:xj_bd3}
\Psi_L (C_j) 
\le 
C_j + \epsilon
\le
\Psi_L (C_j + 2 \epsilon), 
\quad
%L \gg 1  . 
%\quad\cdots (2)
\ee
for sufficietly large
$L$.
Therefore, 
by \eqref{eq:xj_bd1}-\eqref{eq:xj_bd2}, 
%, (1), (2), 
and the fact that $x \mapsto \Psi_L (x)$
is increasing by Proposition \ref{prop:psi_incr1},   we have
\bel{eq:xj_bd4}
C_j + \epsilon \le  \Psi_L (C_j + 2 \epsilon)
\le  \Psi_L (x'_j(L))
%\quad\cdots (3)
\ee
%
%(3)
which contradicts to \eqref{eq:xj_bd1}.
\end{proof}

\medskip
%
%%%%%%%%%%%%%%%%%%%%%%%%%%%%%%%%%%%%%%%%
%%%%%%%%%%%%%%%%%%%%%%%%%%%%%%%%

\subsection{Verification of Assumptions}\label{subsec:verify_clock1}

In this section, 
we verify Assumptions \ref{assumpt:prufer1}
and
\ref{assumpt:prufer2}, 
that is, 
$\lim_{L \to \infty}\Psi_L(x) = x$
and the convergence of 
$\phi(E_c, L)$
in distribution. 
For this, we use the notation and approach introduced in \cite[section 4.2]{jss}.
The authors in 
\cite{jss}
study the transfer matrices at energies 
$E_c + \epsilon$
near the critical energy
$E_c$.
Recalling 
that the matrix
$M \in SL(2, \R)$
maps 
$T_{\pm}^{E_c}$
to rotations, we define the related transmission 
$a^{\epsilon}_{\pm}$
and reflection
$b^{\epsilon}_{\pm}$
coefficients by 
\beq
M T_{\pm}^{E_c + \epsilon} M^{-1} v
&=&
a^{\epsilon}_{\pm} v 
+
b^{\epsilon}_{\pm} \overline{v}, 
\;
\text{ where }
\;
v
:=
\frac {1}{
\sqrt{2}
}
\left(
\begin{array}{c}
1 \\ -i
\end{array}
\right).
\eeq
In order to study the behavior of the Pr\"ufer variables at energies $E_c + \epsilon$, 
we define a general polymer phase shift $S_{\epsilon, \pm} (\theta)$ and amplitude $\rho_{\pm}^{\epsilon} (\theta)$, by
\bea\label{eq:modified_prufer1}
\rho_{\pm}^{\epsilon} ( \theta )
e_{ S_{\epsilon, \pm} (\theta) }
&:=&
M T_{\pm}^{E_c + \epsilon} M^{-1} e_{\theta}  ,
%\quad\cdots (38)
\eea
where the matrix $M$ is defined in \eqref{eq:diagonalize1} and vector $e_\theta$ is defined in \eqref{eq:prufer_defn1}.
Then, the iterated polymer phase shift satisfies
\bea
S^{\ell + 1}_{\epsilon, \omega} ( \theta )
&:=&
S_{\epsilon, \omega_{\ell} }
\left( 
S^{\ell}_{\epsilon, \omega} (\theta) 
\right), 
\quad  \mbox{and } ~~~
S^0_{\epsilon, \omega} (\theta) = \theta.
\eea
By equation (\ref{eq:diagonalize1}),  we have, for any 
$\theta$, 
$\rho^0_{\pm}(\theta) = 1$, 
$\eta_{\pm} = S_{0, \pm}(\theta) - \theta$
up to a multiple of 
$2 \pi$, 
which is hereby fixed.
Setting $\epsilon = 0$ in (\ref{eq:modified_prufer1}), we obtain 
\beq
\rho^0_{\pm} (\theta)
e_{ S_{0, \pm}(\theta) }
&=&
M T_{\pm}^{E_c} M^{-1} e_{\theta}
=
R_{\eta_{\pm}} e_{\theta}
=
e_{\eta_{\pm} + \theta}
\eeq
yielding 
\beq
\rho^0_{\pm} (\theta) = 1, 
\quad
S_{0, \pm} (\theta) = \theta + \eta_{\pm} ,
\eeq
where $\eta_\pm$ are defined in \eqref{eq:diagonalize1}. 
Since
$\rho^0_{\pm} (\theta) = 1$, 
the Lyapunov exponent vanishes at 
$E_c$: $L (E_c) = 0$.
Another 
important formula in 
\cite{jss}
is 
\beq
S_{\epsilon, \pm} (\theta) - \theta
&=&
\eta_{\pm}
+
\epsilon
\cdot
d_{\pm}
- 
\epsilon
\cdot
Im 
\left[
c_{\pm} e^{2i \theta}
\right]
+
\mathcal{O}(\epsilon^2)
\eeq
where 
\bel{eq:coef2}
d_{\pm}
:=
\partial_{\epsilon}
\eta^{\epsilon}_{\pm}
|_{\epsilon=0}, 
\quad
c_{\pm}
:=
\partial_{\epsilon}
b^{\epsilon}_{\pm}
|_{\epsilon=0}
\cdot
e^{i \eta_{\pm}}.
\ee
We first verify the part 2 of Assuption \ref{assumpt:prufer1}: the normalized fractional part of the Pr\"ufer phase 
$\phi(E_c, L)/\pi$
converges \emph{in distribution} to a random variable %$\theta(E_c)$ 
with the uniform distribution on $[0, 1)$.

%the the limiting random variable 
%$\phi \in [0, \pi)$ such that the fractional part of the Prufer phase
%$\phi(E_c, L)$ converges \emph{in distribution} to a random variable $\phi$ in distribution, and $\phi/\pi$ has distribution $\mu$ on $[0, 1)$.
%%the distribution of 
%We now show that the limiting random variable $\theta (E_c)$ is uniformly distributed.

\medskip

\begin{proposition}\label{prop:clock_unif1}  %(Assumption 1, Part 1)
Suppose that  the eigenvalues $\eta_\pm$ of the polymer transfer matrices $T^{E_c}_\pm$  satisfy
\bel{eq:phase1}
|\langle e^{i k\eta_{\pm}}  \rangle|   < 1, 
\quad
\forall k \in {\Z \backslash \{0\}}.
\ee
Then,  the normalized fractional part of the Prufer phase at a critical energy satisfies 
\bel{eq:phase_limit1}
\frac{1}{\pi} \phi(E_c, L)
\stackrel{d}{\to}   \frac{1}{\pi} ( \phi_c - m(\frac{\pi}{2}))  \in 
unif [0, 1).
\ee

\end{proposition}

\medskip

%Here we use the following fact : 
%if a sequence 
%$\{ X_n \}$
%of the fractional part of real valued random variables 
%$X_n$
%converges to an uniform distribution on 
%$[0, 1)$
%then so does a shifted one 
%$\{ X_n + a \}$
%for any constant 
%$a \in \R$. 
%
%Thus 
%$
%\left\{
%\frac {1}{\pi}
%\left(
%\phi(E_c, L) - m \left( \frac {\pi}{2} \right)
%\right)
%\right\}
%$
%converges to an uniformly distributed on 
%$[0, 1)$.
%
%This establishes part 2 of Assumption \ref{assumpt:prufer1}.
%
%Theorem \ref{thm:clock_1}.

%
\begin{proof}  
It suffices to show $\phi(E_c, L) \stackrel{d}{\to} unif [0, \pi)$. This 
will follow from establishing that 
$$\lim_{L \to \infty}{\E}[ e^{2 i k \phi (E_c, L) }] =0, \forall  
k \in \Z \backslash \{0\}. $$  
 \noindent
{Case 1: } Suppose that $L$ is on a polymer node, that is, 
$L = \sum_{\ell=0}^{N-1} L_{\omega_{\ell}}$, 
for some $N$. 
Then 
$\theta_L (E_c)  = \sum_{ \ell = 0}^{N-1} 
\eta_{\omega_{\ell}}$.  
We note that $e^{2ik\theta_L(E_c)} = e^{2i k \phi(E_c,L)}$, for $k \in \Z$.  
It follows from the independence of the random variables $\omega_\ell$, that for any $k \in \Z \backslash \{0\}$,  
\bel{eq:vanish1}
{\E}[ \exp ( 2i k \theta_L (E_c) ) ]  =   {\E} \left[  \exp  \left( 2 ik
\sum_{ \ell = 0 }^{N-1}
 \eta_{\omega_{\ell}} \right)
\right]
=  \prod_{\ell = 0}^{N-1}   {\E}  \left[   \exp \left( 2  i k \eta_{\omega_{\ell}}
 \right)    \right]
=  \langle e^{2i k\eta_{\pm}}
\rangle^N  \to 0, 
\ee
by the hypothesis on the transfer matrix eigenvalues \eqref{eq:phase1}.  \\
\noindent
{Case 2:} General case : Let $L'$  be the maximum number of polymer nodes less than or equal to  $L$. Then, 
\bea\label{generalcase}
\theta_L (E_c)   &=&  \triangle \theta_L(E_c)  +  \theta_{L'} (E_c), 
\quad
\mbox{ where }
\triangle \theta_L (E_c)   :=  \theta_L (E_c)  -  \theta_{L'} (E_c).
\eea
Case 1 shows that ${\E}[ \exp ( 2 i k \theta_{L'} (E_c) ) ] \to 0$. 
Furthermore, the random variables $\triangle \theta$  and  $\theta_{L'}(E_c)$ 
are independent, so we conclude that $\theta_{L} (E_c) \stackrel{d}{\to} unif [ 0, \pi)$ in this case too.
%
%\QED
\end{proof}

When the system size is equal to 
$L$, a polymer mode, we let $N = N(L)$
be the number of the polymers in  $[0, L-1]$. We need a technical lemma from \cite{jss} in order to control the sum
$$
I_{\omega,k}^1 ( \theta, \epsilon) := \sum_{\ell = 0}^{k-1} c_{\omega_\ell} e^{ 2 i S_{\epsilon, \omega}^\ell (\theta) } .
$$ 
%%
%(We are assuming that 
%$L$
%is at the polymer mode).\\
%%

%
%
\begin{lemma}\label{lemma:sum1} 
Suppose that
$| \langle e^{ 2i \eta_{\pm} } \rangle | < 1$.
Let 
$c_\omega := (c_{\omega_\ell}) \in \C$ be the coefficient sequence defined in \eqref{eq:coef2}.  
For 
almost every $\omega$, there exists an integer $N_\omega$, so that for all $N > N_\omega$,  
the sum  $I_{\omega,N}^1 ( \theta, \epsilon)$ satisfies
\bel{eq:sum1}
 I_{\omega,N}^1 ( \theta, \epsilon) 
 %= \sum_{\ell = 0}^{N-1}
%c_{\omega_{\ell}} e^{2i S^{\ell}_{\epsilon, \omega} (\theta)}
= {\mathcal{O}}(N^{\frac 12 + \alpha}),
\ee
almost surely, for any $\alpha > 0$, and any 
$\epsilon$
such that 
$\epsilon \le \frac{C}{\sqrt{N}}$, for a universal constant $C>0$.  
%The coefficients $c_\pm$ are defined in \cite[p.\ 41]{jss}. 
Furthermore, we have
 \bel{eq:sum2}
\E [ I_{\omega,N}^1 ( \theta, \epsilon) ] =  \mathcal{O} (N \epsilon^2).
%= \mathcal{O}(1).
\ee
\end{lemma}

%\begin{remark}
%The definition of 
%$c_{\pm}$ 
%is in p.41 in \cite{jss}.
%%
%In Proposition 2 in \cite{jss}, 
%they show that 
%${\bf E}[ \mbox{ LHS in } (1) ] = 
%\mathcal{O} (N \epsilon) (= \mathcal{O}(1))$.
%\end{remark}

%
\begin{proof}
\eqref{eq:sum2}
is proved in \cite[Proposition 2]{jss}, so it suffices to show \eqref{eq:sum1}. 
For any 
$\alpha > 0$ 
and any 
$N \in \N$, 
we consider the following event: 
\bea
\Omega^0_N
(\alpha, \epsilon, \theta)
&:=&
\left\{
\omega\in \Omega_0
\middle|
\exists k \le N
\mbox{ such that }
| I^1_{\omega, k} (\theta, \epsilon) | 
\ge 
N^{\alpha + 1/2} 
%\epsilon = \frac{C}{\sqrt{N}}  .
\right\}.
\eea
Then 
by \cite[Theorem 6]{jss}, there exist constants $C_1, C_2 > 0$, such that for any $\theta$, any $N$, and any $\epsilon > 0$, with $N \epsilon^2 < 1$, we have 
\bea
{\Pp}_0
(
\Omega_N^0 (\alpha, \epsilon, \theta)
)
& \le &
C_1 
e^{ - C_2 N^{\alpha} }.
\eea
By the Borel-Cantelli Lemma, 
for a.e.
$\omega$, 
$| I_{\omega, k}^1 | \le N^{\alpha + 1/2}$, 
for sufficiently large 
$N$.
%
%\QED
\end{proof}
%

%%%%%%%%%%%%%%%%%%%%%%%%%%%%%%%%%%%%%%%
%
\begin{proposition}\label{prop:rel_prufer_conv1}
(Assumption \ref{assumpt:prufer1}, Part 1) 
Suppose that 
$| \langle e^{ 2i \eta_{\pm} } \rangle | < 1$.
Then 
the family of functions 
$\Psi_L$ 
satisfy 
$\lim_{L \to \infty}\Psi_L (x) = x$, 
%locally uniformly 
a.s.
\end{proposition}
\begin{proof}
\noindent
Case 1 : Suppose that $L$ is on a polymer node:  $L = \sum_{ \ell=0 }^{N-1} L_{\omega_{\ell}}$.
The Pr\"ufer phase $\theta_L$ satisfies 
\bea
\theta_L (E_c + \epsilon)
&=&
S^N_{\epsilon, \omega} (\theta) - \theta
=
\sum_{\ell = 0}^{N-1}
\left(
S_{\epsilon, \omega_{\ell}}
( S^{\ell}_{\epsilon, \omega} (\theta) )
-
S^{\ell}_{\epsilon, \omega} (\theta)
\right).
\eea
We recall that from \cite[section 4.2]{jss}, 
\bea\label{eq:pert1}
S_{\epsilon, \pm} (\theta) - \theta &=&
\eta_{\pm}  + \epsilon  d_{\pm}  -  \epsilon ~Im
\left[ c_{\pm} e^{2i \theta}
\right]   +
R_{\epsilon}, 
\quad
| R_{\epsilon} |
\le C |b^{\epsilon}_{\pm}|^2
=
{\mathcal{O}}(\epsilon^2) 
%\quad\cdots (*)
\eea
where  the coefficients $c_\pm$ and $d_\pm$ are defined in \eqref{eq:coef2}, and $C$  is a universal constant.
Thus, substituting 
$S^{\ell}_{\epsilon, \omega} ( \theta )$
to 
$\theta$
in \eqref{eq:pert1},  we obtain 
\bea\label{eq:pert2}
\theta_L (E_c + \epsilon)
&=&
\sum_{\ell = 0}^{N-1}
\left(
S_{\epsilon, \omega_{\ell} } ( S^{\ell}_{\epsilon, \omega} ( \theta ) )
-
S^{\ell}_{\epsilon, \omega} (\theta)
\right)
\nonumber \\
&=&
\sum_{\ell = 0}^{N-1}
\left(
\eta_{\omega_{\ell}}
+
\epsilon d_{\omega_{\ell}}
-
\epsilon \Im
\left[
c_{\omega_{\ell}} e^{2i S^{\ell}_{\epsilon, \omega} (\theta)}
\right]
+
R_{\epsilon, \ell}
\right), 
\eea
where the remainder satisfies
\bel{eq:remainder1} 
|R_{\epsilon, \ell}|
\le C |b^{\epsilon}_{\pm}|^2  = \mathcal{O}(\epsilon^2). 
\ee
%\quad\cdots (**)
%\\
%
Recalling that at a critical energy 
\bel{eq:unpert1}
\theta_L (E_c)  =  \sum_{\ell = 0}^{N-1}
\eta_{\omega_{\ell}}  ,
\ee
and substituting $\epsilon =  \dfrac {x}{ n(E_c) L }$ 
into 
(\ref{eq:pert2}), 
%the definition of $\Psi_L$, 
we get 
\bea\label{eq:psi2}
\Psi_L (x)
&:=&
\frac {1}{\pi}
\left\{
\theta_L 
\left(
E_c + \frac {x}{ n(E_c) L } 
\right) 
- 
\theta_L (E_c)
\right\}   \nonumber 
\\
&=&
\frac {1}{\pi}
\sum_{\ell = 0}^{N-1}
\left(
%\eta_{\omega_{\ell}}
%+
\epsilon  d_{\omega_{\ell}}
-
\epsilon  \Im
\left[
c_{\omega_{\ell}} e^{2i S^{\ell}_{\epsilon, \omega} (\theta)}
\right]
+
R_{\epsilon, \ell}
\right)
\eea
where 
$\epsilon =  \dfrac {x}{ n(E_c) L }$.
The sum of the exponential term in \eqref{eq:psi2} may be estimated as in Lemma \ref{lemma:sum1}, together with the fact that 
$N = \mathcal{O}(L)$. 
We also have that $R_{\epsilon, \ell} = \mathcal{O}(\epsilon^2)
=
\mathcal{O}(L^{-2})$, see \eqref{eq:remainder1}. These facts allow us to write, for any 
$\alpha > 0$,  
\bea
\Psi_L(x)
&=&
\frac {1}{\pi}
\frac {1}{ n (E_c) }
\frac xL
\sum_{\ell = 0}^{N-1} 
%\left(
d_{\omega_{\ell}} 
+ {\mathcal{O}}(L^{-\frac 12 + \alpha})
%\right) 
\nonumber \\
  & =  &
\frac {1}{\pi}
\frac {1}{ n (E_c) }
\frac NL
\frac xN
\sum_{\ell = 0}^{N-1} 
%\left(
d_{\omega_{\ell}} 
+ 
\mathcal{O}(L^{-\frac 12 + \alpha}).
%\right).
\eea
By the law of large numbers, we have, 
\bea
\frac LN
&=&
\frac 1N
\sum_{\ell = 0}^{N-1}
L_{\omega_{\ell}}
\to 
\langle L_{\pm} \rangle, 
\quad
a.s.
\\
&&
\frac 1N
\sum_{\ell = 0}^{N-1} d_{\omega_{\ell}}
\to 
\langle d_{\pm} \rangle
\quad
a.s.
\eea
So therefore we get 
\beq
\Psi_L (x)
\to
\frac {1}{\pi}
\frac {1}{ 
n (E_c)
}
\frac {
\langle d_{\pm} \rangle
}
{
\langle L_{\pm} \rangle
}
x, 
\quad
a.s.
\eeq
By \cite[(47)]{jss}, we have 
$n (E_c) = \dfrac {1}{\pi}
\dfrac {\langle d_{\pm} \rangle }{ \langle L_{\pm} \rangle }$, 
which yields 
$\Psi_L(x) \to x$, 
a.s. \\
\medskip

\noindent
Case 2 : General case.
As we did in 
(\ref{generalcase}), 
we can write
\bea
\Psi_L (E_c)
&=&
\triangle \Psi
+
\Psi_{L'} (E_c), 
\quad
\mbox{ where }
\triangle \Psi :=
\Psi_L (E_c)
-
\Psi_{L'} (E_c).
\eea
As in \eqref{eq:pert1}, the first term 
$\triangle \Psi$
may be written as 
\bea
\triangle \Psi
&=&
\epsilon d'_{}
-
\epsilon \Im
\left[
c'_{\omega_{\ell}} e^{2i S}
\right]
+
R'_{\epsilon}
\to 0  ,
\eea
for appropriate constants $d'$, $c'$, and $R_\epsilon'$. 
\end{proof}
To 
show the local uniform convergence in Assumption \ref{assumpt:prufer2}, we need a little more analysis.
\begin{proposition}\label{prop:rel_prufer_conv2}
(Assumption \ref{assumpt:prufer2}, Part 1)
Suppose that
$| \langle e^{ 2i \eta_{\pm} } \rangle | < 1$.

Then 
$\lim_{L \to \infty}\Psi_L (x) = x$ 
locally uniformly, a.s.
\end{proposition}

%%%%%
%
\begin{proof}
For given 
$L$, 
let 
$N(L)$
be the number of polymer nodes in 
$[0, L]$.
Note that
$\frac {L}{L_+} \le N(L) \le \frac {L}{L_-}$
so that 
$N(L)/L, L/N(L) = \mathcal{O}(1)$.
Fix
$\alpha > 0$, 
and let 
\beq
\Omega(L, \epsilon)
&:=&
\Biggl\{
\omega \in \Omega
\, \Biggl| \,
\text{ for some 
$N \le N(L)$ }, 
\left|
\sum_{\ell = 0}^{N-1}
c_{\omega_{\ell}} e^{2i S^{\ell}_{\epsilon, \omega} (\theta)}
\right|
\ge
N^{\frac 12+ \alpha}, 
\text{ or }
\\
&&
\left|
\frac 1N
\sum_{\ell = 0}^{N-1}
L_{\omega_{\ell}} 
-
\langle L_{\pm} \rangle
\right|
\ge
N^{- \frac 12 + \alpha}, 
\;
\text{ or }
\;
\left|
\frac 1N
\sum_{\ell = 0}^{N-1}
d_{\omega_{\ell}} 
-
\langle d_{\pm} \rangle
\right|
\ge
N^{- \frac 12 + \alpha}
\Biggr\}
\eeq
where 
$L \epsilon^2 < C$.
Then 
by \cite[Theorem 6]{jss} and standard large deviation estimate for i.i.d. random variables, we have
\beq
P (
\Omega(L, \epsilon)
)
\le
C_1
e^{-C_2 L^{\alpha}}.
\eeq
On the other hand, 
proof of 
Proposition \ref{prop:rel_prufer_conv1} implies that, for 
$\omega \in 
\Omega 
\left(
L, \dfrac {x}{ \rho (E_c) L }
\right)^c$, 
then
\beq
\left|
\Psi_L (x) - x
\right|
\le
C L^{- \frac 12 + \alpha}.
\eeq
Take 
$a, b \in {\bf R}$
s.t.
$a < b$, 
and divide the interval
$[a,b]$
into small intervals 
$\{ I_k \}_{k=1}^K$
with 
$I_k = [x_k, x_{k+1}]$
of size
$L^{- \frac 12 + \alpha}$.
Suppose that, 
\beq
| \Psi_L (x_k) - x_k | \le CL^{- \frac 12 + \alpha}, 
\quad
| \Psi_L (x_{k+1}) - x_{k+1} | \le CL^{- \frac 12 + \alpha}. 
\eeq
Then for any 
$y \in (x_k, x_{k+1})$, 
we have
\beq
x_k - CL^{- \frac 12 + \alpha}
\le
\Psi_L (x_k)
&\le&
\Psi_L (y)
\le
\Psi_L (x_{k+1})
\le
x_{k+1} + CL^{- \frac 12 + \alpha}
=
x_k 
+
(C+1) L^{- \frac 12 + \alpha}
\\
x_k
&  \le &
y
\le
x_{k+1} = x_k + L^{- \frac 12 + \alpha}
\eeq
Hence
\beq
-(C+1)L^{- \frac 12 + \alpha}
&\le&
\Psi_L (y) - y
\le
(C+1) L^{- \frac 12 + \alpha}.
\eeq
Let 
\beq
\Omega(L)
&:=&
\bigcup_{k=1}^K
\Omega
\left(
L, \frac {x_k}{ \rho (E_c) L }
\right).
\eeq
Then 
$P (\Omega (L))
\le
C_1 L^{\frac 12 - \alpha} 
e^{- C_2 L^{\alpha}}$, 
and for 
$\omega \in \Omega (L)^c$, 
\beq
\left|
\Psi_L (x) - x
\right|
\le
C L^{- \frac 12 + \alpha}
\eeq
for any 
$x \in [a,b]$.
Let 
$\Omega_0
:=
\limsup_{L \to \infty}
\Omega (L)$.
Then
$P(\Omega_0) = 0$, 
and for 
$\omega \in \Omega_0^c$, 
$\Psi_L (x) \stackrel{L \to \infty}{\to} x$
uniformly in 
$[a,b]$.
%
%\QED
\end{proof}
%

%
%%%%%%%%%%%%%%%%%%%%%%%%%%%%%%%%%%%%
\subsection{Sharpness of the LES transition} 

In this section, we make an exploration of the transition from Poisson to clock statistics as the reference energy $E_0$ transitions from the localization regime to a critical energy $E_c$. 
%In particular, we study the case when $|E_0-E_c|$ is $L$-dependent. 
%That is, 
We consider the case when the LES is centered at a noncritical energy $E_0(L)$ 
which is $L$-dependent and $E_0(L) \to E_c$ as $L \to \infty$:    
%$E_0$
%is 
%$L$-dependent, but converges to 
%$E_c$ : 
%
\bel{eq:close1}
E_0(L) := E_c + \frac {C}{L^{\delta}}, 
\quad
\delta > \frac{1}{2}.
\ee
Then, in these cases, we can still show that 
$\Psi_L (x) \to x$, locally uniformly. In fact, we have the following extension of Theorem \ref{thm:clock_1}. 
\begin{theorem} \label{thm:sharp2}
Assume the following two conditions.
\begin{enumerate}
\item $| \langle e^{ 2i \eta_{\pm} } \rangle | < 1$, 
\item  $0 < \mathcal{N}(E_c) < 1 $.  
\end{enumerate}
We consider 
the LES centered at $E_0(L)$, with $E_0(L)$ as in \eqref{eq:close1}.
Then 
\begin{enumerate}
\item 
Any accumulation point of 
$\xi_L$
is a clock process with some probability measure
$\mu$.
\\
\item The eigenvalue spacing near $E_0(L)$ exhibits clock behavior as in \eqref{eq:ev_space2}. 
\end{enumerate}
\end{theorem}

\begin{proof}
%Using 
%$\epsilon 
%:= 
%\dfrac {C}{L^{\delta}} + \dfrac {x}{ \rho(E_c) L }$, 
%and
%$\epsilon 
%:=
%\dfrac {C}{L^{\delta}}$
%respectively, 
It 
suffices to show part (1) of 
Assumption \ref{assumpt:prufer1}
and part (1) of 
Assumption \ref{assumpt:prufer2}
for a relative Pr\"ufer phase defined by 
\bea\label{eq:rel_prufer3}
\widetilde{\Psi}_L (x) 
& := & \Psi_L \left( E_0(L) + \dfrac {x}{ n(E_c) L } \right) - \Psi_L ( E_0(L)).
\eea
The idea of proof is similar to that for 
Propositions \ref{prop:rel_prufer_conv1}
and \ref{prop:rel_prufer_conv2}.
\end{proof}

\medskip

\begin{remark}\label{rmk:sharp1}
Theorem \ref{thm:sharp2} supports the sharpness of the transition of the unfolded LES as a function of the energy. In the next section, we show that under certain assumptions on the regularity of the IDS, Assumption \ref{assump:dos1}, the limiting point process is Poisson at noncritical energies. In remark \ref{remark:dos1}, we point out that although the \Holder continuity of the IDS for the AB RPM is known, we do not know of a proof for the \Holder continuity of the inverse function $N^{-1}(E)$. It is possible that these \Holder exponents, $\rho_1$ and  $\rho_2$  vary as the energy approaches a critical energy. In particular if $\rho_1 \rho_2$ no longer satisfies Assumption \ref{assump:dos1}, that is, $\frac{2}{3} < \rho_1 \rho_2 < 1$, our proof of Poisson statistics breaks down. To investigate this, we prove in Appendix \ref{app:sharp1} that in the absence of the lower bound of $\rho_1 \rho_2$, the limit points of the point process at noncritical energies are infinitely divisible. 
Since the clock process is not infinitely divisible (see Theorem \ref{thm:clock_notID1}), our results show that the limiting point process is infinitely divisible (Poisson) for all noncritical $E_0 \in \Sigma$, but not infinitely divisible (clock) for all critical energies $E_0 \in \Sigma$. In this sense, the transition is sharp.  
\end{remark}

%
%\begin{remark} We 
%can also treat the more general case for which $E_0 := E_c + o_L(1)$, $L \to \infty$
%using the same argument.
%%so that the same argument works.
%%
%\end{remark} 

%%%%%%%%%%%%%%%%%%%%%%%%%%%%%%%%%%%

\subsection{Summary: Clock statistics at critical energies}\label{subsec:clock1}

We summarize the main results on clock LES at critical energies. 

\begin{theorem}\label{thm:clock_sum1}
Under these assumptions on the RPM, 
\begin{enumerate}
\item The polymer transfer matrices eigenvalues satisfy:
$$
| \langle e^{ik \eta_\pm}  \rangle | < 1, 
\quad
\text{for any }
%\forall k
k \in \Z \setminus \{ 0 \} ; 
$$

\medskip
 
\item The IDS satisfies
$$
0 < {N}(E_c) < 1; 
$$
\end{enumerate}
we have the following results on the LES at critical energies:
\begin{enumerate}
\item The LES centered at a critical energy satisfies $\xi_L 
\stackrel{d}{\to} clock (unif [0,1))$.

\medskip

\item Strong clock behavior: The local eigenvalue spacing  around a critical energy almost surely satisfies  
\bel{eq:ev_space2} 
\lim_{L \rightarrow \infty} n(E_c) L   \left(   E'_{j+1}(L) - E'_j (L)  \right)  = 1 .
\ee  

\medskip

\item If the LES is centered at $E_0 = E_c + \mathcal{O}(L^{-\delta})$
with 
$\delta > \frac 12$, 
then 
the LES is clock$(\mu)$ with some distribution $\mu$ on $[0,1)$, 
and we have the strong clock behavior.
\end{enumerate}
\end{theorem}

There are examples of RPM for which the two assumptions hold.

%\begin{remark}
%We make the following comments about the two assumptions. 
%\begin{enumerate}
%\item \emph{Assumption 1 on the eigenvalues of the transfer matrices $T_\pm^ (E_c)$}.
%
%
%
%\medskip
%
%\item \emph{Assumption 2 on the positivity of the IDS. }
%%\begin{remark}
%%{\bf Remark }
%%%
%%(1)
%%is most important, due to the following reason. \\
%%%
%%(2)
%%holds if we are allowed to take a subsequence.
%%\\
%%%
%(3)
%holds true, if the number of eigenvalues in 
%$(-\infty,  E_c)$, 
%$(E_c, \infty)$
%goes to infinity as 
%$L \to \infty$, 
%which in turn follows from the positivity of 
%${\N}$
%there, and the results in JSS will help us.
%%
%In fact, by 
%(14) 
%in JSS, 
%%
%\bea
%{\N}(E_c)
%&=&
%\frac {
%\langle 
%L_{\pm} {\mathcal N}_{\pm} (E_c)
%\rangle
%}
%{
%\langle 
%L_{\pm} 
%\rangle
%}
%=
%\frac {
%\langle
%\eta_{\pm} / \pi 
%\rangle
%}
%{
%\langle L_{\pm} \rangle
%}
%\quad
%\left(
%{\N}_{\pm} (E_c)
%=
%\frac {
%\eta_{\pm}
%}
%{
%\pi L_{\pm}
%}
%\right)
%\eea
%%
%so that there should be some examples such that 
%$0 < {\N} (E_c) < 1$.\\
%%
%%\end{remark}
%
%\item By Skorohod's theorem \cite{skorohod1}, 
%we may assume that 
%$\phi(E_c, L) / \pi
%\stackrel{a.s.}{\to}  \theta$, 
%where 
%$\theta \sim \mu$, the distribution of $\theta$. 
%%\end{document}
%
%\end{enumerate}
%\end{remark}
%%%%%%%%%%%%%%%%%%%%%%%%%%%%%%%%%%%%%%%%%%%%%%%%%%%%%%%%%%%%%%%%%%%%%%

\section{Localization: LES is Poisson away from critical energies}\label{sec:poisson1}
\setcounter{equation}{0}

The deterministic spectrum $\Sigma$ of the Anderson-Bernoulli random polymer model (AB RPM) is purely pure point. It is a natural question to study the unfolded LES \eqref{eq:les1} centered at energies that are not critical. 
In order to prove that the unfolded LES for AB RPM is Poisson at energies away from the critical energies $\mathcal{C}$, we encounter several new problems: 
\begin{enumerate}
\item There is only a weak Wegner estimate proved in \cite{CKM};
\item The standard proofs of the Minami estimate \cite{CGK1,minami1} do not apply to Bernoulli probability measures;   
\item The IDS is only H\"older continuous.  
\end{enumerate}
This last item means that we do not know if the density of states measure is absolutely continuous with respect to Lebesgue measure. Consequently, we do not know if the density of states function (DOSf) exists. We recall that for random \Schr operators with single-site probability measures $d\mu(s) = \rho(s) ~ds$, with probability density $\rho \in L^\infty (\R)$, the DOSf  $n(E) :={N}'(E)$ determines the intensity of the Poisson point process centered at energy $E$ \cite{minami1}.  

To circumvent this difficulty, we work with the LES constructed using the \emph{unfolded eigenvalues} as in \eqref{eq:les1}, see \cite{GKlo1,minami2}. 
As is mentioned there, 
the unfolded eigenvalues are the random variables constructed by filtering the local eigenvalues $\{ E_j(\Lambda) \}$ through the IDS ${{N}}(E)$. Instead of dealing with the eigenvalues $\{ E_j(\Lambda) \} = \sigma (H_\Lambda)$ directly, we study the unfolded eigenvalues $\{{{N}}(E_j (\Lambda)) \}$. 
%The LES is obtained as in \eqref{eq:} but using the unfolded eigenvalues. 
The \emph{unfolded LES} centered at ${N}({E}_0)$ is the limit (if it exists) of the finite random point measures given by
\bel{eq:unfLES1}
\xi_L  
= 
\sum_{j=1}^L 
\delta_{  
L 
( {{N}}(E_j (\Lambda)) - {{N}}(E_0) )
} , ~~~~~ E_0 \in \Sigma \backslash \mathcal{C}  .
\ee
Germinet and Klopp \cite{GKlo1} treated the unfolded LES for RSO. They proved that, assuming both Wegner's and Minami's estimates, the unfolded LES for random \Schr operators in the localization regime is a Poisson point process with an intensity measure given by the Lebesgue measure. We will prove a similar result for the unfolded LES for the AB random polymer models. 
We recall that 
at the critical energy 
$E_0 \in \mathcal{C}$, 
the DOSf was proved to exist in \cite[Theorem 3]{jss}. 
As a consequence, 
the unfolded LES at 
$E_c$ 
is the same as the usual LES constructed with the eigenvalues 
$\{ E_j(\Lambda) \}$, 
as noted in section \ref{subsec:main1}. 

To state our main result of this section, that the LES is Poisson at noncritical energies, 
we need to discuss the density of states (DOS) of AB RPM. It is well-known that the DOS plays an important role in the proof of local eigenvalue statistics. 
We first recall a well-known regularity result of the integrated density of states (IDS) for AB RPM.

\medskip

\begin{theorem}\label{thm:dos1}\cite{CKM}
The IDS 
${N}: \Sigma \rightarrow [0,1]$ 
of the AB RPM  is locally uniformly \Holder continuous, that is, for any compact interval 
$I$, 
we can find positive constants
$\rho_1 \in (0,1]$
and 
$C < \infty$,
such that 
%\medskip
\bel{eq:holder_dos1}
|{N}(E) - {N}(E') | \le C |E - E'|^{\rho_1}
\ee
for any 
$E, E' \in I$.
\end{theorem}

For our LES result, 
we also need some regularity on the inverse function ${{N}}^{-1} : [0,1] \rightarrow \Sigma$.

\medskip

\begin{assumption} (\Holder continuity of ${N}^{-1}$)
\label{assump:dos1}

\begin{enumerate}   
\item
${{N}}^{-1}$ 
is \Holder continuous with an exponent 
${\rho_2} \in (0, 1]$, 
that is, for any compact interval 
$I$ 
we can find positive constants
$\rho_2 \in (0,1]$
and 
$C < \infty$
such that for 
$u, u' \in [0,1]$, 
\bel{eq:inv_holder_dos0}
| {N}^{-1}(u) - {N}^{-1}(u') |  
\le  
C |u - u'|^{{\rho_2}}.
\ee
In other words, for 
$E =  {N}^{-1}(u)$, 
and 
$E' =  {N}^{-1}(u')$, 
we have 
\bel{eq:inv_holder_dos1}
|{N}(E) - {N}(E') | \ge C^{\frac{1}{\rho_2}} |E - E'|^{\frac{1}{\rho_2}}.
\ee

\medskip

\item
${\rho_1 \rho_2} >   \dfrac 23 .$
\end{enumerate}
\end{assumption}

By definition, 
$\rho_1 \rho_2 \le 1$.
Part (2) of Assumption \ref{assump:dos1} 
states that this product $\rho_1 \rho_2$ is not too small. 

\medskip

Our main theorem 
is the following result on the unfolded LES \eqref{eq:unfLES1} away from the finite set 
$\mathcal{C}$ 
of critical energies. 

\medskip

%
%%%%%%%%%
\begin{theorem}\label{thm:les_loc1}  
(Poisson statistics) % for unfolded point process)
Let 
$E_0 \in \Sigma \backslash \mathcal{C}$. 
We assume that 
the IDS of the AB RPM satisfies Assumption \ref{assump:dos1}. 
Then, 
the unfolded LES in \eqref{eq:unfLES1} converges in distribution to the Poisson point process with intensity measure given by the Lebesgue measure:  
$$
\xi_L  
\stackrel{d}{\to}  ~~~\mbox{Poisson}(dx).
$$
\end{theorem}
%%%%%
%

\begin{remark}
\label{remark:dos1}
\mbox{}
\begin{enumerate}
\item
In fact, to prove that the LES is Poisson at a noncritical reference energy $E_0$
we only need Assumption \ref{assump:dos1} near the reference energy 
$E_0$. That is, the following \emph{local} condition is enough:
\begin{equation}
0 < 
\lim_{\delta \to 0}
\inf_{
\substack{
| E_0 - E_1 | < \delta, \\
| E_0 - E_2 | < \delta, \\
E_1 \ne E_2
}
}
\frac {
| N(E_1) - N(E_2) |
}
{
| E_1 - E_2 |^{1/\rho_2}
}, 
\quad
\lim_{\delta \to 0}
\sup_{
\substack{
| E_0 - E_1 | < \delta, \\
| E_0 - E_2 | < \delta, \\
E_1 \ne E_2
}
}
\frac {
| N(E_1) - N(E_2) |
}
{
| E_1 - E_2 |^{\rho_1}
}
< \infty.
\label{weakercondition}
\end{equation}
Our proof
for Poisson statistics breaks down if we do not assume the condition 
(\ref{weakercondition})
with
$\rho_1  \rho_2 > \frac 23$, 
although we do not know of an example for which this condition is satisfied. 
However, 
they seems to be the minimum requirement, namely  we do not expect Poisson statistics without this condition.  
In fact, 
\cite[(1.45) of Theorem 1.10]{GKlo1}
also assumes a lower bound on the H\"older continuity of IDS. 
On the other hand, 
Theorem \ref{thm:les_loc1} 
is model-independent, and valid for general situations such that the Hamiltonian satisfies reasonable conditions such as the independence at a distance, 
positivity of the Lyaunov exponent at the reference energy, 
and
applicability of MSA.

\medskip
\item Avila \cite{avila1} showed that the IDS of the almost Mathieu model in the subcritical region $0 < \lambda < 1$ has \Holder exponent $\rho_1 = \frac{1}{2}$ and $N^{-1}$ has \Holder exponent $\rho_2 = \frac{2}{3}$. This is the only example we know of for which $\rho_2$ has been computed.  (except the case where the DOS measure has a density so that the IDS is Lipschitz continuous.)

\medskip

\item A consequence of the assumed \Holder continuity of $N^{-1}$ is that 
for any interval $I \subset \R$, independent of $L$, we have 
\bel{eq: unfolded_int1}
N^{-1} \left( N(E_0) + \frac{1}{L} I \right) \leq C(E_0) \left( \frac{|I|}{L} \right)^{\rho_2} .
\ee

 \end{enumerate}
\end{remark}
\medskip

\medskip

%%%%%%%%%%%%%%%%%%%%%%%%%%%%%%%%%%%%%%%%%%%%%%%%%%%%%%%%%%%%%%%%%%%%%%%%%%%%%%%%
\subsection{Outline of the proof}\label{subsec:outline_loc1}

The strategy
for the proof is basically due to Minami \cite{minami1} and Germinet-Klopp \cite{GKlo1}: in the localized regime, the original system can be well approximated by the direct sum of subsystems, and the eigenvalues of each subsystem are independent of the eigenvalues in the other subsystems. Hence, for example, the random variables that the subsystems have at least one eigenvalue in an interval are close to independent Bernoulli variables, which is a typical situation where the proof of the Poisson limit theorem works.
But since 
the distribution of the random potential is singular, we need to re-construct the proof from the beginning.

\begin{enumerate}
\item
Positivity of Lyapunov exponent: De Bi\`evre and Germinet \cite{dBg1} proved the positivity of the Lyapunov exponent for the Dimer model and Damanik, Sims, and Stolz \cite{DSS} proved it for more general one-dimensional discrete models including the AB RPM. We review these results in section \ref{subsec:lyapunov1}.
%by examining the condition for Furstenberg's theorem, we show the Lyapunov exponent is positive for non-critical energies. 
%
\item
Wegner's estimate : 
we use ``weak-type" Wegner's bound proved by \cite{CKM}, and we improve 
the exponent on the bound.  
\item
Bourgain's lemma : 
Bourgain \cite{bourgain1} 
proved an estimate on the probability that a local Hamiltonian has two eigenvalues in a small interval $I$.
We adjust 
his proof to our situation, and present the proof in Appendix \ref{app:bourgain1}
for completeness. 
\item
Bootstrap MSA and SULE bound : 
We use 
the conclusion obtained by the bootstrap MSA in Germinet-Klein \cite{gk_boot1}.
The conclusion of bootstrap MSA gives a SULE bound for finite volume systems with both polynomial and sub-exponential probability \cite{GKlo1}. 
\item
Minami's estimate : 
we give the bound on the probability that a local Hamiltonian has more than 
$k$
eigenvalues on a small interval, where 
$k \ge 2$.
The 
essential idea is due to \cite{klopp1, shirley1}, but we need some technical modifications to fit to the AB RPM. 
For related remarks, see Appendix
\ref{app:bootstrap1}.
\item
Decomposition 
of the system into small subsystems : 
we construct 
a correspondence between the pairs of eigenvalues and eigenvectors of the original system and those in small subsystems. 
Then, an argument using the proof of Poisson limit theorem completes the proof.
This basic argument is due to Germinet-Klopp\cite{GKlo1}, but the mathematical construction of these correspondences was initiated in \cite{nakano1}.
\end{enumerate}
\medskip

The first proof  of localization for the Anderson-Bernoulli model on $\Z$ is due to Carmona, Klein, and Martinelli \cite{CKM}. 
The standard proofs 
of the Wegner estimate for the Anderson model with singular single-site probability measures do not hold since there is no known spectral averaging result, see \cite{CGK1}. 
In \cite{CKM}, 
the authors prove a weaker version of the usual Wegner estimate of the form: 
\beq
\Pp \{ {\rm Tr} P_{H_{\ell}}(I) \geq 1, |I| \sim e^{- \alpha \ell } \} \leq e^{- \alpha' \ell} , 
\eeq
where 
$P_{H_{\ell}}(I)$ 
is the spectral projector, see Theorem \ref{thm:wegner1}. 
With these results, 
Carmona, Klein, and Martinelli \cite{CKM} were able to control the single-energy multiscale analysis (MSA) and prove exponential localization throughout the deterministic spectrum. 
The same arguments 
hold in the polymer case. 
The dimer case 
$(L_+, L_-)  = (2, 2)$ 
was treated by De Bi\`evre and Germinet \cite{dBg2}. 

The standard proof 
of LES for discrete RSO \cite{minami1} and unfolded LES \cite{GKlo1} requires a Minami estimate, which is an upper bound on 
$\Pp \{ {\rm Tr} P_{H_L}(I) \geq k \}$, 
for 
$k \ge 2$. 
We present a proof 
of such an estimate for AB RSO and RPM (see Theorem \ref{thm:minami1}) drawing from ideas in \cite{bourgain1} and \cite{klopp1}. 
Two other results 
are needed for localization: 
1) The positivity of the Lyapunov exponent, see \cite{DSS}, and 
2) and the local \Holder continuity of the IDS, proved for the AB model in \cite[Appendix]{CKM}. 
Both results 
play a crucial role in their proof of the exponential decay of the Green's function. 
For a recent proof 
of localization for Anderson models on 
$\ell^2 (\Z)$ 
with singular single-site probability measures, we refer to \cite{BDFGVWZ}. 
%

%%%%%%%%%%%%%%%%%%%%%%%%%%%%%%%%%%%%%%%%

\subsection{Positivity of the Lyapunov exponent} \label{subsec:lyapunov1}

We recall 
the definition of the transfer matrix for a finite string of polymers:
\bel{eq:klock_trans_matrix2}
T_\omega^E (k, m) = T^E_{\omega_k} \ldots T^E_{\omega_m}, ~~~ k > m .
\ee
If 
$m > k$, 
we have 
$T^E_\omega(k,m) = T^E_\omega(m,k)^{-1}$, 
and 
$T^E_\omega (k,k) = I_2$, 
where 
$I_2$ 
is the identity matrix. 
The Lyapunov exponent 
is given by 
\bea  \label{eq:lyapunov1}
L (E) &:=& \lim_{k \to \infty}
\frac {1}{ k \langle L_{\pm} \rangle }
\log \| T^E_{\omega} (k, 0) \|, 
\quad
\mbox{ where }
\quad
\langle c_{\pm} \rangle
:= p_+ c_+ + p_- c_-.
\eea
A key component 
of the proof of localization for the dimer model 
$L_+ =  L_- = 2$ 
by De Bi\`evre and Germinet \cite{dBg2} is the verification of the two conditions necessary in order to prove the positivity of the Lyapunov exponent. 
The Furstenberg theorem 
(see, for example, \cite{DSS}) requires two properties:
\begin{enumerate}
\item 
The subgroup 
$G_\pm$ 
of 
$SL(2, \R)$ 
generated by the two polymer transfer matrices 
$\{ T_+^E, T_-^E \}$, 
for 
$E \neq E_c$, 
is noncompact.

\medskip

\item
The subgroup 
$G_\pm$ 
is strongly irreducible.

\end{enumerate}
These conditions 
were proved to hold in \cite{dBg2} for the dimer model by explicit computation.  
The proof 
for the general AB RPM is given in \cite{DSS}. 
We cite 
the result applicable to the AB RPM. 
%re are algebraic difficulties in establishing these results for $G_\pm$ so we will %assume that the Lyapunov exponent is strictly positive away from the critical %energies.

\medskip

\begin{theorem}\cite[Theorem 6.2]{DSS}\label{thm:lyapunov1}
For the AB RPM, 
there is a finite subset 
$\mathcal{C} \subset \R$ 
so that for all 
$E \in \Sigma \backslash \mathcal{C}$, 
the subgroup 
$G_\pm (E)$ 
of 
$SL(2 , \R)$ 
is noncompact and strongly irreducible. Consequently, the Lyapunov exponent $L(E)$ 
for the AB RPM at energies 
$E \in \Sigma \backslash \mathcal{C}$ 
is strictly positive. 
\end{theorem}

\medskip

%%%%%%%%%%%%%%%%

%%%%%%%%%%%%%%%%%%%%%%%%%%%%%%%%

\subsection{Wegner estimate}\label{subsec:wegner1} 
We use 
the version of the Wegner's estimate established in \cite{CKM} for RSO with singular single-site probability measures, including Bernoulli measures.
The proof 
extends to AB RPM. 
Briefly, 
the proof in \cite{CKM} depends on two lemmas. 
The first is 
\cite[Lemma 5.1]{CKM} on the exponential decay of the expectation of products of one-step transfer matrices. 
For the RPM, 
one uses the one-polymer transfer matrices 
$T^E_\pm$, 
with 
$E \in \Sigma \backslash \mathcal{C}$, 
as defined in \eqref{eq:one_poly_trans1}. 
The main bound is as follows. 
For any bounded interval 
$I \subset \Sigma$, 
there are finite constants 
$\alpha_1, \delta > 0$, 
and a positive integer 
$n_0$, 
so that for 
$n \geq n_0$ 
we have:
\bel{eq:transfer_it1}
\E \{ \| T_{\omega_n}^E T_{\omega_{n-1}}^E \ldots T_{\omega_1}^E {\bf{x} }\|^{-\delta} \} \leq C e^{-\alpha_1 n} .
\ee
The proof 
uses the positivity of the Lyapunov exponent. 

The second lemma, 
\cite[Lemma 5.2]{CKM}, 
concerns the existence of approximate eigenvectors $\widetilde{\varphi}$ of a local \Schr operator $H_\ell$, and is based on a result of Halperin. 
This lemma uses the \Holder continuity of the IDS with \Holder exponent 
$\rho_1 > 0$. 
For any $E \in \Sigma \backslash {\mathcal{C}}$, there exists a function 
$\widetilde{\varphi}$ so that 
\bea\label{eq:halperin1}
\lefteqn{\Pp \{ 
\exists E' \in (E-\epsilon, E+\epsilon), 
\mbox{and an approximate eigenvector} 
~~ \widetilde{\varphi}, \| \widetilde{\varphi} \| = 1,    }  
\nonumber \\
   &&  
   \mbox{so that} ~ |\widetilde{\varphi}(\- \ell)|^2 + |\widetilde{\varphi}(-\ell)|^2 \leq \epsilon^2 \}      
   \nonumber \\ 
  && \leq  C \ell \epsilon^{\rho_1} . 
\eea
Given these preliminaries, 
the Wegner estimate for the AB RPM, based on \cite{CKM}, is: 
\medskip

\begin{theorem}\label{thm:wegner1}
For any 
$0 < \beta \leq 1$, 
and for all 
$\sigma > 0$, 
there exist finite constants
$\ell_0 = \ell_0(I, \beta, \sigma) > 0$ 
and 
$\alpha = \alpha(I, \beta, \sigma) > 0$ 
and a constant 
$C_W = C_W(\beta, I) > 0$, 
such that
\bel{eq:wegner1}
\Pp 
\{ 
{\mbox dist} 
( \sigma ( H_\ell) , E) 
\leq 
e^{- \sigma \ell^\beta} 
\} 
\leq 
C_W e^{- \alpha \ell^\beta} ,
\ee
for all 
$E \in I$ 
and all 
$\ell > \ell_0$.
%
%where $\rho'$, the Holder exponent of the IDS,  is defined in \eqref{eq:holder_dos1} .
\end{theorem}

\medskip

\begin{remark}
We note 
that the result in \cite{CKM} can be extended to include 
$\beta = 1$, 
but this fact is not used in this paper. 
\end{remark}

%%%%%%%%%%%%%%%%%%%%%%%%%%%%%%
For any $0 < \tau < 1$, the constant 
$\alpha$ 
in \eqref{eq:wegner1}
is bounded by
\beq
\alpha &<&   \min  \Bigl\{  \sigma \rho_1, \,  \frac {\kappa}{2} \rho_1, \,
\frac {\kappa + \sigma}{2} - \tau\log M, \,
\frac {\alpha_1}{2} \tau  \Bigr\}
\eeq 
where 
$\rho_1$ 
is the \Holder exponent of the IDS, defined in \eqref{eq:holder_dos1},
$\kappa  := 
\dfrac {\tau \alpha_1}{2 \delta}$, 
and
\beq
%\quad
M 
&:=& 
\max \left\{ 
\sup_{E \in I} {\E} \{ \| T_+^E \| \}, 
\sup_{E \in I} {\E} \{ \| T_- ^E \| \} 
\right\}.  
%], 
%\quad
%\alpha_1 = 
%\frac {\gamma}{4} \delta
\eeq
The constants $\alpha_1, \delta$ are defined in \eqref{eq:transfer_it1}. 
%\cite[Lemma 5.1]{CKM}. 
%A constant $\gamma$ is defined by $\gamma := \frac{4 \alpha_1}{\delta}$.

%are constants in Lemma YY, and $\rho$ is defined in \eqref{eq:rho1}. 
%%%%%%%%%%%%%%%%%%%%%%%%%%%%%%%%%%%%%%%
%
The constant 
$\sigma \rho_1$ 
in the definition of 
$\alpha$ 
is crucial in our argument. 
Consequently, taking 
$\tau > 0$
sufficiently small such that 
$\dfrac {\kappa}{2} - \tau \log M > 0$, 
we set 
\bel{eq:const_q1}
q :=  
\min  
\Bigl\{  
\frac {\kappa}{2} \rho_1, \,
\frac {\kappa}{2} - \tau\log M, \,
\frac {\alpha_1}{2} \tau
\Bigr\}.
\ee
Then
we obtain a second version of the Wegner estimate that will be useful in the proof of the Minami's estimate used in this paper. 
%

%
%%%%%
\begin{follow}\label{corollary:wegner_v2}(Wegner Estimate, second version)
For any bounded interval $ I \subset \R$, and for constants $0 <  \beta \le 1$, 
$\sigma > 0$, there exists a length scale $\ell_0 := \ell_0 (I, \beta, \sigma) > 0$, and an index $q > 0$ as in \eqref{eq:const_q1}, %:= q (I, \beta, \sigma) > 0$,
so that 
\bel{eq:wegner2}
{\Pp} \bigl\{d (E, \sigma(H_{{\ell}}))  \le  e^{- \sigma \ell^{\beta}}  \bigr\} 
\le  C_W   \left(
e^{- \sigma \rho_1 \ell^{\beta}}
+   e^{- q \ell^{\beta}}
\right)
\ee
for all $E \in I$, and $\ell \ge \ell_0$. %The exponent $q$ is defined in \eqref{eq:exp1}.
\end{follow}  
%%%%%
%

%%%%%%%%%%%%%%%%%%%%%%%%%%%%%%%%%%%%%

\subsection{Bootstrap MSA and SULE bounds}\label{subsec:loc_bds1}

We use 
the bootstrap multiscale analysis (MSA) developed by Germinet and Klein \cite{gk_boot1}.
We begin 
with some notation and definitions.
% 
%\begin{defn}\label{defn:subsets1}
%\begin{enumerate}
%%
%\item 
%%
%The interval of length 
%$L$ 
%centered at 
%$x \in \Z$ 
%is denoted by 
%$\Lambda_L (x) := \{ y \in \Z ~|~ |x-y| \leq \frac{L}{2} \}$ 
%and 
%$|\Lambda_L(x)| = 2\lfloor \frac{L}{2} \rfloor +1$.  
%
%
%\medskip
%
%%
%\item 
%%
%The characteristic 
%function of the interval 
%$\Lambda_L(x)$ 
%is denoted by  
%$\chi_{x,L}$. 
%%
%When $x=0$, 
%we write 
%$\Lambda_L$, 
%and when $L=1$, 
%we write 
%$\chi_x := \chi_{\Lambda_1(x)}$, 
%for short. 
%
%\medskip
%
%%
%\item 
%%
%We let 
%$\Gamma _{x,L}$ 
%denote the function localized near the boundary points of 
%$\Lambda_L(x)$, 
%that is, near the points 
%$x =\pm \lfloor  \frac{L}{2} \rfloor $.
%%
%\end{enumerate}
%\end{defn}
%%
%\medskip
%%%%%%%%%%%%%%%%%%%%%%%%%%%%%%%%
Although the usual definitions of subintervals of $\Z$ can be extended to the case $L_+ = L_- > 1$, when 
$L_+ \ne L_-$, 
 the position of each polymer nodes is random which may cause inconveniences in the usual decomposition of 
$H_{\Lambda}$
into subsystems in the proof of Poisson limit theorem.
Thus, we work under alternative definition of finite intervals. 
\begin{defn}\label{defn:subsets_polymer}
We always assume that $L \gg L_\pm$.
For given
$\omega = (\omega_{\ell})_{\ell \in \Z}$, 
we fix the configuration 
$(t_{\omega}, v_{\omega})
=
( \hat{t}_{\omega_{\ell}},  \hat{v}_{\omega_{\ell}})_{\ell \in \Z}$
such that the left end of 
$\omega_0$
starts at the origin of 
$\Z$.
Let 
\beq
p_n 
&:=&
\begin{cases}
\sum_{\ell=0}^{n-1} L_{\omega_{\ell}} & (n \ge 1) \\
0 & (n=0) \\
-\sum_{\ell=n}^{-1} L_{\omega_{\ell}} & (n \le -1) 
\end{cases}
\eeq
be the position of the 
$n$-th polymer node. That is, the left endpoint of $n$-th polymer of length $L_{\omega_n}$.
Moreover, let 
$B_{n} \subset \Z$
be the corresponding 
$n$-th 
block (polymer): 
\beq
B_n
&:=&
p_n + 
\{ 0, 1, \cdots, L_{\omega_n} \}, 
\quad
n \in \Z.
\eeq
Then, for any $\omega$, we have
$\Z = \bigcup_{\ell \in \Z} B_{\omega_{\ell}}$.
We define the finite interval 
$\Lambda_L$
of size
$L$
to be the union of 
$L$-blocks starting from the origin : 
\beq
\Lambda_L
&:=&
\bigcup_{\ell=0}^{L-1}
B_{\omega_{\ell}}.
\eeq
Then 
$\Lambda_L$
is a random set, but we have
$L_- L \le \# \Lambda_L \le L_+ L$.
%for some constants 
%$c, d > 0$.
%
Likewise,  
for a polymer node
$x \in \{ p_n \}_{n \in \Z}$, 
we define
$\Lambda_L (x)$
to be the union of 
$\left
\lfloor
\frac L2
\right
\rfloor$-intervals to both sides of 
$x$, 
that is, if 
$x = p_n$
for some 
$n\in \Z$,  
then
\beq
\Lambda_L (x)
&:=&
B_{ \omega_{n-L'}}
\cup
\cdots
\cup
B_{ \omega_{n+L'-1} },
\quad
\text{where }
\quad
L'
:=
\left
\lfloor
\frac L2
\right
\rfloor
\eeq
We call 
$\Lambda_L$
an interval of size 
$L$, 
and call 
$\Lambda_L (x)$
the interval of size
$L$
centered at 
$x$.  We note that $\Lambda_L$ is not equal to $\Lambda_L(0)$. 
\end{defn}

For the characteristic function of an interval $\Lambda_L(x)$, we write $\chi_{x,L}$. For a general interval $\Lambda \subset \Z$, we write $\chi_\Lambda$ for its characteristic function. We let 
$\Gamma _{x,L}$ 
denote the function localized near the boundary points of 
$\Lambda_L(x)$.  
%that is, near the points 
%$x =\pm \lfloor  \frac{L}{2} \rfloor $.
%%

The concept 
of a \emph{regular interval} for the Green's function an energy 
$E$ 
is central to MSA. 
For 
$x, y \in \Z$, 
and 
$z \in \C \backslash \Sigma$, 
the Green's function 
$G_H (x,y ; z)$ 
is the integral kernel of the resolvent 
$R_H(z) := (H - z)^{-1}$
 of 
$H$ 
evaluated at 
$x$ 
and 
$y$: 
$G_H (x,y;z) := \langle \delta_x, R_H(z) \delta_y \rangle$. 
For any 
$\Lambda \subset \Z$, 
the Green's function of the restricted \Schr operator 
$H_\Lambda := \chi_{\Lambda} H_\omega \chi_\Lambda$ 
is denoted by 
$G_\Lambda (x,y; z)$, 
or by 
$G_\Lambda (z)$.  
If the interval 
$\Lambda$ 
is of length 
$L$, 
we write 
$H_L$ 
and 
$G_L(x,y;z)$. 
For a random \Schr operator 
$H_\omega$, 
we sometimes suppress the configuration index 
$\omega$. 

\medskip
 
\begin{defn}\label{defn:regular0}
The interval 
$\Lambda_L(x)$ 
of length 
$L$ 
centered at 
$x$, 
is called 
\textbf{$(m, E)$-regular} 
if
\bel{eq:regular1}
\| \Gamma_{x, L} G_{ \Lambda_L (x) } (E) \chi_{x, L/3} \| \le e^{- \frac{m L}{2}}  ,
\ee
where the operator norm is on 
$\ell^2 (\Lambda_L (x) )$. 
%\end{enumerate}
\end{defn}

\medskip

%
%%%%%%%%%
%\begin{itembox}[l]
%\begin{assumption}\label{assump:msa1} (Bootstrap MSA)
\begin{theorem}\label{thm:msa1} (Bootstrap MSA)
Let $I$ be any interval $I \subset \Sigma \backslash \mathcal{C}$. For $x,y \in \Z$ and $E \in \R$, we define the event:
\bea
R(m, L, I, x, y)
&:=&
\Bigl\{
\omega \in \Omega 
\, \Bigl| \,  
\text{ For any } E \in I, 
\nonumber
\\
&&
\qquad
\text{either }
\Lambda_L (x) \mbox{ or }
\Lambda_L (y)   \mbox{ is }  (m, E)   \mbox{-regular}  
\Bigr\}.
\eea
%
%We assume that for any interval $I \subset \Sigma \backslash \mathcal{C}$,
Then, for any $\zeta \in (0,1)$, and for all $\alpha \in (1, \zeta^{-1})$, and with a scale $L_k := L_0^{\alpha^k}$, the following estimates hold for $k \in \N$: 
\bea
&&
{\Pp} \left\{ R(m, L_k, I, x, y) \right\}  \ge  1 - e^{-L_k^{\zeta}} .
%\quad\cdots (MSA)
\eea
\end{theorem}
%\end{assumption}

Germinet and Klein \cite{gk_boot1} proved this theorem for the AB RSO. As mentioned by Germinet and De Bi\`evre \cite{dBg2}, the proof extends to the dimer model.  The positivity of the Lyapunov exponent is essential in these proofs. The proof can be extended to the AB RPM without major difficulties, see Appendix \ref{app:bootstrap1}.  

Two important results of Theorem \ref{thm:msa1} are 1)  the subexponential decay of operators of the form $\chi_x f(H_\omega) P_{H_\omega} (I) \chi_y$ for intervals $I \subset \Sigma \backslash \mathcal{C}$, and 2) strong Hilbert-Schmidt dynamical localization in such intervals.

Among the main results of the bootstrap MSA given in Theorem \ref{thm:msa1} are the essential localization bounds on eigenvectors. These are explicit exponential bounds on the decay of eigenvectors. 
To state these, we need the notion of a \emph{center of localization} (COL) of an eigenvector.

\begin{defn}\label{defn:col1}
A center of localization of an eigenvector $\varphi$ of $H_{\Lambda_L}$ is any $m \in \Lambda_L$ so that $| \varphi (m)| \geq |\varphi(p)|$, for all $p \in \Lambda_L$. We will write COL for any such eigenvector maximizer and denote one by $x_{\varphi}(\Lambda_L)$.
\end{defn}

Then 
the bootstrap MSA result Theorem \ref{thm:msa1} implies a SULE estimate on the eigenvectors of 
$H$ 
restricted to an interval 
$\Lambda_L$.

\medskip

\begin{lemma}
\label{lemma:lemma1.1}
\cite[Lemma 1.1]{GKlo1}
Let 
$H_{\Lambda_L}$ 
be the \Schr operator for the AB RPM restricted to 
$\Lambda_L \subset \Z$. 
Let 
$\varphi_j^L$ 
be a normalized eigenvector of 
$H_{\Lambda_L}$ 
associated with eigenvalue 
$E_j(\Lambda_L) \in I$,
and let  
$x_j(\Lambda_L) \in \Lambda_L$ 
be a COL of 
$\varphi_j^L$.
Then 
we have the following localization result for the eigenvectors.   
%(For the Wegner estimate, see, Theorem \ref{thm:wegner1}, section \ref{subsubsec:wegner1}.) 
%
\begin{enumerate}
\item
For any 
$p > 0$ 
and 
$\xi \in (0,1)$, 
and for 
$L > 0$ 
large enough, there exists a set 
$\mathcal{U}_L \subset \Omega$ 
of configurations such that 
$\Pp \{ \mathcal{U}_L \} \geq 1 - L^{-p}$, 
and for any 
$\omega  \in \mathcal{U}_L$, 
one has
\bel{eq:sule1}
| \varphi_j^L (x) | 
\leq 
L^{p+1}e^{- |x - x_j(\Lambda_L) |^\xi}, 
\quad
x \in \Lambda_L,
\ee
\medskip
\item  
For any 
$\nu, \xi \in (0,1)$, 
with 
$\nu < \xi$,  
and for 
$L > 0$ 
large enough, there exists a set 
$\mathcal{V}_L \subset \Omega$ 
of configurations such that 
$\Pp \{ \mathcal{V}_L \} \geq 1 - e^{-L^{\nu}}$, 
and for any 
$\omega  \in \mathcal{V}_L$, 
one has
\bel{eq:sule2} 
| \varphi_j^L (x) | 
\leq 
e^{2L^{\nu}} 
e^{- |x - x_j(\Lambda_L) |^\xi} .
\ee
\end{enumerate}
\end{lemma}

\medskip

We remark 
that in the localized regime, the distance between any two COL for an eigenvector 
$\varphi$ 
of 
$H_{\Lambda_L}$ 
is 
$\mathcal{O}( (\log L)^{\frac{1}{\xi}} )$, 
with probability 
$1 - L^{-p}$, 
for any $p > 0$
\cite{nakano1, GKlo1}. 
Lemma \ref{lemma:lemma1.1} 
is proved in 
\cite[Theorem 6.1]{GKlo1}.
%

%%%%%%%%%%%%%%%%%%%%%%%%%%%%%%%%%%%%%%%%%%%%%%%%%%%%%%%%%%%%%%%%%%%

\subsection{Minami's estimate for AB RPM}\label{subsec:w_m1}
In what follows, 
we will consider to decompose the given interval 
$\Lambda_L$ of lattice sites into disjoint union of smaller intervals, such as 
$\Lambda_L = \bigcup_{p} C_p$.  We also consider the restriction 
$H_{C_p} = H_{\Lambda_L} |_{C_p}$
of our Hamiltonian 
$H_{\Lambda}$
to  $C_p$.
In order that 
$(H_{C_p})$
are independent, we henceforth require that the all division points of our decomposition  $\{C_p \}$
must at successive be polymer nodes. 
This may cause the slight change of the size of each interval, but it does not affect our estimates. 
We divide the interval $\Lambda_L$
of size
$L$
into smaller intervals 
$\{ C_p \}_{p=1}^{N_L}$
of size 
$\ell_1 := L^{\beta}$
($0 < \beta < 1$), 
so that $N_L
=
\mathcal{O}(L^{1 - \beta})$. We denote by $H_{C_p}$ the restriction of $H_{\Lambda_L}$ to $C_p$.
Let 
$0 < \gamma \le 1$, 
and let 
$J \subset {\bf R}$
with 
$|J| = c_2 >0$
be a bounded interval (independent of $L$), and set 
$I := E_0 + \dfrac {J}{L^{\gamma}}$, so that 
$| I | = c_2 / L^{\gamma}$.
%
%Moreover, let 
%$\eta_p (I)
%:=
%\#
%\{
%\text{eigenvalues of $H_{C_p}$ in } I\}$.
%

%
%%%%%%%%%
%\begin{itembox}[l]
\begin{theorem}\label{thm:minami1} ({Minami estimate for AB RPM}) \\
\smallskip
%{\bf Theorem (Minami's estimate for AB RPM)}\\
For $H_{C_p}$ and energy interval $I \subset \Sigma \backslash \mathcal{C}$ defined above, we set 
$$\eta_p (I)
:=
\#
\{
\text{eigenvalues of $H_{C_p}$ in } I\}. 
$$
\mbox{}\\
(1)
Suppose that 
$\beta \le \rho_1\gamma$.
Then for any 
$\epsilon > 0$
we can find a positive constant
$C_{\epsilon}$
s.t. 
\begin{equation}
%\sum_p
\sum_{k \ge 2}
\Pp
\left(
\eta_p (I) \ge k
\right)
\le
C_{\epsilon}
\left(
\frac {
L^{\beta}
}
{
L^{\rho_1\gamma/(1+\epsilon)}
}
\right)^2.
\label{eq:minami2}
\end{equation}
(2)
Suppose that 
$\beta \le \rho_1\gamma$
and
$\rho_1\gamma > \dfrac {1+\beta}{2}$.
Then
\begin{equation}
\sum_p
\sum_{k \ge 2}
\Pp
\left(
\eta_p (I) \ge k
\right)
= o(1).
\label{eq:minami1}
\end{equation}
\end{theorem}
\begin{proof}
We further divide each
$C_p$
into smaller intervals 
$\{ D^{(p)}_j \}_j$
of size 
$\ell_2 := (c_1 \log L)^{1 / \xi}$, 
where 
$c_1 > 0$, 
and 
$\xi \in (0,1)$
is the constant in the SULE estimate in Lemma \ref{lemma:lemma1.1}. 
We note that 
$L
=
e^{\frac {\ell_2^{\xi}}{c_1}}$,
and 
$|I|
=
\dfrac {c_2}{L^{\gamma}}
=
c_2
\cdot
e^{
-\frac {\ell_2^{\xi}}{c_1}
\gamma
}$.
Here, 
we work under the configuration such that we have a SULE bound, and estimate the error caused by this replacement.
To be precise, 
we take any interval $C_p$, 
and let 
$\{ E_i(C_p) \}_{i=1}^{N(I, C_p)}$,
$\{ \varphi^{(C_p)}_i \}_{i=1}^{N(I, C_p)}$, 
and 
$\{ x_i(C_p) \} _{i=1}^{N(I, C_p)}$
%$\{ \varphi^{(p)}_j \}_{j=1}^{| C_p |}$
be the eigenvalues, corresponding normalized eigenvectors, and their COLs, respectively, of the Hamiltonian $H_{C_p}$ in $I$.
The index $N(I,C_P)$ is defined by $N(I, C_p)
:=
\# \sigma (H_{C_p}) \cap I$.
Let 
\beq
L (I, p, q)
&:=&
\left\{
\omega \in \Omega
\, \middle| \,
\forall i = 1, 2, \cdots, N(I, C_p), 
\;
\exists x_{i} \in C_p
\;
\mbox{ s.t. }
| \varphi_i^{(C_p)} (x) | 
\le
(L^{\beta})^q
e^{- |x - x_i|^{\xi}}
\right\}.
%\quad \cdots (Loc)
\eeq
Then by Lemma \ref{lemma:lemma1.1}, 
we have, for any 
$p'>0$, 
we can find
$q = q(p', d) > 0$, 
s.t. for 
$L$
large enough, 
\beq
\Pp
\left(
L(I, p, q)
\right)
\ge
1 - L^{-\beta p'}.
%\quad\cdots (Loc)_{p}
\eeq
Then
the error caused by replacing 
$\Omega$
by 
$\Omega \cap L(I, p, q)$
turns out to be negligible.
In fact, 
we compute
\beq
\sum_p
\sum_{k \ge 2}
\Pp
\left(
\eta_p (I) \ge k
\right)
&=&
\sum_p
\sum_{k \ge 2}
\Bigl\{
\Pp
\left(
\eta_p (I) \ge k
\, ; \,
L(I, p,q)
\right)
+
\Pp
\left(
\eta_p (I) \ge k
\, ; \,
L(I, p,q)^c
\right)
\Bigr\}
\\
&=&
\sum_p
\sum_{k \ge 2}
\Pp
\left(
\eta_p (I) \ge k
\, ; \,
L(I, p,q)
\right)
+
\mathcal{O}( L^{- \beta p' + 1 } ). 
\eeq
Here we set
$\Pp (A \,; \, B)
:=
\Pp(A \cap B)$.
For given 
$\beta$, 
by taking  
$p'$
sufficiently large,  
this error becomes negligible.
Thus 
it suffices to estimate the first term.
Here 
we consider the following two events. 
For each
$p = 1, 2, \cdots, N_L$, 
let 
\beq
\Omega_p^b
&:=&
\left\{
\omega \in \Omega
\, \middle| \,
\exists j
\mbox{ s.t. }
\#
\{ x_i(C_p) \} \cap D_j^{(p)} \ge 2
\right\}, 
\quad
p = 1, 2, \cdots, N_L,
\\
\Omega_p^g
&:=&
\Omega \setminus \Omega_p^b.
\eeq
We decompose
\beq
&&
\Pp
\left(
\eta_p (I) \ge k
\, ; \,
L(I, p, q)
\right)
\\
&=&
\Pp
\left(
\eta_p (I) \ge k
\; ; \,
\Omega^b_p \cap L(I, p, q)
\right)
+
\Pp
\left(
\eta_p (I) \ge k
\; ; \,
\Omega^g_p \cap L(I, p, q)
\right)
\\
&=:&
I_{p,k} + II_{p,k}.
\eeq
We shall show separately that 
$ \sum_{k \ge 2}
I_{p, k}$, 
$\sum_{k \ge 2}
II_{p, k}$
satify the estimate part (1) in Theorem \ref{thm:minami1}, 
and then show further that 
$\sum_p 
\sum_{k \ge 2}
I_{p, k} = o (1)$, 
and
$\sum_p 
\sum_{k \ge 2}
II_{p, k} = o (1)$.
For each 
$D_j^{(p)}$, 
let 
$\widetilde{D}_j^{(p)}$
be the box of size
$2 \ell_2 = 2 (c_1 \log L)^{1/\xi}$
which has the same center as 
$D_j^{(p)}$.\\
(1)
Estimate on 
$I_{p, k}$ : \\
For
$\omega \in \Omega^b_p \cap L(I, p, q)$, 
we can find 
$D_j^{(p)}$
and 
$\varphi_{\sharp}^{(C_p)}$, 
with 
$\sharp = i, i' (= 1, 2, \cdots, N(I, C_p))$
such that the corresponding COL'S satisfy
$x_i(C_p), x_{i'}(C_p) \in D_j^{(p)}$.
Then
\begin{equation}
| \varphi_{\sharp}(x) |
\le
( L^{\beta} )^q 
e^{
- 
\Bigl(
\frac {(c_1 \log L)^{1/\xi}}{4} 
\Bigr)^{\xi}
}
=
L^{\beta q - \frac {c_1}{4^{\xi}}}
\ll
\frac {c_2}{L^{\gamma}}, 
\quad
x \in \partial \widetilde{D}_j^{(p)}, 
\quad
\sharp = i, i'
%\quad\cdots (*)
\label{starminami}
\end{equation}
by taking 
$c_1$
sufficiently large. 
Then in the box
$\widetilde{D}^{(p)}_j$, 
$\varphi_{\sharp}$, 
$\sharp = i, i'$
become the approximate eigenfunctions of 
$H_{p, j}:= H |_{\widetilde{D}^{(p)}_j}$
so that 
$H_{p, j}$
has two eigenvalues in 
$I$, 
at the cost of increasing slightly the constant 
$c_2$.
By 
Lemma \ref{lemma:bourgain1}, 
we have
\beq
I_{p,k}
\le
\Pp
\left(
\Omega^b_p
\right)
& \le &
C
\left(
e^{ 
- \frac 
{2 \rho_1}
{ 2^{\xi} \cdot c_1 (1 + \epsilon)} 
(2\ell_2)^{\xi}
\gamma
}
+
e^{- q t_0 (2\ell_2)^{\xi}}
\right)
\cdot
\sharp 
\{ 
\mbox{ boxes $\widetilde{D}^{(p)}_j$ in $C_p$ }
\}
\\
&=&
C
\left(
e^{ - \frac {2 \rho_1}{c_1 (1 + \epsilon)} \ell_2^{\xi}\gamma}
+
e^{- q t_0 (2\ell_2)^{\xi}}
\right)
\cdot
\frac 
{L^{\beta}}{ (c_1 \log L)^{1 /\xi} }.
\eeq
Here we take 
$\beta = \xi \in (0,1)$
in Lemma \ref{lemma:lemma1.1}. 
Note that
the constant 
$c_1$
in Lemma \ref{lemma:bourgain1} should be replaced by 
$\dfrac {
2^{\xi} c_1
}
{
\gamma 
}$
and moreover, 
since we have 
$(\ref{starminami})$
and since 
$\dfrac {c_2}{ L^{\gamma }}
\ll
\dfrac {1}{ 10 \ell_2 }$, 
the corresponding eigenvectos of 
$H_{p, j}$
satisfy the condition 
(3)
in Lemma \ref{lemma:bourgain1}, under the any choice of 
$0 < t_0 < \frac 12$.
Now 
we take the sum 
$\sum_{k=2}^{|C_p|}$
with respect to 
$k = 2, \cdots, |C_p|$, 
where 
$|C_p| = L^{\beta} (1+o(1))$.
\beq
\sum_{k=2}^{|C_p|} 
I_{p, k}
& \le &
C \sum_{k=2}^{|C_p|}
\frac 
{L^{\beta}}{ (c_1 \log L)^{1 /\xi} }
\left(
e^{ - \frac {2 \rho_1\gamma}{c_1 (1 + \epsilon)} \ell_2^{\xi}}
+
e^{- q t_0 (2\ell_2)^{\xi}}
\right)
\cdot
\\
&\stackrel{\ell_2 = (c_1 \log L)^{1 /\xi}}{=}&
C L^{\beta}
\cdot
\frac 
{L^{\beta}}{ (c_1 \log L)^{1 /\xi} }
\left(
e^{ - \frac {2 \rho_1\gamma}{(1 + \epsilon)} \log L}
+
e^{- q t_0 2^{\xi} c_1 \log L}
\right)
\\
& \le &
C \frac {L^{2 \beta}}{(c_1 \log L)^{1 /\xi}}
\left\{
\left(
\frac 1L
\right)^{\frac {2 \rho_1\gamma}{1 + \epsilon}}
+
\left(
\frac 1L
\right)^{q t_0 2^{\xi} c_1}
\right\}
\\
& \le &
C \frac {L^{2 \beta}}{(c_1 \log L)^{1 /\xi}}
\left(
\frac 1L
\right)^{\frac {2 \rho_1\gamma}{1 + \epsilon}}.
\eeq
In the last line,
we take
$c_1$
sufficiently large so that the second term in RHS is negligible.
We 
thus proved part (1) of Theorem \ref{thm:minami1} for 
$I_{p, k}$.
Next 
we take the sum  
$\sum_p$
with respect to 
$p = 1, 2, \cdots, N_L$ : 
\beq
\sum_p \sum_k
I_{p, k}
& \le &
C  L^{1 - \beta}
\frac {L^{2 \beta}}{(c_1 \log L)^{1 /\xi}}
\left(
\frac 1L
\right)^{\frac {2 \rho\gamma}{1 + \epsilon}}
\le
C \frac {L^{1 + \beta - \frac {2 \rho\gamma}{1 + \epsilon}}
}
{(c_1 \log L)^{1 /\xi}}.
\eeq
Since
$\rho_1\gamma > \dfrac {1+\beta}{2}$, 
we can take 
$\epsilon > 0$
small enough such that 
$\dfrac {2 \rho_1\gamma}{1 + \epsilon} 
> 1+\beta$.
Therefore 
\begin{equation}
\sum_p 
\sum_{k \ge 2}
I_{p, k} = o (1).
%\quad\cdots (I)
\label{firstminami}
\end{equation}
%

%%%%%%%%%%%%%%%%%%%%%%%%%%%%%%%%%%%
%
(2)
Estimate on 
$II_{p, k}$ : 
First of all, 
note that we may suppose each COL's
$\{ x_i(C_p) \}$
are away from each other in the distance of 
$3\ell_2$.
In fact, let 
\beq
\Omega_p^g
&:=&
\left\{
\omega \in \Omega
\, \middle| \,
%\forall  D^{(p)}_j
\mbox{  
at most one COL for all }
D^{(p)}_j
\right\}
=:
\Omega_{p, 1}^g
\cup
\Omega_{p, 2}^g
\\
\Omega_{p, 1}^g
&:=&
\left\{
\omega \in \Omega
\, \middle| \,
%\forall  D^{(p)}_j
\mbox{  
at most one COL for all }
D^{(p)}_j
\text{ and 
$|x_i(C_p) - x_{i'}(C_p)| > 3 \ell_2$
for any 
$i,i'$ }
\right\}
\\
\Omega_{p, 2}^g
&:=&
\Omega_{p}^g \setminus \Omega_{p, 1}^g
\eeq
Then
\beq
II_{p, k}
&=&
\Pp
\left(
\eta_p (I) \ge k
\; ; \,
\Omega^g_p \cap L(I, p, q)
\right)
\\
&=&
\Pp
\left(
\eta_p (I) \ge k
\; ; \,
\Omega^g_{p,1} \cap L(I, p, q)
\right)
+
\Pp
\left(
\eta_p (I) \ge k
\; ; \,
\Omega^g_{p,2} \cap L(I, p, q)
\right)
\\
&=:&
II_{p, k, 1}
+
II_{p, k, 2}.
\eeq
Then 
the argument in (1) shows that the second term 
$II_{p, k, 2}$
in RHS satisfies the same estimate as 
$I_{p, k}$ 
above.
In fact, 
we divide 
$C_p$
into boxes 
$F_j^{(p)}$'s 
of size 
$a := 3 \ell_2$
and also into boxes 
$G_j^{(p)}$'s 
of the same size but centered on the division points of 
$F_j^{(p)}$'s. 
Then on the event
$\Omega^g_{p,2} \cap L(I, p, q)$, 
there exists 
$j$
such that 
we can find two COL's
$x_i (C_p)$, $x_{i'} (C_p)$, 
on the same box 
$F_j^{(p)}$
(or 
$G_j^{(p)}$).
Let 
$\widetilde{F}_j^{(p)}$
(resp. $\widetilde{G}_j^{(p)}$)
be the box with the same center as 
$F_j^{(p)}$
(resp. $G_j^{(p)}$)
of size 
$6\ell_2$, 
Then 
$\widetilde{H}_j := H |_{ \widetilde{F}_j }$
or
$\widetilde{H'}_j := H |_{ \widetilde{G}_j }$
has two eigenvalues in 
$I$
so that the argument in (1) gives
\beq
\sum_{k \ge 2}
II_{p, k, 2}
& \le &
C \frac {L^{2 \beta}}{(c_1 \log L)^{1 /\xi}}
%\left\{
\left(
\frac 1L
\right)^{\frac {2 \rho_1\gamma}{1 + \epsilon}}.
\eeq
We 
thus showed the part (1) of Theorem \ref{thm:minami1} for 
$II_{p, k, 2}$.
Similarly as in (1), 
we take 
$\epsilon > 0$
small enough such that 
$\dfrac {2 \rho_1\gamma}{1 + \epsilon} 
> 1+\beta$
which yields
\begin{equation}
\sum_p 
\sum_{k \ge 2}
II_{p, k, 2} = o (1).
%\quad\cdots (II)
\label{secondminami}
\end{equation}
Thus 
it suffices to estimate 
$II_{p, k, 1}$.
For 
$\omega \in \Omega_{p,1}^g \cap L(I, p, q)$, 
We can find 
$k$ boxes 
$\widetilde{D}_{j_1}^{(p)}, \widetilde{D}_{j_2}^{(p)}, \cdots, \widetilde{D}_{j_k}^{(p)}$
on which the corresponding Hamiltonians 
$\widetilde{H}_{p, j_1}, \cdots, \widetilde{H}_{p, j_k}$
have eigenvalues in 
$I$.
Note that, 
since COL's satisfy
$|x_i(C_p) - x_{i'}(C_p)| > 3 \ell_2$, 
for all 
$i, i'$,  
these boxes are disjoint. 
On the other hand, since
$\ell_2 = (c_1 \log L)^{1 /\xi}$
implies 
$\dfrac {c_2}{L^{\gamma}}
=
c_2  
\cdot 
e^{ 
-\frac {\ell_2^{\xi}}{c_1}\gamma
}
=
c_2  \cdot 
e^{ 
-\frac {(2\ell_2)^{\xi}}{2^{\xi} \cdot c_1}\gamma
}$, 
Wegner's estimate
(Corollary \ref{corollary:wegner_v2})
yields
\beq
II_{p, k, 1}
\le
\Pp
\left(
\Omega_p^g
\right)
& \le &
\left(
\begin{array}{c}
\frac {L^{\beta}}{(c_1 \log L)^{1/\xi}} \\ 
k 
\end{array}
\right)
\Pp
\left(
d (E, \sigma (H_{\ell_2})) 
\le
e^{ 
-\frac {(2\ell_2)^{\xi}}{2^{\xi} \cdot c_1}
\gamma
}
%c_2  \cdot e^{ -\frac {\ell_2^{\xi}}{c_1}}
\right)^k
\\
& \stackrel{}{\le}&
\left(
\begin{array}{c}
\frac {L^{\beta}}{(c_1 \log L)^{1/\xi}} \\ 
k 
\end{array}
\right)
\left\{
C_W 
\left(
e^{ 
- \frac {\rho_1}{c_1}  \ell_2^{\xi} 
\gamma
}
+
e^{- q t_0 (2\ell_2)^{\xi}}
\right)
\right\}^k
\\
& \le &
\frac {
\left(
\frac {L^{\beta}}{(c_1 \log L)^{1/\xi}}
\right)^k
}
{
k!
}
\left\{
\left(
C_W 
L^{- \rho_1 \gamma}
%e^{ 
%- \frac {\rho}{c_1}  \ell_2^{\xi} 
%\gamma
%}
+
e^{- q t_0 2^{\xi} \cdot c_1 \log L}
\right)
\right\}^k.
\eeq
In the last line, 
we used
\beq
\left(
\begin{array}{c}
n \\ k
\end{array}
\right)
&=&
\frac {n (n-1) \cdots (n-k+1)}{k!}
\le
\frac {n^k}{k!}.
\eeq
Here 
we take 
$c_1>0$
large enough 
so that the second term is negligible, at the cost of increasing slightly the constant 
$C_W$. 
Now 
we take the sum with respect to 
$k = 2, \cdots, |C_p|$.
\beq
\sum_{k=2}^{L^{\beta}} II_{p, k, 1}
& \le &
\sum_{k \ge 2}
\frac {A^k}{k!}, 
\qquad
\text{where }
\quad
A :=
\frac {L^{\beta}}{(c_1 \log L)^{1/\xi}}
\cdot
C_W 
e^{ 
- \frac {\rho}{c_1}  \ell_2^{\xi}\gamma 
}
=
\frac {
C_W L^{ \beta - \rho_1 \gamma}
}
{
(c_1 \log L)^{1/\xi}
}
\\
&=&
A^2 e^{A}
=
\left(
\frac {
C_W L^{ \beta - \rho_1 \gamma}
}
{
(c_1 \log L)^{1/\xi}
}
\right)^2
e^{
\frac {
C_W L^{ \beta - \rho_1 \gamma}
}
{
(c_1 \log L)^{1/\xi}
}
}.
\eeq
Since
$\beta \le \rho_1 \gamma$, 
the last factor in RHS is 
$\mathcal{O}(1)$.
We 
thus showed the part (1) of Theorem \ref{thm:minami1} for 
$II_{p, k, 1}$.
Finally, 
we take the sum with respect to 
$p = 1, 2, \cdots, N_L$ : 
\beq
\sum_p 
\sum_{k \ge 2}
II_{p, k, 1}
& \le &
L^{1 - \beta}
\left(
\frac {
C_W L^{ \beta - \rho_1 \gamma}
}
{
(c_1 \log L)^{1/\xi}
}
\right)^2
e^{
\frac {
C_W L^{ \beta - \rho_1 \gamma}
}
{
(c_1 \log L)^{1/\xi}
}
}
\\
&\le&
C \frac { L^{1 - \beta + 2 \beta - 2 \rho_1\gamma}
}
{
(c_1 \log L)^{1/\xi}
}
\cdot
e^{
\frac {
C_W L^{ \beta - \rho_1 \gamma}
}
{
(c_1 \log L)^{1/\xi}
}
}.
\eeq
Therefore 
by the same argument as in (1), we have 
\begin{equation}
\sum_p 
\sum_{k \ge 2}
II_{p, k, 1} = o (1).
%\quad\cdots (III)
\label{thirdminami}
\end{equation}
By 
(\ref{firstminami}), 
(\ref{secondminami}), 
and
(\ref{thirdminami}), 
%(I), (II), (III), 
we arrive at the conclusion.
\end{proof}
\medskip

%%%%%%%%%%%%%%%%%%%%%%%%%%%%%%%%%%%%%%%%%%%%%%%%%%%%%%%%%%%%
\subsection{Decomposition into subsystems}
\text{}\\
In this subsection, 
we consider to decompose 
$\Lambda_L$
into small intervals
$\Lambda_{\ell}$
of size
$\ell$, 
and relate the eigenvalues of $H_{\Lambda_L}$ to those of the family $\{ H_{\Lambda_\ell (\gamma_j)}$. We show that, with good probability, number of COL's of eigenvectors of 
$H_{\Lambda_L}$ with eigenvalues in a given interval $I_\Lambda$ 
is at most one for each 
$\Lambda_{\ell}(\gamma_j)$ and, moreover, 
a COL of 
$H_{\Lambda_L}$
is in 
$\Lambda_{\ell}(\gamma_j)$
if and only if
$H_{\Lambda_{\ell}(\gamma_j)}$
has an eigenvalue in a slightly enlarged interval $\widetilde{I}_{\Lambda_L}$. 
We basically follow the strategy and notation in \cite[Theorem 1.1]{GKlo1}. 
As above, we work with length scales $\ell := L^{\beta}$ and $\ell'  := L^{\beta'}$,
for $0 < \beta' < \beta < 1$, and intervals $\Lambda_L$ and $\Lambda_\ell (\gamma)$, an interval centered at $\gamma \in \Z$.  
%%%%%%%%%%%%%%%%%%%%%%%%%%%%%%
%\quadset  
%\beq
%\Lambda_L 
%&:& 
%\text{an interval of size $L$}
%\\
%%
%\Lambda_{\ell} (\gamma) 
%&:& 
%\text{an interval of size 
%$\ell$
%centered at 
%$\gamma$
%}.
%\\
%%
%I_{\Lambda}
%&:=&
%N^{-1}
%\left(
%N(E_0) + \frac {I}{| \Lambda |}
%\right), 
%\quad
%I
%\text{ : bounded interval}
%\\
%%%%%%%%
%\ell 
%&:=& L^{\beta}, 
%\\
%%
%\ell' 
%&:=& L^{\beta'}, 
%\quad
%0 < \beta' < \beta < 1. 
%\quad
%%c_1, c_2 > 0. 
%\eeq
%%
%
%%
%
%

%
%%%%%%%%%
\begin{theorem}
\label{thm:decomposition} 
%{\bf Theorem 1.1}
%
Let $H_{\Lambda_L}$ be the AB RMP on an interval of length $L$ and satisfying Assumption \ref{assump:dos1}, and let
%We assume 
%Assumption \ref{assump:dos1}, and let 
%
\bel{eq:unfolded_int1}
I_{\Lambda_L}
:=
N^{-1}
\left(
N(E_0) + 
\frac {[-K, K]}{L}
\right), 
\quad
K > 0.
\ee
Consider the following decomposition of $\Lambda_L$: 
$\Lambda_L = \bigcup_j \Lambda_{\ell} (\gamma_j) \cup \Upsilon$
where 
$\{ 
\Lambda_{\ell}(\gamma_j)
\}_j$
are disjoint and 
$\Upsilon$
is the set 
$\Lambda_L \setminus \bigcup_j \Lambda_{\ell}(\gamma_j)$
enlarged by a length 
$\ell'$, 
such that
\beq
&(i)& \quad
d
\left(
\Lambda_{\ell}(\gamma_j), \Lambda_{\ell} (\gamma_k)
\right)
\ge
\ell', 
\quad
j \ne k
\\
&(ii)& \quad
d
\left(
\Lambda_{\ell}(\gamma_j), \partial\Lambda_L
\right)
\ge
\ell'
\\
&(iii)& \quad
| \Upsilon |
\le
C  L
\left( \frac { \ell' }{\ell} \right).
\eeq
Then,  for $L$
large enough, we can find constants 
$\alpha > 0$, 
$0 < \beta' < \beta < 1$,
and a set of configurations 
$\mathcal{Z}_{\Lambda_L}  \subset  \Omega$,
such that 
%
%\beq
%&(i)& \quad
%d
%\Pp( \mathcal{ Z}_{\Lambda_L} )
%\ge 
%1 - L^{-\alpha}
%\\
%%
%&(ii)& \quad
%d
%\left(
%\Lambda_{\ell}(\gamma_j), \partial\Lambda_L
%\right)
%\ge
%\ell'
%\\
%%
%&(iii)& \quad
%| \Upsilon |
%\le
%C  L
%\left( \frac { \ell' }{\ell} \right).
%\eeq
%\noindent
\begin{enumerate}
%$(1)\quad
\item $\Pp( \mathcal{ Z}_{\Lambda_L} )
\ge 
1 - L^{-\alpha}$, 
\medskip 
%
%$(2)\;$
%
\item 
For 
$\omega \in \mathcal{ Z}_{\Lambda_L}$, 
the COL's associated to the eigenvalues
%$E_j (\omega, \Lambda_L) \in I_{\Lambda}$
of 
$H_{\omega} (\Lambda_L)$
in 
$I_{\Lambda_L}$
are contained by 
$\bigcup_j \Lambda_{\ell} (\gamma_j)$.
%\eeq
%
Each interval $\Lambda_{\ell}(\gamma_j)$
satisfies 
\\
(i)
$H_{\omega}(\Lambda_{\ell}(\gamma_j))$
has at most one eigenvalue 
$E_j (\Lambda_{\ell} (\gamma_j))$
in 
$I_{\Lambda_L}$, 
\\
(ii)
$\Lambda_{\ell} (\gamma_j)$
contains at most one COL 
$x_{k_j}(\Lambda_L)$
corresponding to an eigenvalue 
$E_{k_j}(\Lambda_L)$
of 
$H_{\omega}(\Lambda_L)$.
\\
(iii)
$\Lambda_{\ell} (\gamma_j)$
contains a COL 
$x_{k_j}(\Lambda_L)$
of 
$H_{\omega} (\Lambda_L)$
if and only if 
$\sigma (H_{\omega} (\Lambda_{\ell}(\gamma_j))) \cap I_{\Lambda_L}
\ne
\emptyset$, 
in which case, denoting 
$\sigma (H_{\omega} (\Lambda_{\ell}(\gamma_j))) \cap I_{\Lambda_L}
=
\{
E_j (\Lambda_{\ell}(\gamma_j))
\}$, 
we have 
\beq
&&
| 
E_{k_j} (\Lambda_L)
-
E_j (\Lambda_{\ell}(\gamma_j))
|
\le
e^{- (\ell')^{\xi}}
\\
&&
d 
\left(
x_{k_j}(\Lambda_L), 
\Lambda_L \setminus \Lambda_{\ell}(\gamma_j)
\right)
\ge
\ell'.
\eeq
\end{enumerate}
\end{theorem}
%%%
%
%
\begin{remark}
\label{rmk:general_end}
In Theorem \ref{thm:decomposition}, 
$\Lambda_L$
is defined in Definition \ref{defn:subsets_polymer} so that the right end of 
$\Lambda_L$
is in a polymer node 
$p_n \in \mathcal{P}$.
However, 
by adjusting the definition of 
$\Upsilon$, 
one can see that 
the statement of Theorem \ref{thm:decomposition} is also valid for any box of the form : 
$\Lambda_L 
\cup
\{ p_n + 1, \cdots, p_n + L_{\max} \}$, 
where 
$L_{\max} := \max \{ L_+ , L_- \}$.
\end{remark}
\noindent
For proof, we need the following definition. 
%

%
%%%%%%%%%
\begin{defn}
Let 
$\Lambda, \Lambda'$
be finite intervals with 
$\Lambda' \subset \Lambda$, 
and let 
$E$
be an eigenvalue of 
$H_{\Lambda}$.
We say that 
{\bf 
$E$
is centered in 
$\Lambda'$
}
if and only if the  
COL
$x_{\varphi}$
of the corresponding eigenvector 
$\varphi$
of
$E$
satisfies 
$x_{\varphi} \in \Lambda'$.
\end{defn}
%%%
%

%
\begin{proof}
Let 
$\mathcal{ S}_{\ell, L}$, 
$\widetilde{\mathcal{ S}}_{\ell, L}$
be the set of intervals such that 
\beq
\mathcal{ S}_{\ell, L}
&:=&
\left\{
\Lambda_{\ell} (\gamma_j)
\, \middle| \,
\mbox{ $\exists$ 
more than two eigenvalues of 
$H_{\Lambda_L}$ 
in
$I_{\Lambda_L}$
centered in 
$\Lambda_{\ell - \ell'} (\gamma_j)$ }
\right\}, 
\\
\widetilde{\mathcal{ S}}_{\ell, L}
& := &
\left\{
\Lambda_{\ell} (\gamma_j)
\, \middle| \,
\mbox{ $\exists$ 
more than two eigenvalues of 
$H_{\Lambda_{\ell}(\gamma_j)}$ 
in
$I'_{\Lambda_L}$
%centered in 
%$\Lambda_{\ell - \ell'} (\gamma_j)$ 
}
\right\}, 
\eeq
where 
$K$
is replaced by some 
$K' >K$
in \eqref{eq:unfolded_int1}, the definition of 
$I'_{\Lambda_L}$.
Then, by the localization estimate in part (2) of Lemma \ref{lemma:lemma1.1}, if 
$\omega \in \mathcal{ U}_{\Lambda_L}$, we have  
$\mathcal{ S}_{\ell, L} 
\subset 
\widetilde{\mathcal{ S}}_{\ell, L}$.
Next, we apply Theorem \ref{thm:minami1}, 
where we take $\gamma := \rho_2$, and assume 
\begin{equation}
\beta \le \rho_1  \rho_2.  
%\quad\cdots (0)
\label{zerothcondition}
\end{equation}
Then, for any  $\epsilon > 0$,  
\begin{eqnarray}
\Pp
\left(
\sharp \mathcal{ S}_{\ell, L} \ge 1
; \mathcal{ U}_{\Lambda_L}
\right)
& \le &
\Pp
\left(
\sharp \widetilde{\mathcal{ S}}_{\ell, L} \ge 1
\right)
\le
%
%&\le&
\E
\left[
\sharp \widetilde{\mathcal{ S}}_{\ell, L}
\right]
\nonumber
\\
&=&
\sharp \left\{ 
\Lambda_{\ell}(\gamma_j)
\right\}
\Pp
\left(
\sharp 
\left\{
\mbox{eigenvalues of $H_{\Lambda_{\ell}}(\gamma_1)$
in $I'_{\Lambda_L}$}  \right\} \ge 2  \right)
\nonumber
\\
& \le &  C  | \Lambda |^{1 - \beta}
\left(
\frac {L^{\beta}}{L^{ \rho_1 \rho_2/(1 + \epsilon)}} \right)^2
=
C  L^{- \delta} 
%\quad \cdots (*)
\label{star}
\end{eqnarray}
where 
$\delta
:=
-1 - \beta + \dfrac {2\rho_1  \rho_2}{(1 + \epsilon)}$.
We can find
$\epsilon > 0$
such that 
$\delta > 0$, 
if and only if 
\begin{equation}
\rho_1  \rho_2
>
\frac {1 + \beta}{2}
%\quad\cdots (1)
\label{firstcondition}
\end{equation}

We recall that 
$\Upsilon$
is the set 
$\Lambda_L \setminus \bigcup_j \Lambda_{\ell}(\gamma_j)$
enlarged by a length 
$\ell'$.
We decompose
$\Upsilon$
into 
small intervals of size
$\widetilde{\ell} := (c \log L)^{1 / \xi}$
where 
$0 < \xi \le 1$
is the constant appearing in SULE estimate (Lemma \ref{lemma:lemma1.1}), and for some $c > 0$, and write
$\Upsilon = \bigcup_k \Lambda_{\widetilde{\ell}}(k)$.
Then by localization estimate again, 
\begin{eqnarray}
&&
\Pp
\left(
\mbox{ 
$H_{\Lambda_L}$
has eigenvalues in 
$I_{\Lambda_L}$ 
centered in 
$\Upsilon$ }
; \mathcal{ U}_{\Lambda_L}
\right)
\nonumber
\\
& \le &
\sum_k
\Pp
\left(
\mbox{ 
$H_{\Lambda_L}$
has eigenvalues in 
$I_{\Lambda_L}$ 
centered in 
$\Lambda_{\widetilde{\ell}}(k)$ }
; \mathcal{ U}_{\Lambda_L}
\right)
\nonumber
\\
& \le &
\sum_k
\Pp
\left(
\mbox{ 
$H_{\Lambda_{\widetilde{\ell}} (k) }$
has an eigenvalue in 
$I'_{\Lambda_L}$
%centered in 
%$\Lambda_{\widetilde{\ell}}(j)$ 
}
\right). 
\label{starstar}
%\quad \cdots (**)
\end{eqnarray}
Note that, by assumption,  
$|I_{\Lambda_L}| 
\le
C  \left(
\dfrac CL
\right)^{
\rho_2
}
=
\mathcal{ O}
\left(
e^{ 
- \frac {
\widetilde{\ell}^{\xi}
}
{c}
\cdot
\rho_2
}
\right)
$.
By 
Wegner's estimate(Corollary \ref{corollary:wegner_v2}), we have
\beq
\Pp
\left(
\mbox{ 
$H_{\Lambda_{\widetilde{\ell}} (k) }$
has an eigenvalue in 
$I'_{\Lambda_L}$
%centered in 
%$\Lambda_{\widetilde{\ell}}(j)$ 
}
\right)
\le
C_W
\left(
e^{ 
- \frac {
\widetilde{\ell}^{\xi}
}
{c}
\cdot
\rho_1 \cdot \rho_2
}
+
e^{ - q \widetilde{\ell}^{\nu} }
\right)
\le
C  
L^{
- \rho_1 \cdot \rho_2
}
\eeq
where we take
$c > 0$
large enough 
so that the second term
$e^{ - q \widetilde{\ell}^{\xi} }$
in RHS is negligible.
Then
\beq
&&
\Pp
\left(
\mbox{ 
$H_{\Lambda_L}$
has eigenvalues in 
$I_{\Lambda_L}$
centered in 
$\Upsilon$ }
; \mathcal{ U}_{\Lambda_L}
\right)
\\
&\le&
\sharp 
\left\{
\Lambda_{\widetilde{\ell}}(k) \subset \Upsilon
%\mbox{   }
\right\}
C  L^{- \rho_1 \cdot \rho_2
}
%e^{- c \log L \cdot \rho_1}
\\
& \le &
C  
L^{
1 - \beta + \beta' - \rho_1 \cdot \rho_2
}
\\
&\le&   C  
L^{- \alpha}
\quad
\text{for some }
\alpha.
\eeq
Here we have used
$
\#
\{
\Lambda_{ \widetilde{\ell} } (k) \subset \Upsilon
\}
\le 
\dfrac {
L^{1 - \beta} \cdot L^{\beta'}
}
{
\widetilde{\ell}
}
$.
We can find some 
$\beta'$
such that 
$0 < \beta' < \beta$
and
$\alpha > 0$, 
if and only if 
\begin{equation}
\rho_1 \cdot \rho_2 
>
1 - \beta.
%\quad\cdots (2)
\label{secondcondition}
\end{equation}
Since
\beq
\min_{\beta \in (0,1)}
\max
\left\{
\frac {1+ \beta}{2}, 1 - \beta, \beta
\right\}
=
\frac 23, 
\eeq
we can find 
$\beta \in (0,1)$
such that all the conditions 
(\ref{zerothcondition}), (\ref{firstcondition}), 
and
(\ref{secondcondition})
are satisfied if and only if 
$\rho_1 \cdot \rho_2
>
\dfrac 23$.
Now, let 
\beq
\mathcal{ Z}_{\Lambda}
&:=&
\bigcap_j
\left\{
\sharp \{
\mbox{ eigenvalues of  }
H_{\Lambda_{\ell}(\gamma_j)}
\} 
\le 1
\right\}
\cap
\bigcap_k
\left\{
\mbox{ 
$H_{\Lambda_{ \widetilde{\ell} }}(k)$
does not have eigenvalue }
\right\}.
%\cap
%{\cal U}_{\Lambda_L}.
\eeq
Then 
the above argument and SULE estimate in Lemma \ref{lemma:lemma1.1}
yield
$\Pp (\mathcal{ Z}_{\Lambda})
\ge
1 - L^{-\alpha}$, 
for some 
$\alpha> 0$.
Moreover, if 
$\omega \in \mathcal{ Z}_{\Lambda}$, 
(\ref{starstar})
implies that 
eigenvalues of 
$H_{\Lambda_L}$
in
$I_{\Lambda_L}$
are centered in 
$\bigcup_j \Lambda_{\ell - \ell'}( \gamma_j )$
and 
(\ref{star})
implies there are at most one eigenvalue
$E_{k_j}$ 
of 
$H_{\Lambda_L}$
centered on each 
$\Lambda_{\ell - \ell'}(\gamma_j)$'s.
And since 
(\ref{star})
also implies 
$H_{\Lambda_{\ell}(\gamma_j)}$
has at most one eigenvalue 
$E'_j$, 
if there is any, it satisfies
$|E_{k_j} - E'_j| \le e^{- (\ell')^{\xi}}$. 
\end{proof}
%
%%%%%%%%%%%%%%%%%%%%%%%%%%%%%%%%%%%%
\subsection{Completion of the proof of Theorem \ref{thm:les_loc1}: Poisson limit theorem}
In this subsection, we prove the following statement which completes the proof of Theorem \ref{thm:les_loc1}: the convergence of 
$\xi_L$
to a Poisson process for 
$E_0 \in \Sigma\setminus \mathcal{C}$.
\begin{theorem}
\label{thm:Poisson}
%\mbox{}
%
Let 
$I_1, \cdots, I_p \subset \R$
be disjoint bounded intervals, and let 
$k_1, \cdots, k_p \in \N \cup \{ 0 \}$.
Then, under the hypotheses of Theorem \ref{thm:les_loc1}, 
\beq
\Pp
\left(
\bigcap_{m=1}^p
\left\{
\omega 
\, \middle| \,
\xi_L (I_{m}) = k_{m}
\right\}
\right)
\stackrel{L \to \infty}{\to}
\prod_{m=1}^p
\frac {
| I_{m} |^{k_{m}}
}
{
k_{m} !
}
e^{ - | I_{m} |}.
\eeq
\end{theorem}
\begin{remark}
\label{rmk:any_box}
By definition, 
$\xi_L$
is composed of the eigenvalues of 
$H_{\Lambda_L}$, 
where the right end of 
$\Lambda_L$
is a polymer node 
$p_n \in \mathcal{P}$.
However, 
by remark \ref{rmk:general_end}, 
one can see that 
$\Lambda_L$
may be replaced by any box of the form
$\Lambda_L 
\cup
\{ p_n +1, \cdots, p_n + L_{\max} \}$.
\end{remark}

As is already mentioned, 
we prove Theorem \ref{thm:Poisson} by approximating  
$H_{\Lambda_L}$
by the direct sum of small subsystems 
$\oplus_j H_{ \Lambda_{\ell}(\gamma_j) }$
and apply the method of proof of the Poisson limit theorem.
Since 
this decomposition has been done in Theorem \ref{thm:decomposition}, 
it suffices to estimate the probability of rare events.
In order to do that, we 
begin by setting up the notation as in \cite{GKlo1}. 
For an interval
$J \subset \R$, we define
\beq
N(J, \ell, \ell')
&:=&
\#
\left\{\text{eigenvalues of }
H_{\Lambda_{\ell}}
\text{in $J$ centered in }
\Lambda_{\ell - \ell'}
\right\},
\\
I_{\Lambda_L}
&:=&
N^{-1}
\left(
N(E_0) + \frac {I}{L}
\right), 
\quad
I : 
\text{bounded interval},
\eeq
and, letting ${\bf{1}}_A$ denote the characteristic function on a set $A$,  the Bernoulli random variables
\beq
X_{\Lambda_{\ell}, I}
&:=&
{\bf{1}}_{
\left\{
H_{\Lambda_{\ell}}
\text{ has an eigenvalue in }
I_{\Lambda_L}
\text{ centered in }
\Lambda_{\ell - \ell'}
\right\}}
\\
&=&
{\bf{1}}_{\left\{
N(I_{\Lambda}, \ell, \ell') \ge 1
\right\}}.
\eeq
We always assume that 
$J$
is contained in the localized regime of the spectrum such that SULE estimate in Lemma \ref{lemma:lemma1.1} is valid.
Then, by definition of IDS and by the localization estimate, we expect that 
$\E[ N(J, \ell, \ell')]$
is close to 
$N(J) |\Lambda_{\ell}|$
for 
$\ell$
large enough, where $N(J)$ is the DOS measure of $J$.
Next, 
by Minami's estimate, Theorem \ref{thm:minami1}, we expect that 
$\Pp(X_{\Lambda_{\ell}, I}=1)$
is close to 
$\E[ N(I_{\Lambda_L}, \ell, \ell')]$.
Lemma \ref{lemma:A}
below proves that these observations are valid, and will be important in the proof of Lemma \ref{lemma:B} to estimate the probability of the rare events, that is,  
$\{ X_{\Lambda_{\ell}, I} = 1 \}$, necessary
to prove the Poisson limit theorem.
%

%
%%%%%%%%%
%\begin{itembox}[l]
\begin{lemma}
\label{lemma:A}
For any energy interval $J \subset \R$, we have
\beq
&(1)& \quad
\Bigl| 
\E [ N(J, \ell, \ell') ]
-
N(J) |\Lambda_{\ell}|
\Bigr|
\le 
C
N(J)
|\Lambda_{\ell}|
\frac {\ell'}{\ell}
+
C
\ell
e^{- \rho_1 (\ell')^{\xi}}.
\\
&(2)& \quad
|
\Pp(X_{\Lambda_{\ell}, I} = 1) - N(I_{\Lambda}) |\Lambda_{\ell}|
|
\le
\mathcal{ O}
\left(
L^{
2 
(
\beta - \rho_1 \cdot \rho_2/(1 + \epsilon)
)
}
\right)
+
C
N(I_{\Lambda})| \Lambda_{\ell} | \ell' \ell^{-1}
+
C
\ell
e^{- \rho_1 (\ell')^{\xi}}.
\eeq
\end{lemma}
%\end{itembox}\\
%%%
%

%
\begin{proof}
(1)
Let 
$\mathcal{ V}_{\Lambda_{\ell}} \subset \Omega$
be the set of configurations in Lemma \ref{lemma:lemma1.1}, part (2).
For 
$\omega \in \mathcal{ V}_{\Lambda_{\ell}}^c$, 
we use the fact that 
$\# \sigma (H_{\omega}(\Lambda_{\ell})) 
\le 
| \Lambda_{\ell} |$, 
so that 
\beq
\E[
{\bf{1}}_{ \mathcal{ V}^c_{\Lambda_{\ell}} }
N(J, \ell, \ell')
]
\le
\# \sigma (H_{\Lambda_{\ell}}) 
\Pp( \mathcal{ V}_{\Lambda_{\ell}}^c )
\le
| \Lambda_{\ell} | e^{- \ell^{\xi}}.
\eeq
For 
$\omega \in \mathcal{ V}_{\Lambda_{\ell}}$, 
we use the localization estimate so that 
\begin{equation}
tr 
\left( {\bf{1}}_{ \Lambda_{\ell - 2 \ell'} }
{\bf{1}}_{J_-} (H_{\omega})
\right)
+
\mathcal{ O}( \ell e^{- (\ell')^{\xi}} )
\le
N(J, \ell, \ell')
\le
tr 
\left(
{\bf{1}}_{ \Lambda_{\ell} }
{\bf{1}}_{J_+} (H_{\omega})
\right)
+
\mathcal{ O}( \ell e^{- (\ell')^{\xi}} )
\label{sandwich}
\end{equation}
where 
$H_{\omega}$
is the infinite volume operator, and
\beq
J_+
&:=&
J + [ - e^{-(\ell')^{\xi}}, e^{- (\ell')^{\xi}} ]
\\
J_-
&:=&
J \setminus
(J^c + [ - e^{-(\ell')^{\xi}}, e^{- (\ell')^{\xi}} ]).
\eeq
are the enlarged (respectively, compressed) versions of 
$J$
to absorb the localization error.
We take 
the expectation on both sides of (\ref{sandwich}). 
Then, the expectation of the  
RHS of (\ref{sandwich}) is equal to 
\beq
\E[ RHS \text{ of } (\ref{sandwich}) ]
&=&
\E
\left[ 
{\bf{1}}_{ \mathcal{ V}_{\Lambda_{\ell}} } 
tr 
\left(
{\bf{1}}_{ \Lambda_{\ell} }
{\bf{1}}_{J_+} (H_{\omega})
\right)
\right]
+
\mathcal{ O}( \ell e^{- (\ell')^{\xi}} )
\\
&=&
\E
\left[ 
%1_{ {\cal V}_{\Lambda_{\ell}} } 
tr 
\left(
{\bf{1}}_{ \Lambda_{\ell} }
{\bf{1}}_{J_+} (H_{\omega})
\right)
\right]
+
\mathcal{ O}( \ell e^{- (\ell')^{\xi}} )
\\
&=&
N(J_+) | \Lambda_{\ell} |
+
\mathcal{ O}( \ell e^{- (\ell')^{\xi}} )
\\
&=&
N(J) | \Lambda_{\ell} |
+
\mathcal{ O}( \ell e^{- \rho_1(\ell')^{\xi}} ).
\eeq
We have 
used the fact that the IDS
$N$
is H\"older continuous with exponent 
$\rho_1$.
The expectation of the LHS of (\ref{sandwich}) is equal to 
\beq
\E[ LHS \text{ of } (\ref{sandwich}) ]
&=&
\E
\left[ {\bf{1}}_{ \mathcal{ V}_{\Lambda_{\ell}} } 
tr 
\left(
{\bf{1}}_{ \Lambda_{\ell - 2 \ell'} }
{\bf{1}}_{J_-} (H_{\omega})
\right)
\right]
+
\mathcal{ O}( \ell e^{- (\ell')^{\xi}} )
\\
&=&
\E
\left[ 
%1_{ {\cal V}_{\Lambda_{\ell}} } 
tr 
\left(
{\bf{1}}_{ \Lambda_{\ell- 2 \ell'} }
{\bf{1}}_{J_-} (H_{\omega})
\right)
\right]
+
\mathcal{ O}( \ell e^{- (\ell')^{\xi}} )
\\
&=&
N(J_-) | \Lambda_{\ell- 2 \ell'} |
+
\mathcal{ O}( \ell e^{- (\ell')^{\xi}} )
\\
&=&
N(J) | \Lambda_{\ell} |
\left(
1 - 
\mathcal{O}\left(
\frac {\ell'}{\ell}
\right)
\right)
+
\mathcal{ O}( \ell e^{- \rho_1(\ell')^{\xi}} )
\eeq
completing the proof of part (1) of Lemma \ref{lemma:A}.
\\
(2) We let  $\eta$ denote the number of eigenvalues of $H_{\Lambda_\ell}$ in $I_{\Lambda_L}$:
\beq
\eta
&:=&
\sharp 
\left\{
\mbox{ eigenvalues of 
$H_{\Lambda_{\ell_L}}$ 
in 
$I_{\Lambda}$ }
\right\}.
\eeq
Since $| I_{\Lambda_L} |
\le C  L^{- \rho_2}$,  part (1) of Theorem \ref{thm:minami1}  
%noting that 
%$| I_{\Lambda_L} |
%\le C  L^{- \rho_2}$, 
%
yields
\beq
\sum_{k \ge 2}
\Pp( \eta \ge k )
\le  C \left(
\frac {
L^{\beta}
}
{
L^{ \rho_1 \cdot \rho_2/(1 + \epsilon)}
}
\right)^2
=
C 
L^{
2 
\bigl(
\beta - \rho_1 \cdot \rho_2/(1 + \epsilon)
\bigr)
}.
\eeq
Then, noting that 
$X_{\Lambda_{\ell}, I} = 
{\bf{1}}(
N(I_{\Lambda_L}, \ell, \ell') \ge 1
)
$, 
\beq
\E
[ N(I_{\Lambda_L}, \ell, \ell') ] 
-
\Pp
(X_{\Lambda_{\ell}, I} = 1)
&=&
\sum_{k \ge 1}
\Pp
\left(
N(I_{\Lambda_L}, \ell, \ell') \ge k
\right)
-
\Pp
\left(
N(I_{\Lambda_L}, \ell, \ell') \ge 1
\right)
\\
&\le&
\sum_{k \ge 2}
\Pp
(\eta \ge k)
\\
& \le &
C L^{
2 
\bigl(
\beta - \rho_1 \cdot \rho_2/(1 + \epsilon)
\bigr)
}.
\eeq
Now 
it suffices to apply part (1) of Lemma \ref{lemma:A}(1) with 
$J = I_{\Lambda_L}$.
\end{proof}
%

%
%%%%%%%%%%%%%%%%%%%%%%%%%%%%%%%%%%%%
%
%%%%%%%%%%%%%%%%%%%%%%%%%%%%%%%%%%%%%%%%%%%%%%%%%%%%%%%%%%%%%%%%%%%%%%%%%%%%%%
To prove Theorem \ref{thm:Poisson}, 
let 
$I_1, \cdots, I_p \subset {\bf R}$
be disjoint bounded energy intervals, and let 
\beq
I_j^+
&:=&
I_j + [ - e^{- (\ell')^{\xi}}, e^{- (\ell')^{\xi}}]
\\
I_j^-
&:=&
I_j \setminus 
(I_j^c + [ - e^{- (\ell')^{\xi}}, e^{- (\ell')^{\xi}}]), 
\quad
j = 1, 2, \cdots, p.
\eeq
their enlarged and compressed counterparts.
Then, we have 
$I_j^- \subset I_j \subset I_j^+$, and, for $L$ large, 
$I_j^+ \cap I_k^+ = \emptyset$,
for 
$j \ne k$.
Moreover, we denote the rescaled intervals by $\widetilde{J}^\pm_m$, and their union $\widetilde{J}_{\pm}$, by
\beq
%\widetilde{J}_{\pm}
%&:=&
%\bigcup_{m=1}^p
%\widetilde{J}_{m}^{\pm}, 
%\quad
%\text{ where }
%\quad
\widetilde{J}_{m}^{\pm}
 & :=  &
\bigcup_{m=1}^p
N^{-1}
\left(
N(E_0) + 
\frac {I_m^{\pm}}{L}
\right), 
\quad
%\\
%
\text{and}
\quad
\widetilde{J}_{\pm}
:= 
\bigcup_{m=1}^p
\widetilde{J}_{m}^{\pm} . 
\quad
\eeq
We denote the number of subintervals 
$\Lambda_{\ell} (\gamma_j)$ by $\widetilde{N} =
\mathcal{O} \left(L^{1 - \beta}\right)(1 + o(1))$.
%%&:=&
%\#
%\Lambda_{\ell}(\gamma_j)
%=
%\mathcal{O}
%\left(
%L^{1 - \beta}
%\right)$. 
%\beq
%\widetilde{N}
%&:=&
%\#
%\Lambda_{\ell}(\gamma_j)
%=
%\mathcal{O}
%\left(
%L^{1 - \beta}
%\right)
%\eeq
%%
%be the union of intervals in question, and 
%denote the number of 
%$\Lambda_{\ell} (\gamma_j)$'s. 
%respectively.\\
%

To prove Theorem \ref{thm:Poisson}, 
we need to estimate the probability of the rare event 
${\Pp}
\left(
X_{
\Lambda_{\ell}, I_m^{\pm}
} = 1
\right)$
which is done in Lemma \ref{lemma:B} below. 
%

%
%%%%%%%%%
\begin{lemma}
\label{lemma:B}\mbox{}\\
(1)
Suppose that
$\rho_1  \rho_2 > \dfrac 12$.
Then
\begin{equation}
\Pp
\left(
\sum_{m=1}^p
X_{ \Lambda_{\ell}, I_m^{\pm} } = 0
\right)
=
1 - 
\frac {
1 - \mathcal{ O}( L^{- \delta} )
}
{
\tilde{N}
}
\sum_{m=1}^p |I_m^{\pm}|,
\label{eq:lemmaB1}
\end{equation}
(2)
Suppose that
$\rho_1  \rho_2 > \dfrac 12$.
Then we have 
\beq
\Pp
\left(
X_{
\Lambda_{\ell}, I_m^{\pm}
} = 1
\right)
&=&
|I_m^{\pm}| 
\widetilde{N}^{-1}
%L^{\beta -1}
(1 + \mathcal{ O}(L^{-\delta})),
\eeq
where we take 
$\beta, \beta' > 0$
sufficiently small such that  
$0 < \beta' < \beta$ 
and
\beq
\delta
&:=&
\min
\left\{
2 \rho_1  \rho_2- (\beta +1), 
\,
\beta - \beta'
\right\} 
\eeq
is positive.
\end{lemma}
%%%
%

%
\begin{proof} 
(1) We first estimate  
\bel{eq:error1}
\mathcal{E}_\pm := {\Pp}  \left\{  \sum_{l=1}^p   X_{\Lambda_{\ell}(\gamma_1), I_l^\pm} \ge 1
\right\} - {\E}  \left\{  {{N}}\left(  N(\widetilde{J}_\pm), \ell, \ell'  \right)  \right\} .
\ee
We define the counting function
%\bea\label{eq:y_defn}
$$
Y_l^\pm
:=
\sharp 
\left\{ \mbox{eigenvalues of $H_{\Lambda_{\ell}}$ in 
${{N}}^{-1}  \left(  N(E_0) + \frac {I_l^\pm}{L}
\right)$
with COL $\in \Lambda_{\ell - \ell'}$  }
\right\}  .
$$
%\eea
Since $(\sum_{l} Y_l^\pm \ge 1)= (
 \sum_l X_{\Lambda_{\ell}(\gamma_1), I_l^\pm}  \ge 1)$,
it follows that 
\bel{eq:error2}
\mathcal{E}_\pm  = \sum_{k \ge 2} 
{\Pp} \left\{  \sum_l Y_l^\pm  \ge k   \right\}  ,
\ee
Noting that 
$$
\eta_{\Lambda_\ell (\gamma(1))} (I_l^\pm) \geq Y_\ell ,
$$
it follows from \eqref{eq:error2}, together with the Minami estimate, that  
\bel{eq:error3}
 {\Pp}  \left\{  \sum_{l=1}^p   X_{\Lambda_{\ell}(\gamma_1), I_l^\pm} \ge 1
\right\} =  {\E}  \left\{  {{N}}\left(  N(\widetilde{J}_\pm), \ell, \ell'  \right)  \right\} 
+ \mathcal{ O}
\left(
\left(
\frac 
{L^{\beta}}
{
L^{
\rho_1  \rho_2/(1 + \epsilon)
%\rho' / (1 + \overline{\rho})
}
}
\right)^2
\right).
\ee
%%%%%%%%%%%%%%%%%%%%%%
Estimating the expectation of the counting function in  \eqref{eq:error3} by part (2) of Lemma \ref{lemma:A}, we then have
\bea\label{eq:error4}
\Pp
\left(
\sum_{m} 
X_{\Lambda_{\ell}(\gamma_1), I_m^{\pm}} \ge 1
\right)
%&\stackrel{A(2)}{=}&
&=&
N(\widetilde{J}_{\pm}) | \Lambda_{\ell} |
\left(
1 + 
\mathcal{O}
\left(
\frac {\ell'}{\ell}
\right)
\right)
+
C\ell e^{- (\ell')^{\xi}}
+
\mathcal{ O}
\left(
\left(
\frac 
{L^{\beta}}
{
L^{
\rho_1  \rho_2/(1 + \epsilon)
%\rho' / (1 + \overline{\rho})
}
}
\right)^2
\right).   \nonumber  \\
 & &
\eea
By definition of 
$\widetilde{J}_{\pm}$, 
the corresponding IDS is given by 
\bel{eq:error5}
N(\widetilde{J}_{\pm}) | \Lambda_{\ell} |
=
\frac {
\sum_m |I_m^{\pm}|
}
{
| \Lambda |
}
| \Lambda_{\ell} |
=
\left(
\sum_m |I_m^{\pm}|
\right)
\widetilde{N}^{-1}
(1 + \mathcal{ O}(L^{-(\beta - \beta')}).
\ee
We estimate the error term in 
(\ref{eq:lemmaB1}) from \eqref{eq:error4} and \eqref{eq:error5}.
This error is decomposed as  
\beq
\mbox{ error }
&:=&
\Pp
\left(
\sum_{m} 
X_{\Lambda_{\ell}(\gamma_1), I_m^{\pm}} \ge 1
\right)
-
\left(
\sum_m |I_m^{\pm}|
\right)
\widetilde{N}^{-1}
\\
&=&
\left(
\sum_m |I_m^{\pm}|
\right)
\widetilde{N}^{-1}
\mathcal{ O}(L^{-(\beta - \beta')})
+
C
N(\widetilde{J}_{\pm}) | \Lambda_{\ell} |
\mathcal{O}
\left(
\frac {\ell'}{\ell}
\right)
+
\ell e^{- (\ell')^{\xi}}
+
\mathcal{ O}
\left(
\left(
\frac 
{L^{\beta}}
{
L^{
\rho_1  \rho_2/(1 + \epsilon)
}
}
\right)^2
\right)
\\
&=:&
I + II + III+ IV.
\eeq
We 
estimate
$I, \cdots, IV$
separately to show they are in the order of 
$\mathcal{ O}( L^{- \delta - (1 - \beta)})$.
The first and second terms
$I$, $II$
satisfy
\beq
I
&=&
\left(
\sum_m |I_m^{\pm}|
\right)
\widetilde{N}^{-1}
\mathcal{ O}(L^{-(\beta - \beta')})
\le  C
L^{\beta-1}\cdot L^{-(\beta - \beta')}
= C
L^{\beta'-1}
\\
II
&=&
C
N(\widetilde{J}_{\pm}) | \Lambda_{\ell} |
\mathcal{O}
\left(
\frac {\ell'}{\ell}
\right)
\le C
\frac {1}{| \Lambda_L |}
\cdot
|\Lambda_{\ell}|
\cdot
\frac {L^{\beta'}}{L^{\beta}}
=  C 
\frac {1}{L}
\cdot
L^{\beta}
\cdot
\frac {L^{\beta'}}{L^{\beta}}
=
L^{\beta' - 1}.
\eeq
To have 
$I, II = \mathcal{ O}(L^{\beta - 1 - \delta})$
for some 
$\delta > 0$, 
we need 
$\beta - \beta' > 0$.
The third term 
$III$
is exponentially small : 
\beq
III
&=&
\ell
e^{- (\ell')^{\xi}}
=
L^{\beta}
e^{- L^{\beta \xi}}
\eeq
which is negligible.
As for the fourth term
$IV$, 
\beq
IV
&=&
\left(
\frac 
{L^{\beta}}
{
L^{
\rho_1  \rho_2/(1 + \epsilon)
}
}
\right)^2
=
L^{
2 \beta - 2\rho_1  \rho_2/(1 + \epsilon)
}
\eeq
We have
$IV = \mathcal{ O}( L^{\beta -1 - \delta})$
for some 
$\delta > 0$ 
if and only if 
$2\rho_1  \rho_2 - 1  > 0$.
Under this condition, 
by taking
\beq
\delta
&:=&
\min
\left\{
- \beta - 1 + 2  \rho_1  \rho_2, 
\,
\beta - \beta'
\right\}, 
\eeq
we have the statement part (1) of Lemma \ref{lemma:B}.\\
(2)
In part (2) of Lemma \ref{lemma:A}, 
we replace 
$I_{\Lambda}$
by
\beq
\widetilde{J}_m^{\pm}
&:=&
N^{-1}
\left(
N(E_0) + \frac {I_{m}^{\pm}}{|\Lambda_L|}
\right)
\\
X_{\Lambda_{\ell}, I_m^{\pm}} &:=&  {\bf{1}}_{\left\{\mbox{ 
$H_{\Lambda_{\ell}}$ has an eigenvalue in 
$\widetilde{J}_m^{\pm}$  centered in 
$\Lambda_{ \ell - \ell' }$} \right\}} ,
\eeq
which yields
\beq
\Bigl|
\Pp
\left(X_{\Lambda_{\ell}, I_m^{\pm}} = 1 \right) 
- 
N(\widetilde{J}_m^{\pm}) 
| \Lambda_{\ell}|  
\Bigr|
\le
C
N(\widetilde{J}_m^{\pm}) | \Lambda_{\ell} |
\cdot
\frac {\ell'}{\ell}
+
C
\ell e^{- (\ell')^{\xi}}
+
\mathcal{ O}
\left(
\left(
\frac 
{L^{\beta}}
{
L^{
\rho_1  \rho_2/(1 + \epsilon)
}
}
\right)^2
\right).
\eeq
Here we use 
\beq
N(\widetilde{J}_m^{\pm}) | \Lambda_{\ell} |
&=&
\frac {
|I_m^{\pm}|
}
{
| \Lambda_L |
}
| \Lambda_{\ell} |
=
|I_m^{\pm}|
\widetilde{N}^{-1}
\left(
1 + \mathcal{ O}(L^{- (\beta -\beta')})
\right).
\eeq
Thus
\beq
&&
\Bigl|
%p_m^{\pm}
\Pp
\left(X_{\Lambda_{\ell}, I_m^{\pm}} = 1 \right)
- 
|I_m^{\pm}|
\widetilde{N}^{-1}
\Bigr|
\\
&\le&
|I_m^{\pm}|
\widetilde{N}^{-1}
\mathcal{ O}(L^{- (\beta -\beta')})
+
C
N(\widetilde{J}_m^{\pm}) | \Lambda_{\ell} |
\cdot
\frac {\ell'}{\ell}
+
C
\ell e^{- (\ell')^{\xi}}
+
\mathcal{ O}
\left(
\left(
\frac 
{L^{\beta}}
{
L^{
\rho_1  \rho_2/(1 + \epsilon)
}
}
\right)^2
\right)
\\
&=:&
I + II + III + IV.
\eeq
As 
in the proof of part (1) of Lemma \ref{lemma:B}, we can show 
\beq
RHS
&=&
\widetilde{N}^{-1}
\mathcal{ O}
\left(
L^{-\delta}
\right).
\eeq
with the same constant 
$\delta > 0$
as in part (1). 
\end{proof}
%
%%%%%%%%%%%%%%%%%%%%%%%%%%%%%%%%%%
We 
are now in a position to prove Theorem \ref{thm:Poisson} based on the argument used to prove Poisson limit theorem.\\
\noindent
{\it Proof of Theorem \ref{thm:Poisson}}\\
Take 
$K > 0$
s.t.
$\bigcup_{m=1}^p I_{m} \subset [-K, K]$.
We apply
Theorem \ref{thm:decomposition}
to the following interval:
\beq
I_{\Lambda_L}
:=
N^{-1}
\left(
N(E_0) + \frac {[-K, K]}{L}
\right). 
\eeq
Let 
$\mathcal{ Z}_{\Lambda_L} \subset \Omega$
be the set of configurations in Theorem \ref{thm:decomposition}.
By the localization estimate, we have
\beq
&&
\bigcap_{m=1}^p
\left\{
\omega 
\, \middle| \,
\# \left\{
j 
\, \middle| \,
X_{\Lambda_{\ell}(\gamma_j), I_{m}^-} = 1
\right\}
=
k_{m}
\right\}
\cap 
\mathcal{ Z}_{\Lambda_L}
\\
&&
\quad
\subset
%\Omega
%\left(
%I_1, k_1 ; I_2, k_2 ; \cdots ; I_p, k_p
%\right)
\bigcap_{m=1}^p
\left\{
\omega 
\, \middle| \,
\xi_L (I_{m}) = k_{m}
\right\}
\cap
\mathcal{Z}_{\Lambda_L}
\\
&&
\qquad
\subset
\bigcap_{m=1}^p
\left\{
\omega 
\, \middle| \,
\# \left\{
j 
\, \middle| \,
X_{\Lambda_{\ell}(\gamma_j), I_{m}^+} = 1
\right\}
=
k_{m}
\right\}
\cap 
\mathcal{ Z}_{\Lambda_L}.
\eeq
Let us prove the upper bound. 
For 
$\omega \in 
\bigcap_{m=1}^p
\left\{
\omega 
\, \middle| \,
\# \left\{
j 
\, \middle| \,
X_{\Lambda_{\ell}(\gamma_j), I_{m}^+} = 1
\right\}
=
k_{m}
\right\}
\cap 
\mathcal{ Z}_{\Lambda_L}$, 
let 
\beq
K_m
:=
\left\{
j \in \{1, \cdots, \widetilde{N} \}
\,\middle| \,
X_{\Lambda_{\ell}(\gamma_j), I_{m}^+} = 1
\right\}, 
\quad
m = 1, 2, \cdots, p.
\eeq
be the set of indices $j$
such that 
$H_{\Lambda_{\ell}(j)}$
has an eigenvalue in 
$I_m^+$
centered in 
$\Lambda_{\ell - \ell'}(\gamma_j)$.
Then 
$\# K_m = k_m$,
and since
$\omega \in \mathcal{ Z}_{\Lambda_L}$, 
$K_m \cap K_{m'} = \emptyset$
for
$m \ne m'$.

Then, we have 
\begin{eqnarray}
&&
\Pp
\left(
%\Omega
%\left(
%I_1, k_1 ; I_2, k_2 ; \cdots ; I_p, k_p
%\right)
\bigcap_{m=1}^p
\left\{
\omega 
\, \middle| \,
\xi_L (I_{m}) = k_{m}
\right\}
\cap
\mathcal{ Z}_{\Lambda_L}
\right)
\nonumber
\\
& \le &
\Pp
\left(
\bigcap_{m=1}^p
\left\{
\omega 
\, \middle| \,
\# \left\{
j 
\, \middle| \,
X_{\Lambda_{\ell}(\gamma_j), I_{m}^+} = 1
\right\}
=
k_{m}
\right\}
\cap 
\mathcal{ Z}_{\Lambda_L}
\right)
\nonumber
\\
& = &
\sum_{
\substack{
K_m \subset \{1, \cdots, \widetilde{N} \} \\
\# K_m =k_m, \; 1 \le m \le p \\
K_m \cap K_{m'} = \emptyset, \; m \ne m'
}
}
\Pp
\left(
\bigcap_{m=1}^p 
\left\{
\omega
\, \middle| \,
\begin{array}{c}
\forall j \in K_m, \; X_{\Lambda_{\ell}(\gamma_j), I_m^+}=1
\\
\forall j \notin K_m, \; X_{\Lambda_{\ell}(\gamma_j), I_m^+}=0
\end{array}
\right\}
\cap
\mathcal{ Z}_{\Lambda_L}
\right)
\nonumber
\\
& = &
\sum_{
\substack{
K_m \subset \{1, \cdots, \widetilde{N} \} \\
\# K_m =k_m, \; 1 \le m \le p \\
K_m \cap K_{m'} = \emptyset, \; m \ne m'
}
}
\Pp
\left(
\bigcap_{m=1}^p 
\left\{
\omega
\, \middle| \,
\begin{array}{c}
\forall j \in K_m, \; X_{\Lambda_{\ell}(\gamma_j), I_m^+}=1
\\
\forall j \notin K_m, \; X_{\Lambda_{\ell}(\gamma_j), I_m^+}=0
\end{array}
\right\}
\right)
+
o(1). 
%\quad\cdots (1)
\label{firstequation}
\end{eqnarray}
In the last inequality, 
we used the estimate
$\Pp(\mathcal{ Z}_{\Lambda_L}) = 1 + o(1)$.
Since 
$\{ K_m \}_m$
are disjoint, 
$\{ X_{\Lambda_{\ell}(\gamma_j), I_m^+ } 
\}_{j \in K_m, \, 1 \le m \le p}$
are independent.
Thus 
\begin{eqnarray}
&&
\Pp
\left(
\bigcap_{m=1}^p 
\left\{
\omega
\, \middle| \,
\begin{array}{c}
\forall j \in K_m, \; X_{\Lambda_{\ell}(\gamma_j), I_m^+}=1
\\
\forall j \notin K_m, \; X_{\Lambda_{\ell}(\gamma_j), I_m^+}=0
\end{array}
\right\}
\right)
\nonumber
\\
& = &
\prod_{m=1}^p
\prod_{j \in K_m}
\Pp
\left(
X_{
\Lambda_{\ell}(\gamma_j), I_m^+
}
=1
\right)
\prod_{j \notin \bigcup_{m=1}^p K_m}
\Pp
\left(
\sum_{m=1}^p
X_{
\Lambda_{\ell}(\gamma_j), I_m^+
} = 0
\right)
\nonumber
\\
& = &
\prod_{m=1}^p
\Pp
\left(
X_{
\Lambda_{\ell}, I_m^+
}
=1
\right)^{k_m}
\cdot
\Pp
\left(
\sum_{m=1}^p
X_{
\Lambda_{\ell}, I_m^+
} = 0
\right)^{ \widetilde{N} - k }
%\quad\cdots (2)
\label{secondequation}
\end{eqnarray}
where
$
k
:=
\sum_{m=1}^p k_m.
$
By Lemma \ref{lemma:B}, 
we have 
\begin{eqnarray}
&&
\prod_{m=1}^p
\Pp
\left(
X_{
\Lambda_{\ell}, I_m^+
}
=1
\right)^{k_m}
\Pp
\left(
\sum_{m=1}^p
X_{
\Lambda_{\ell}, I_m^+
} = 0
\right)^{ \widetilde{N} - k }
\nonumber
\\
& = &
\prod_{m=1}^p
\left(
\frac { |I_m^+ | }{ \widetilde{N} }
\right)^{k_m}
e^{
- \sum_{m=1}^p 
\frac {|I_m^+|}{ \widetilde{N} } 
(\widetilde{N} - k )
}
(1 + o(1))
\nonumber
\\
&=&
\prod_{m=1}^p
\left(
| I_m^+ |^{k_m}
e^{ - |I_m^+|}
\right)
\frac {1}{
\widetilde{N}^{k}
}
(1 + o(1)).
%\quad\cdots (3)
\label{thirdequation}
\end{eqnarray}
On the other hand, note that 
\begin{equation}
\sum_{
\substack{
K_m \subset \{1, \cdots, \widetilde{N} \} \\
\# K_m =k_m, \; 1 \le m \le p \\
K_m \cap K_{m'} = \emptyset, \; m \ne m'
}
}
=
\frac {
\widetilde{N}!
}
{
k_1! \cdots k_p! (\widetilde{N} - k)!
}.
%\quad\cdots (4)
\label{fourthequation}
\end{equation}
By 
(\ref{firstequation}), 
(\ref{secondequation}), 
(\ref{thirdequation}), 
and
(\ref{fourthequation}), 
%(1) $\sim$ (4), 
we have 
\beq
&&
\Pp
\left(
\bigcap_{m=1}^p
\left\{
\omega 
\, \middle| \,
\xi_L (I_{m}) = k_{m}
\right\}
\cap
\mathcal{ Z}_{\Lambda_L}
\right)
\\
& \le &
\frac {
\widetilde{N}!
}
{
k_1! \cdots k_p! (\widetilde{N} - k)!
}
\cdot
\prod_{m=1}^p
\left(
| I_m^+ |^{k_m} 
e^{ - |I_m^+|}
\right)
\frac {1}{
\widetilde{N}^{k}
}
(1 + o(1))
\\
&=&
\prod_{m=1}^p
\frac { |I_m^+ |^{k_m} }{ k_m ! }
e^{ - |I_m^+|}
(1+o(1))
\\
&=&
\prod_{m=1}^p
\frac { |I_m |^{k_m} }{ k_m ! }
e^{ - |I_m|}
(1+o(1)).
\eeq
which completes the proof of upper bound.
Lower bound 
follows similarly. 
\QED\\

\begin{appendices}

\section{Properties of free and modified Prufer phases from \cite{jss}}\label{app:phases1}
\setcounter{equation}{0}

Following \cite{jss}, we define the \emph{free Prufer variables} 
$(R^{0}_n(E), \theta^{0}_n(E))$ by
\bel{eq:prufer_defn_free1}
R^{0}_n(E) e_{\theta^{0}_n(E) }:=
 \left( \begin{array}{c}
t(n) u(n) \\ u(n-1)
\end{array}
\right), 
\ee
%
%\quad
%e_{\theta}
%:= \left(  \begin{array}{c}
%\cos \theta \\ \sin \theta
%\end{array}
%\right)  ,
%\ee 
\bel{eq:phase_range1}
%&&
- \frac {\pi}{2}
<
%\left|
\theta_{n+1}^0 (E) - \theta_n^0 (E)
%\right|
<
\frac {3 \pi}{2}  .
\ee
and the \emph{modified Prufer variables} $(R_n (E), \theta_n (E))$ by 
\bel{eq:prufer_defn2}
R_n (E)
e_{\theta_n (E)}
=
M \left(
\begin{array}{c}
t(n) u(n) \\ u(n-1)
\end{array}
\right), 
\quad
e_{\theta}
:=
\left(
\begin{array}{c}
\cos \theta \\ \sin \theta
\end{array}
\right)  ,
\ee 
and where $M$ is the invertible matrix mapping the polymer transfer matrices $T_{\pm}(E_c)$ to rotations as in \eqref{eq:diagonalize1}: 
\bel{eq:diagonal1}
M T_{\pm}^{E_c} M^{-1}
= R(\eta_{\pm}), 
\quad
R(\eta)
:=
\left(
\begin{array}{cc}
\cos \eta & - \sin \eta \\
\sin \eta & \cos \eta
\end{array}
\right)  .
%\quad\cdots (*)
\ee

We recall from \eqref{eq:m_fncs1} that $M$ has the property that 
\bel{eq:m_defn2}
M e_\theta = r(\theta) e_{m(\theta)},
\ee
with $r(\theta) > 0$, and the smooth function $m$ satisfies $m(\theta + \pi ) = m(\theta) + \pi$.
It follows from these definitions that 
free Prufer and modified Prufer phases are related by 
\bel{eq:free_mod1}
\theta_n (E)  =  m \left(  \theta^{0}_n (E)  \right).
\ee

%%%%%%%%%%%%%%%%%%%%%%%%%%%%%%%%%%%
%
%\subsection{Properties of modified Prufer phase}
%%
%Let us recall definitions.
%%
%$M$
%is the real 
%$2 \times 2$
%matrix which diagonalizes 
%$T_{\pm}(E_c)$
%simultaneously : 
%%
%\beq
%M T_{\pm}^{E_c} M^{-1}
%=
%R(\eta_{\pm}), 
%\quad
%%
%R(\eta)
%:=
%\left(
%\begin{array}{cc}
%\cos \eta & - \sin \eta \\
%%
%\sin \eta & \cos \eta
%\end{array}
%\right)
%\quad\cdots (*)
%\eeq
%%
%The 
%free Prufer variable 
%$(R^{0, E}(n), \theta^{0, E}(n))$
%is defined  by
%%
%\beq
%R_{0,E} (n)
%e_{\theta^{0, E}(n)}
%:=
%%
%\left(
%\begin{array}{c}
%t(n) u(n) \\ u(n-1)
%\end{array}
%\right), 
%%
%\quad
%e_{\theta}
%:=
%%
%\left(
%\begin{array}{c}
%\cos \theta \\ \sin \theta
%\end{array}
%\right)
%%
%\eeq
%%
%For any fixed 
%$n$, 
%$E \mapsto \theta^{0, E}(n)$
%is increasing(Lemma 2 in JSS).
%
%
%The modified 
%%
%Prufer variable
%$(R_n (E), \theta_n (E))$
%is defined by 
%%
%\beq
%&&
%R_n (E)
%e_{\theta_n (E)}
%:=
%M
%%
%\left(
%\begin{array}{c}
%t(n) u(n) \\ u(n-1)
%\end{array}
%\right), 
%%
%\quad
%e_{\theta}
%:=
%%
%\left(
%\begin{array}{c}
%\cos \theta \\ \sin \theta
%\end{array}
%\right)
%%
%\\
%%
%&&
%- \frac {\pi}{2}
%<
%\left|
%\theta_{n+1} (E) - \theta_n (E)
%\right|
%<
%\frac {3 \pi}{2}
%\eeq
%%
%Here 
%we note that there exists 
%$C^{\infty}$-function
%$m : {\bf R} \to {\bf R}$
%such that
%$m (\theta + \pi) = m(\theta) + \pi$
%and
%%
%\beq
%M e_{\theta} 
%&=&
%r(\theta) e_{ m(\theta) }
%\eeq
%%
%for some 
%$r (\theta) > 0$.
%\\
%%
%%%%%%%%%%%%%%%%%%%%%%%%%%%%%%%%%%%%%%%%

\section{Proof of the increasing property of $x \mapsto \Psi_L(x)$}\label{app:incr_psi1}
\setcounter{equation}{0}

We recall the function $\Psi_L(x)$ defined by 
\bel{eq:psi1}
\Psi_L (x)  :=  \frac {1}{\pi}
\left\{ \theta_L  \left(  E_c+\frac {x}{ n(E_c) L}  \right)  - \theta_L(E_c)  \right\}, 
\quad
x \in {\R}.
\ee
We prove that $x \in \R \to \Psi_L(x)$ is an increasing function. 

\medskip
%
%%%%%
\begin{lemma}\label{lemma:det_m1}
The function $m$ defined in \eqref{eq:m_defn2}
is strictly increasing if and only if 
$\det M > 0$.
\end{lemma}
%%%%%
%
\begin{proof}
Since $M$ is a $2 \times 2$ invertible matrix, we write 
\beq
M 
&=&
\left(
\begin{array}{cc}
a & b \\ c & d 
\end{array}
\right), 
\quad
a, b, c, d \in {\R}, 
\quad
\det M = a d - b c \neq 0.
\eeq
From the definition of $m$ in \eqref{eq:m_fncs1}, we have 
\beq
M e_{\theta}
&=&
\left(
\begin{array}{c}
a \cos \theta + b \sin \theta \\
c \cos \theta + d \sin \theta
\end{array}
\right) =
r(\theta) e_{m(\theta)}, 
\quad
e_{\theta}
:=
\left(
\begin{array}{c}
\cos \theta \\ \sin \theta
\end{array}
\right) .
\eeq
It follows that 
%
%\leadsto
%\qquad
\bel{eq:phase11}
m (\theta)
  :=
\arg M e_{\theta}
=
\arctan
g(\theta), 
\quad
g(\theta)
:=
\frac {
c \cos \theta + d \sin \theta
}
{
a \cos \theta + b \sin \theta 
} .
\ee
Consequently, we have,  
\beq
m'(\theta)
&=&
\frac {g'(\theta)}{1 + g (\theta)^2 }, 
\quad
g'(\theta)
=
\frac {
ad - bc
}
{
(a \cos \theta + b \sin \theta)^2
}
=
\frac {
\det M
}
{
(a \cos \theta + b \sin \theta)^2
}
.
\eeq
\end{proof}
%
%
%%%%%
\begin{lemma}
Without loss of generality, 
we may suppose that
$\det M >0$. 
\end{lemma}
%%%%%
%
\begin{proof}
We set
\beq
J := 
\left(
\begin{array}{cc}
0 & 1 \\ 
1 & 0 
\end{array}
\right), 
\quad
R (\theta)
:=
\left(
\begin{array}{cc}
\cos \theta & - \sin \theta \\
\sin \theta & \cos \theta
\end{array}
\right), 
\quad
\theta \in {\R}.
\eeq
Then, noting that $J^{-1} = J$, we have 
\beq
J MJ = J
\left(
\begin{array}{cc}
a & b \\
c & d
\end{array}
\right)
J =
\left(
\begin{array}{cc}
d & c \\
b & a
\end{array}
\right)  . 
%
%\quad
%J^{-1} = J.
\eeq
%
%We write 
%$T$
%instead of 
%$T_{\pm}(E_c)$
%for simplicity.
%%
%We multiply
%$J$
%to 
Since $M T_{\pm}(E_c)  M^{-1} = R (\eta)$, applying $J$ 
to both sides of this equation, we obtain  
\bea\label{eq:det_m2}
(JM) T_{\pm}(E_c)  (JM)^{-1}        
   =  J R(\eta) J  = R(- \eta). 
  \eea
%\\
%
%\leadsto
%\qquad
%(JM) T (JM)^{-1}
%&=&
%　
%\left(
%\begin{array}{cc}
%\cos \eta & \sin \eta \\
%- \sin \eta & \cos \eta
%\end{array}
%\right)
%
%=
%R ( - \eta ).
%\quad\cdots (**)
%\eeq
%
Thus, if 
$\det M < 0$, 
then 
$\det (JM) = \det J \cdot \det M > 0$, 
and by \eqref{eq:det_m2}, 
%$(**)$
%implies that 
we may replace 
$M$
and
$\eta_{\pm}$
by 
$JM$
and
$- \eta_{\pm}$,
respectively. 
\end{proof}
%
%
%%%%%%%%%%%%%%%%%%%%%%%%%%%%%%%%%%%%%%%%

\begin{proposition}\label{prop:psi_incr1}%{\bf Proposition A.3}
The map $x \in \R \mapsto \Psi_L (x)$
is increasing.
\end{proposition}
%%%%%
%
\begin{proof}
we recall from \eqref{eq:free_mod1}
that free Prufer and modified Prufer variables are related by : 
\beq
\theta_n (E)
&=&
m 
\left(
\theta^{0}_n (E)
\right).
\eeq
By \cite[Lemma 2]{jss}, 
$E \mapsto \theta^{0}_n(E)$
is increasing for any fixed 
$n$.
By lemmas A.1, A.2, 
$\theta \mapsto m(\theta)$
is increasing.
Therefore
$E \mapsto \theta_n (E)$
is increasing for any fixed 
$n$.
The increasing property of $\Psi_L(x)$ then follows from the definition of 
$\Psi_L (x)$ : 
\[
\Psi_L (x)
:=
\frac {1}{\pi}
\left\{
\theta_L
\left(
E_c+\frac {x}{ n(E_c) L}
\right)
-
\theta_L(E_c)
\right\}, 
\quad
x \in {\R},
\]
and these properties. 
\end{proof}
%

%%%%%%%%%%%%%%%%%%%%%%%%%%%%%%%%%%%%%%%%%%%%%%%%%%%%%%%%%%%%%%%%%%%%%%%%%%%%%%%

%\section{Phase shifts and action multipliers from \cite{jss}}\label{app:shifts1}
%\setcounter{equation}{0}

%In this appendix, we summarize the results of \cite{jss} needed in section \ref{subsec:verify_clock1}. The authors in \cite{jss} study the transfer matrices at energies $E_c + \epsilon$, near a critical energy $E_c$. Recalling that the matrix $M \in SL(2 , \R)$ maps $T_\pm^{E_c}$ to rotations,  we define the related transmission $a_\pm^\epsilon$ and reflection $b_\pm^\epsilon$ coefficients by
%\bel{eq:trans_refl1}
%M T_\pm^{E_c+\epsilon} M^{-1} \nu = a_\pm^\epsilon \nu + b_\pm^\epsilon \overline{\nu} ,
%\ee
%where 
%$$
%v := \frac{1}{\sqrt{2}} \left( \begin{array}{c}
%                           1 \\  
%                           -i    
%                            \end{array}  \right).
%$$ 
%We have that $| a_\pm^\epsilon |^2 - | b_\pm^\epsilon|^2 = 1.$     %We note \cite[(41)]{jss}:
%Two other sets of coefficients occur in the theory developed in \cite{jss}:
%\bel{eq:coef2}
%c_\pm := ( \partial_\epsilon b_\pm^\epsilon) |_{\epsilon = 0} e^{i \eta_\pm} ; 
%~~~~~ d_\pm := ( \partial_\epsilon \eta_\pm^\epsilon ) |_{\epsilon = 0} .
%\ee

%%%%%%%%%%%%%%%%%%%%%%%%%%%%%%%%%%%%%%%%%

\section{Bootstrap MSA for the AB RPM}\label{app:bootstrap1}

The paper of von Dreifus and Klein \cite{vDK1} presents a proof of exponential localization for the AB model on $\ell^2(\Z)$ using the positivity of the Lyapunov exponent for the initial length scale estimate, and the Wegner estimate of \cite{CKM}. 
The technical conclusion of their analysis is
\bel{eq:vdk1}
\Pp \{ \forall E \in \Sigma, {\rm either} ~\Lambda_{L_k}(x) ~{\rm or} ~ \Lambda_{L_k}(y) 
 ~{\rm is} ~ (m,E)-{\rm regular} \} \geq 1 - \frac{1}{L^p} .
\ee
The bootstrap MSA improves this probability bound to $\geq 1 - e^{L_k^\xi}$ so that the probability of bad events are subexponentially small.

There is an aspect of the general RPM that differs from the AB model with respect to MSA. When the lengths of the two polymers are the same, $L_- = L_+$, we can choose intervals of length $L$ that are integer multiples of  $L_+ = L_-$, and the MSA proceeds in the usual manner as for the AB model. However, for the general RPM for which $L_- \neq L_+$, the configurations of the two different polymers in an interval of fixed length $L$ vary. 
That is, if one fixes an interval $\Lambda_L$, the number of random polymers that intersect this interval varies with the configuration $\omega$ and, in general, will contain only parts of a polymer possibly at one or both ends. However, the bootstrap MSA allows for some flexibility in choosing independent intervals. Condition 3 in \cite{gk_boot1}, \textbf{IAD}, independence at a distance, requires only that there exist $\delta > 0$, so that $\delta$-nonoverlapping intervals are independent. If we choose $\delta = \max \{ L_-, L_+ \}$, then the local \Schr operators associated with $\Lambda_L(x)$ and $\Lambda_L(y)$ are independent provided $|x-y| > \delta$. 

The other basic assumptions of bootstrap MSA are 
\begin{enumerate}
\item  \textbf{SLI}. The Simon-Lieb inequality is used to iterate resolvents over lengths scales. 
\item \textbf{EDI}. The exponential decay inequality provides an upper bound on the norm of a generalized eigenvector in terms of a localized resolvent. 
\item \textbf{IAD}. The independence at a distance was discussed above.
\item \textbf{NE}. The number of eigenvalues is an upper bound on the expected number of eigenvalues of $H_L$ in an interval in terms of $L$.
\item \textbf{W}. The Wegner estimate is discussed in detail in section \ref{subsec:w_m1}.
\end{enumerate}

%%%%%%%%%%%%%%%%%%%%%%%%%%%%%%%%%%%

For AB model or the RPM, let $I \subset \Sigma \backslash \mathcal{C}$
be an interval contained in localization regime minus the critical set $\mathcal{C}$, so that the Lyapunov exponent $L(E)$ is strictly positive for $E \in I$.
Then, the conclusion of bootstrap MSA on the interval $I$ is : 
%\omega$. 
\bel{eq:event_msa1}
{\Pp}\left\{ \forall E \in I, \mbox{ either 
$\Lambda_L(x)$ or $\Lambda_L(y)$
is $(m, E)$-regular  }
\right\}   
\geq
1 - e^{- L_k^{\zeta}}.
\ee
To state the theorems below,  we set
\bel{eq:event_msa_defn1}
R(m, L, x, y)
:=
\left\{
\forall E \in I, 
\mbox{ either 
$\Lambda_L(x)$
or
$\Lambda_L(y)$
is
$(m, E)$-regular  }
\right\}
\ee

\medskip

We recall
the notations and statements in Germinet-Klein Bootstrap paper \cite{gk_boot1} for $\ell^2 (\Z)$. 
\begin{defn}\label{defn:msa1}
\beq
\Lambda_L(x)
&:=&
\left\{
y \in {\Z}^d
\, \middle| \,
|x-y| \le L/2
\right\}, 
\quad
|x| := \max_{1 \le j \le d} |x_i|, 
\quad
x = (x_1, \cdots, x_d) \in {\Z}^d
\\
R_{x, L}(z)
&:=&
(H_{\Lambda_L(x)} - z)^{-1}
\\
\chi_{x, \ell}
&:=&
\chi_{ \Lambda_{\ell} (x) }, 
\quad
\chi_x
:=
\chi_{x, 1}
\\
\| A \|_{x, L}  &:=&   \mbox{operator norm for an operator
$A$
on 
$\ell^2(\Lambda_L(x))$ }
\eeq
\end{defn}

\medskip

\begin{defn}\label{defn:suitable1} 
Let 
$E \in {\R}$, 
$\theta > 0$, 
$x \in {\Z}^d$, 
$L \in 6 {\N}$.
A box 
$\Lambda_L (x)$
is 
{\bf  $(\theta, E)$-suitable}
\beq
\;\stackrel{def}{\Longleftrightarrow}\;
\| 
\Gamma_{x, L} R_{x,L}(E) \chi_{x, L/3}
\|_{x, L}
\le
L^{- \theta}
\eeq
\end{defn}

\medskip

\begin{defn}\label{defn:regular1}
Let 
$E \in {\R}$, 
$m > 0$, 
$x \in {\Z}^d$, 
$L \in 6 {\N}$.
A box 
$\Lambda_L (x)$
is 
{\bf $(m, E)$-regular}
\beq
\;\stackrel{def}{\Longleftrightarrow}\;
\| 
\Gamma_{x, L} R_{x,L}(E) \chi_{x, L/3}
\|_{x, L}
\le
e^{- m \cdot  \frac L2}
\eeq
\end{defn}

\medskip

Germinet and Klein derived derived four versions of MSA : Theorems 
5.1, 5.2, 5.6, and 5.7, appearing as Theorem \ref{thm:5.1}-\ref{thm:5.7} 
below.
By using for theorems successively, 
the starting hypothesis of Theorem 5.1 
leads us to the conclusions of Theorem 5.7.

\medskip

\begin{theorem}\cite[Theorem 5.1]{gk_boot1}\label{thm:5.1}
Let 
$E_0 \in I_0$, 
$\theta > bd$,
$0 < p < \theta - bd$.
Suppose that 
\beq
{\Pp}  \left\{ \Lambda_{L_0}(0)
\mbox{ is  $(\theta, E_0)$-suitable }
\right\} 
\ge
1 - (3Y-4)^{-2d}, 
\quad
L_0 \gg 1.
\eeq
Then setting
$L_{k+1} := Y L_k$, 
$k = 0, 1, \cdots$, 
we have 
\beq
{\Pp} \left\{  \Lambda_{L_k}(0)  \mbox{ is 
$(\theta, E_0)$-suitable }
\right\}  
\ge
1 - L_k^{-p}, 
\quad
\mbox{ for }
k \gg 1.
\eeq
\end{theorem}

\medskip

\begin{theorem}\cite[Theorem 5.2]{gk_boot1}\label{thm:5.2}
Let 
$E_0 \in I_0$, 
$\theta > bd$,
$0 < p' < p < \theta - bd$, 
and
$1 < \alpha < 1 + \dfrac {p'}{p'+2d}$.
Suppose that 
\beq
{\Pp}  \left\{  \Lambda_{L_0}(0)
\mbox{ is  $(\frac {2}{L_0}\cdot \theta \log L_0, E_0)$-regular }
\right\}
\ge
1 - L_0^{-p}, 
\quad
%L_0 \gg 1.
\eeq
Then 
we can find an interval
$I(\delta_1)
:=
[E_0 - \delta_1, E_0 + \delta_1] \cap I_0$
and setting 
$m_0 := \frac {2}{L_0} \theta \log L_0$, 
$L_{k+1} := [ L_k^{\alpha} ]$, 
$k = 0, 1, \cdots$, 
we have
\beq
{\Pp}  \left\{
\Lambda_{L_k}(0)
\mbox{ is 
$(\frac {m_0}{2}, E)$-regular }
\right\}   
\ge
1 - L_k^{-p'}, 
\quad
\mbox{ for }
\forall 
E \in I(\delta_1)
%L_0 \gg 1.
\eeq
and moreover
\beq
{\Pp}   \left\{   
R
\left(
\frac {m_0}{2}, L_k, I(\delta_1), x, y
\right)
\right\}  
\ge
1 - L_k^{-2p'}, 
\quad
|x - y| > L_k, 
\quad
k = 0, 1, \cdots
\eeq
\end{theorem}

\medskip

\begin{defn}\label{defn:subexp1}
For any $\zeta \in (0,1)$, 
$E \in {\R}$, 
$x \in {\Z}^d$, 
$L \in 6 {\N}$, the interval $\Lambda_L(x)$ is 
{\bf  $(\zeta, E)$-subexponentially suitable}  if
\beq
\;\stackrel{def}{\Longleftrightarrow}\;
\| 
\Gamma_{x, L} R_{x,L}(E) \chi_{x, L/3}
\|_{x, L}
\le
e^{- L^{\zeta}}
\eeq
\end{defn}

\medskip

\begin{theorem}\cite[Theorem 5.6]{gk_boot1}\label{thm:5.6}
$\zeta \in (0,1)$, 
$E \in {\R}$, 
$x \in {\Z}^d$, 
$L \in 6 {\N}$.
$\Lambda_L(x)$
is 
{\bf$(\zeta, E)$-subexponentially suitable}
\beq
\;\stackrel{def}{\Longleftrightarrow}\;
\| 
\Gamma_{x, L} R_{x,L}(E) \chi_{x, L/3}
\|_{x, L}
\le
e^{- L^{\zeta}}.
\eeq
\end{theorem}

\medskip

\begin{theorem} \cite[Theorem 5.7]{gk_boot1}\label{thm:5.7}
Let 
$E_0 \in I_0$, 
$\zeta_0 \in (0,1)$, 
$Y \gg 1$.
For 
$0 < \forall \zeta_1 < \zeta_0$, 
suppose that
\beq
{\Pp}   \left\{ \Lambda_{L_0}(0)
\mbox{ is
$(\zeta_0, E_0)$-subexponentially suitable }
\right\}  
>
1 - (3Y-4)^{-2d}
\eeq
Then
setting 
$L_{k+1} := Y L_k$, 
we have
\beq
{\Pp}   \left\{ \Lambda_{L_k}(0)
\mbox{ is
$(\zeta_0, E_0)$-subexponentially suitable }
\right\}  
>
1 - e^{- L_k^{\zeta_1}}, 
\quad
\mbox{ for }
k \gg 1. 
\eeq
\end{theorem}

\medskip

For the AB or random dimer model, 
the statements in Theorems 5.2, 5.6, and 5.7 
are valid with slight modification in the proof. 
Moreover, 
the assumption in Theorem 5.7 has been given for AB model by the paper by von Dreifus and Klein \cite{vDK1}.
Therefore, 
by following Theorem 5.7, 
we have the conclusion in Theorem 5.7 for AB model.
We remark that the exponent $\xi$ appearing in the subexponential decay bounds and the eigenvector bounds can be set equal to one for models amenable to the fractional moment method \cite{am1}. However, this method has not been extended to models with singular probability measures such as the AB model.

%%%%%%%%%%%%%%%%%%%%%%%%%%%%%%%%%%%%%%%%%%%%%%%%%%%%%%%%%%%%%%%%%%%%%%%%%%%%%%%%%%

\section{Bourgain's Lemma}\label{app:bourgain1}
\setcounter{equation}{0}

As a preliminary to proving Theorem \ref{thm:minami1}, 
we give a proof of a result on the the probability of two eigenvalues of $H_\ell$ being in an interval $I \subset \R$. 
This is 
originally given in \cite[Lemma 3]{bourgain1} but modified to fit to our model.

\begin{lemma}\label{lemma:bourgain1}
Let 
$I := (E_0 - \delta, E_0 + \delta)$, 
$\ell > 0$, 
$0 < \beta \le 1$, 
and define length scales as follows: 
$\delta := c_2 e^{- \frac{\ell^{\beta}}{c_1}}$ 
and parameters 
$0 < t_0 < \frac{1}{2}$ 
and 
$0 < \epsilon < 1$. 
Let us 
consider the following event 
$\Omega_\delta$ :
\noindent
There exist two orthonormal vectors  $\xi, \xi' \in \R^\ell$ so that:
\begin{enumerate}
\item They are approximate eigenvectors for $H_\ell$ and $E_0$:
$$
\| (H_\ell - E_0) X \| < \delta, ~~~{\rm for} ~~X = \xi, \xi' $$ ;
\item The eigenvector mass is less distributed near the boundary : If $\xi = (\xi_j)$, then
$$
\max_{0 \leq j \leq t_0 \ell} ( | \xi_j | , | \xi_{l - j} | , | \xi'_j | , | \xi'_{\ell- j} | ) < \frac{1}{10\ell}. 
$$
\end{enumerate}
%\right\} .
%\eeq
We then have the bound for sufficiently large  
$\ell$ : 
\bel{eq:two_ev2}
\Pp \{ \Omega_\delta \} 
\leq 
e^{- \frac{2 \rho_1}{c_1 (1 + \epsilon)} \ell^{\beta}} 
+ 
e^{-q t_0 \ell^{\beta}} ,
\ee
where 
$\rho_1$
is the \Holder exponent for the IDS, and 
$q$ 
is defined in \eqref{eq:const_q1}.
\end{lemma}
%%%%%%%%%%%%%%%%%%%%%%%%%%%%%%%%%%%%%%%%
%\subsection{A version of Minami's estimate}
%%
%The following lemma 
%is Lemma 3 in 
%\cite{bourgain1}.
%\\
%%
%
%%
%%%%%%%%%%
%%\begin{itembox}[l]
%\begin{lemma} \cite[Lemma 3]{bourgain1} (A Minami's Estimate)
%%
%$I = (E_0 - \delta, E_0 + \delta)$, 
%%
%$\ell = c_1 \log L$, 
%%
%$\delta 
%:= 
%\dfrac {c_2}{L}
%=
%c_2 e^{ - \frac {\ell}{c_1}}$, 
%%
%$0 < t_0 < \frac 12$, 
%%
%$0 < \epsilon \ll 1$.
%%
%Then for 
%$\ell \gg 1$, 
%%
%\bea
%{\Pp}
%\Biggl(
%\exists \xi, \xi' \in {\bf R}^{\ell}, 
%\mbox{ s.t. }
%&(1)& \;
%\| \xi \| = \| \xi' \| = 1, 
%\xi \perp \xi', 
%\\
%%
%&(2)&\;
%\| (H_{\ell} - E_0) \xi \| < \delta, 
%\;
%\| (H_{\ell} - E_0) \xi' \| < \delta,
%\\
%%
%&(3)&\;
%\max_{
%0 \le j  \le t_0 \ell
%}
%\Bigl(
%\| \xi_j \|, 
%\| \xi_{\ell - j} \|, 
%%
%\| \xi'_j \|, 
%\| \xi'_{\ell - j} \| 
%\Bigr)
%< \frac {1}{\ell^{10}}
%\Biggr)
%%
%\le 
%%
%e^{
%-
%\frac {2 \rho}{c_1 (1 + \epsilon)} \ell
%}
%+
%e^{- p t_0 \ell}.
%\eea
%\end{lemma}
%%
%%\end{itembox}\\
%%%%
%

%
\begin{remark} %Remark : \\
\begin{enumerate}
\item We do not use localization estimate on the approximate eigenvectors $\xi, \xi'$, but conditions (1), (2) usually follow from it. 

\item This result was proved by Bourgain for the AB model, that is, rank one perturbations at each point of $\Z$ with independent, Bernoulli distributed coefficients. In particular, in the proof, it is important that the local \Schr operators associated with two disjoint subintervals of $[0, \ldots, \ell]$ of the from $[0, j_1]$ and $[j_1+1, \ell]$, are independent. For the AB model this is immediate. For the RPM, however, this is not necessarily true. We will show that we may choose the endpoints $j_1$ and $j_1 + 1$ to be polymer nodes, thereby insuring the independence of the restrictions of the \Schr operator to these two subintervals. 
\end{enumerate}

\end{remark}
%

%%
%\noindent
%{\bf Idea of Proof }\\
%%
%Let $\xi, \xi' \in \R^\ell$ be unit vectors satisfying conditions (1) and (2) of the lemma:
%$\xi \perp \xi'$ and
%$(H-E_0) \xi, (H-E_0) \xi' \sim \delta$.

\medskip

\begin{proof}
\noindent
1. Since $\|\xi\| = 1$, and condition (2) holds, for $t_0 \in (0, \frac{1}{2})$, there exists an index $v$, with  
%$\exists v$
%s.t.
%
\beq
%\exists v
%\;
%\mbox{ s.t. }
t_0 \ell < v < (1 - t_0) \ell
\quad
\mbox{and}
\quad
| \xi_v | \ge \dfrac {1}{\sqrt{\ell}}.
\eeq
Moreover, since $\xi \perp \xi'$,  we have
\bel{eq:v_lb1}
\| \xi_v \xi' - \xi'_v \xi \|  =
\left( |\xi_v|^2 + |\xi'_v|^2  \right)^{1/2}
\ge  |\xi_v|  \ge   \dfrac {1}{\sqrt{\ell}}  .
\ee
Hence,  we obtain : There exists an index $v_1$, with $ t_0 \ell < v_1 < (1 - t_0) \ell$, so that 
\bel{eq:33}   
%\exists v_1  \;
%\mbox{ s.t. }
%t_0 \ell < v_1 < (1 - t_0) \ell
%\quad
%\mbox{and}
%\quad
| \xi_v \xi'_{v_1} - \xi'_v \xi_{v_1} |
\ge
\frac {1}{\ell}.
%\quad\cdots (33)
\ee
In fact, 
if this is not so,  let $w := 
\xi_v \xi' - \xi'_v \xi$ . Then any component $w_k$ satisfies 
$| w_k | < 1 / \ell$, for all $k \in [0, \ell]$. This implies 
$$
\| w \| = \left( \sum_{k=0}^\ell | w_k|^2 \right)^\frac{1}{2} \leq \left( \ell \cdot \frac{1}{\ell^2} \right)^\frac{1}{2}  = \frac{1}{\ell^{\frac{1}{2}}} ,
$$
 which is a contradiction to \eqref{eq:33}. 
%
%$\leadsto$
%$| \rho | 
%<
%\sqrt{
%\frac {1}{\ell} \cdot \ell^2
%}
%=
%\frac {1}{\sqrt{\ell}}$
%%
%which is a contradiction to \eqref{eq:33}. 
%%
Without loss of generality,  we can assume  
$v < v_1$.
We define a phase $\theta_j$ by 
\beq
\xi_j + i \xi'_j
=:
(|\xi_j|^2 + |\xi'_j|^2 )^{1/2}
e^{i \theta_j}, 
\quad
1 \le j \le \ell.
\eeq
Then by \eqref{eq:33}, we have
%
%\beq
%| \sin ( \theta_v - \theta_{v_1}) | \ge \frac {1}{2\ell}.
%\eeq
%
%In fact,
%
\beal{eq:trig1}
\frac {1}{\ell}
& \le &
| \xi_v \xi'_{v_1} - \xi'_v \xi_{v_1} |  \nonumber
\\
&=&
\left(
|\xi_v|^2 + |\xi'_v|^2
\right)^{1/2}
\left(
|\xi_{v_1}|^2 + |\xi'_{v_1}|^2
\right)^{1/2}
\left|
\cos \theta_v \cdot \sin \theta_{v_1} 
-
\sin \theta_v \cdot \cos \theta_{v_1}
\right|   \nonumber
\\
&=&
\left(
|\xi_v|^2 + |\xi'_v|^2
\right)^{1/2}
\left(
|\xi_{v_1}|^2 + |\xi'_{v_1}|^2
\right)^{1/2}
| \sin (\theta_v - \theta_{v_1}) |
\\
& \le &
2
| \sin (\theta_v - \theta_{v_1}) |.  \nonumber
\eeal
In the last inequality, 
we used the fact that 
$\| \xi \| = \| \xi' \| = 1$
implies
$|\xi_v| \le 1$, 
$| \xi'_v | \le 1$
for all indices $v$.  This implies:
\bel{eq:angle_lb1}
| \sin ( \theta_v - \theta_{v_1}) | \ge \frac {1}{2\ell}.
\ee
\noindent
2. We begin the construction of approximate eigenvectors for the restriction of $H_\ell$ to subintervals of $[0, \ldots, \ell]$. The goal is to find two subintervals on which the restricted \Schr operators are independent. For this reason, we consider inner polymer nodes. The endpoints $0$ and $\ell$ do not have to be polymer nodes. 

\medskip

\noindent
{\bf Case 1 } Suppose there is a polymer node $j_1$ so that $v < j_1 < v_1$, and the phases satisfy
\beq
%v < \exists j_1 < v_1
%\;
%\mbox{ s.t. }
| \sin (\theta_v - \theta_{j_1}) |
>
\frac {1}{10 \ell}, 
\quad
| \sin (\theta_{v_1} - \theta_{j_1}) |
>
\frac {1}{10 \ell}. 
\eeq
We define another vector $\eta$, a linear combination of $\xi$ and $\xi'$, by
\beq
\eta
&:=&
\frac {
\xi_{j_1} \xi' - \xi'_{j_1} \xi
}
{
\|
\xi_{j_1} \xi' - \xi'_{j_1} \xi
\|
}
=
\frac {\xi_{j_1} \xi' - \xi'_{j_1} \xi }{( |\xi_{j_1}|^2 + | \xi'_{j_1} |^2 )^{1/2}} , 
\eeq
so that $\| \eta \| = 1$ and $\eta_{j_1} = 0$. This vector $\eta$ has entries satisfying 
%\beq
%\|  (H_{\ell} - E_0) \eta  \|   &\le&  \frac {|\xi_{j_1}|  \| (H_{\ell} - E_0) \xi' \|
%+
%|\xi'_{j_1}|  \| (H_{\ell} - E_0) \xi \|  }{  ( |\xi_{j_1}|^2 + | \xi'_{j_1} |^2 )^{1/2}}
%<
%2 \delta
%\\
%
\beal{eq:eta1}
| \eta_v |
&=&
\frac {
\left(
|\xi_v|^2 + |\xi'_v|^2
\right)^{1/2}
\left(
|\xi_{j_1}|^2 + |\xi'_{j_1}|^2
\right)^{1/2}
| \sin (\theta_v - \theta_{j_1}) |
}
{
\left(
|\xi_{j_1}|^2 + |\xi'_{j_1}|^2
\right)^{1/2}
}
\ge
\frac {1}{\sqrt{\ell}}
\cdot
\frac {1}{10\ell}
=
\frac {1}{10\ell^{\frac{3}{2}}} ;  \nonumber
\\
| \eta_{v_1} |
&=&
\frac {
\left(
|\xi_{v_1}|^2 + |\xi'_{v_1}|^2
\right)^{1/2}
\left(
|\xi_{j_1}|^2 + |\xi'_{j_1}|^2
\right)^{1/2}
| \sin (\theta_{v_1} - \theta_{j_1}) |
}
{
\left(
|\xi_{j_1}|^2 + |\xi'_{j_1}|^2
\right)^{1/2}
}
\ge
\frac {1}{\sqrt{\ell}}
\cdot
\frac {1}{10\ell}
=  \frac {1}{10 \ell^{\frac{3}{2}}}.  \nonumber
\eea
Here we used \eqref{eq:33} and  \eqref{eq:trig1}:
\beq
\frac {1}{\ell}
\le
|
\xi_v \xi'_{v_1} - \xi'_v \xi_{v_1}
|
\le
\left(
|\xi_{v}|^2 + |\xi'_{v}|^2
\right)^{1/2}
\left(
|\xi_{v_1}|^2 + |\xi'_{v_1}|^2
\right)^{1/2}
\le
\sqrt{2}
\left(
|\xi_{v_1}|^2 + |\xi'_{v_1}|^2
\right)^{1/2}.
\eeq
Furthermore, $\eta$ is an approximate eigenvector for $H_\ell$:
\bel{eq:ev1} 
\|  (H_{\ell} - E_0) \eta  \|   \le  \frac {|\xi_{j_1}|  \| (H_{\ell} - E_0) \xi' \|
+
|\xi'_{j_1}|  \| (H_{\ell} - E_0) \xi \|  }{  ( |\xi_{j_1}|^2 + | \xi'_{j_1} |^2 )^{1/2}}
<
2 \delta
\ee
%%%%%%%%%%%%%%%%%%%%%%%%%%%%%%%%
We define restrictions of $H_\ell$ to small intervals determined by the index $j_1$ by letting 
\beq
H^{(1)}
&:=&
H |_{[1, j_1-1]}, 
\quad
H^{(2)}
:=
H |_{[j_1+1, \ell]}  ,
\eeq
  and corresponding vectors 
  \beq
\eta^{(1)}
&:=&
\eta |_{1, j_1-1]}, 
\quad
\eta^{(2)}
:=
\eta |_{[j_1+1, \ell]}  .
\eeq
By \eqref{eq:eta1}-\eqref{eq:eta2}, these vectors have norms bounded above and below by
\bea\label{eq:eta_norm1}
 \frac{1}{10 \ell^{\frac{3}{2}} } \leq &   \| \eta^{(1)} \| \nonumber \\
   \frac{1}{10 \ell^{{2}}} & \leq     \| \eta^{(2)} \|   .
\eea
We let $\tilde{\eta}^{(k)} := \eta{(k)} \| \eta^{(k)} \|^{-1}$ denote the corresponding unit vectors. 

These vectors $\tilde{\eta}^{(1)}$ and $\tilde{\eta}^{(1)}$ are approximate eigenvectors for $H^{(1)}$ and $H^{(2)}$, respectively: 
\bel{approx_ev1}
\left\| (H^{(1)} - E_0)  \tilde{\eta}^{(1)} \right\|
 \leq 2 \ell^2 \delta,  %~~~~k = 1,2. 
\ee
and
\bel{approx_ev2}
\left\| (H^{(2)} - E_0)  \tilde{\eta}^{(2)} \right\|
 \leq 2 \ell^\frac{3}{2} \delta. %,  ~~~~k = 1,2. 
\ee
 %\quad
%\left\| (H^{(2)} - E_0) \eta^{(2)}  \right\|  <  2 \delta  .
%\eeq   
%
%\leadsto
%\quad
%&&
%\left\|
%(H^{(1)} - E_0)
%\frac {\eta^{(1)}}{ \| \eta^{(1)} \| }
%\right\|
%<
%2 \delta \cdot 10\ell^{3/2}, 
%\quad
%\left\|
%(H^{(2)} - E_0)
%\frac {\eta^{(2)}}{ \| \eta^{(2)} \| }
%\right\|
%<
%2 \delta \cdot 10 \sqrt{2}\ell^2
%\\
%%
%\leadsto
%\quad
%&&
%d (E_0, \sigma (H^{(1)})) \le 2 \delta  \cdot 10\ell^{3/2}, 
%\quad
%d (E_0, \sigma (H^{(2)})) 
%\le 
%2 \delta \cdot 10 \sqrt{2}\ell^{2}
%=
%2 \cdot c_2 \cdot e^{- \frac {\ell}{c_1}} \cdot10 \sqrt{2} \ell^2
%\eeq
%%
We take 
$\ell \gg 1$
s.t.
\beq
2 \cdot c_2 \cdot e^{- \frac {\ell^{\beta}}{c_1}} \cdot 10 \sqrt{2}\ell^2
<
e^{ 
- \frac {\ell^{\beta}}{c_1 (1 + \epsilon)}
}
\eeq
We take  $0 < t < 1$  so that  the size  of the box where 
$H^{(1)}$ lives is equal to 
$j_1 = t \cdot \ell$
and  that where  $H^{(2)}$
lives is equal to $\ell - j_1 = (1-t) \cdot \ell$.
By construction, the two \Schr operators $H{(k)}$, for $k=1,2$, are independent so that the second version of the Wegner estimate in 
Corollary \ref{corollary:wegner_v2} can be used to the joint probability that each \Schr operator has an eigenvalue in a small interval:
\beq
&&
{\Pp}\Bigl\{  d (E_0, \sigma (H^{(1)})) \le 2 \delta \ell^{3/2}, 
\quad
d (E_0, \sigma (H^{(2)})) 
\le 
2 \delta \ell^{2}
\Bigr\}
\\
& \le &
{\Pp} \Bigl\{  d (E_0, \sigma (H^{(1)})) 
\le e^{ - \frac {\ell^{\beta}}{c_1 (1 + \epsilon)}
}, 
\quad
d (E_0, \sigma (H^{(2)})) 
\le 
e^{  - \frac {\ell^{\beta}}{c_1 (1 + \epsilon)}
}
\Bigr\}  
\\
&=&
{\Pp}  \Bigl\{  d (E_0, \sigma (H^{(1)})) \le  
 e^{ - \frac {1}{c_1 (1 + \epsilon)t^{\beta}} \cdot (t \ell)^{\beta}  }, 
\quad
d (E_0, \sigma (H^{(2)})) 
\le e^{ - \frac {1}{c_1 (1 + \epsilon)(1 - t)^{\beta}} 
\cdot 
((1-t)\ell)^{\beta}
}
\Bigr\}
\\
& \le &
\left( e^{ - \frac {\rho_1}{c_1 (1 + \epsilon) t^{\beta}} \cdot (t \ell)^{\beta} } +
e^{- q (t \ell)^{\beta}}  \right)
\left( e^{ - \frac {\rho_1}{c_1 (1 + \epsilon) (1-t)^{\beta}}  ((1-t) \ell)^{\beta}
}
+
e^{- q ((1-t) \ell)^{\beta} }
\right)
\\
&=&
e^{ 
- \frac {2\rho_1}{c_1 (1 + \epsilon) } 
\cdot  \ell^{\beta} 
}
+
e^{ - \bigl(  \frac {\rho_1}{c_1 (1 + \epsilon)}
+ q(1-t)^{\beta}
\bigr)
\ell^{\beta} }
+
e^{ - \bigl(  \frac {\rho_1}{c_1 (1 + \epsilon)}
+ 
q t^{\beta}  \bigr)  \ell^{\beta}  }  
+  
e^{- q (t^{\beta}+(1-t)^{\beta})\ell^{\beta}}
\\
& \le &
e^{ 
- \frac {2\rho_1}{c_1 (1 + \epsilon) } 
\cdot  \ell^{\beta} 
}
+
e^{
- \frac {2 \rho_1}{c_1 (1+\epsilon)} \ell^{\beta}}
+
e^{-2 q t_0 \ell^{\beta}}
+
e^{-q \ell^{\beta}}.
\eeq
In the last inequality, 
we used 
$xy \le 
\frac {x^2 +y^2}{2}$, 
and 
$t_0 < t < 1 - t_0$, 
$t_0 < 1-t$, 
and 
$\beta \le 1$.
$q$ 
is defined in \eqref{eq:const_q1}. 
This established the bound in Case 1.

\medskip
%%%%%%%%%%%%%%%%%%%%%%%%%%%%%%%%%%%%%%%%%%%%%%%%%%%%%%%%%%%%%%%%%%%%%%%%%%%%%

\noindent
{\bf Case 2  }
Suppose that, for any 
$v \le j \le v_1$,
either of following two situations happens : 
\beq
| \sin (\theta_v - \theta_j) | 
\le
\frac {1}{10\ell} 
\quad
\mbox{ or }
\quad
| \sin (\theta_{v_1} - \theta_j) | 
\le
\frac {1}{10\ell}.
\eeq
We denote by $j_1$ the smallest polymer node index $v < j_1 < v_1$ so that
$| \sin (\theta_{v_1} - \theta_{j_1}) | \le \frac {1}{10\ell} $. 
There are two subcases to consider: i) the index $j_1$ is a left polymer node, and ii) the index $j_1$ is a right polymer node. 

\medskip
%%%%%%%%%%%%%%%%%%%%%%%%%%%%%%%%%%%%%%%%%%%%%%%%%%%%%%%%%%%%%%%%%%%%%%%%%%%%%

\noindent
{\bf Case 2i }
For this case, the index $j_1$ is a left polymer node. We then have that $j_1 - 1$ is a right polymer node index. It follows from the definition of $j_1$ that 
\beq
| \sin (\theta_{v_1} - \theta_{j_1-1}) | 
> \frac {1}{10\ell}.
\eeq
We define two vectors 
\beq
\eta^{(1)}
&:=&
\frac {
\xi_{j_1} \xi' - \xi'_{j_1} \xi
}
{
( |\xi_{j_1}|^2 + |\xi'_{j_1}|^2 )^{1/2}
}, ~~~\mbox{for indices} ~~~ k \in [ 1, \ldots, j_1-1] ,
\eeq
and 
 %\cdot
%1 (1 \le j \le j_1 - 1),
%\quad
%%
\beq
\eta^{(2)}
&:=&
\frac {
\xi_{j_1-1} \xi' - \xi'_{j_1-1} \xi
}
{
( |\xi_{j_1-1}|^2 + |\xi'_{j_1-1}|^2 )^{1/2}
}, ~~~\mbox{for indices} ~~~ k \in [ j_1, \ldots, \ell]. 
\eeq
%\cdot
%1 (j_1 \le j \le \ell)
%\\
%
Note that $\eta^{(1)}_{j_1} = 0 $ and $\eta^{(2)}_{j_1-1} = 0$. Furthermore, we have the lower bounds on the norms of these vectors:
\beq
\| \eta^{(1)} \| &\ge&  | \eta^{(1)}_v |
=  ( |\xi_v|^2 + |\xi'_v|^2 )^{1/2} 
| \sin (\theta_v - \theta_{j_1}) |
\ge   \frac {1}{\sqrt{\ell}}
\Bigl( | \sin (\theta_v - \theta_{v_1}) |  - | \sin (\theta_{v_1} - \theta_{j_1}) |
\Bigr)
\\
& \ge &
\frac {1}{\sqrt{\ell}}
\left(  \frac {1}{\ell} - \frac {1}{10 \ell}  \right)
\ge    \frac {C_1}{\ell^{3/2}}  ,   \\
\| \eta^{(2)} \|
&\ge&  | \eta^{(2)}_{v_1} | =  ( |\xi_{v_1}|^2 + |\xi'_{v_1}|^2 )^{1/2} 
| \sin (\theta_{v_1} - \theta_{j_1-1}) |
\ge   \frac {1}{\sqrt{\ell}}
\Bigl(  | \sin (\theta_{v_1} - \theta_{v}) |  -  | \sin (\theta_{v} - \theta_{j_1-1}) |
\Bigr)  \\
& \ge  &  \frac {1}{\ell}  \left(   \frac {1}{\ell} - \frac {1}{10 \ell}  \right)
\ge   \frac {C_2}{\ell^{2}}   ,
\eeq
for finite constants $C_k > 0$, $k=1,2$, independent of $\ell$.
Here we have used: 
\beq
| \sin (x+y) |
&=&
| \cos x | \cdot | \sin y |
+
| \cos y | \cdot | \sin x |
\le
| \sin x | + | \sin y |.
\eeq

\medskip

\noindent
We define the associated local \Schr operators as restriction to each subinterval:
\beq
H^{(1)}
&:=&
H |_{[1, j_1-1]}, 
~~~~~ \mbox{and}  ~~~~H^{(2)}  :=
H |_{[ j_1, \ell]}. 
\eeq
As above, we let $\tilde{\eta}^{(k)}$ denote the corresponding normalized vectors. As a result, we find that $\eta^{(k)}$ are approximate eigenvectors for $H^{(k)}$, respectively:
%
%\leadsto
%\;
\beq 
\| (H^{(1)} - E_0) \tilde{\eta}^{(1)} \| < 2 \ell^2 \delta, 
\quad
\| (H^{(2)} - E_0) \tilde{\eta}^{(2)} \| < 2 \ell^\frac{3}{2} \delta  .
\eeq
The rest of the argument is similar to that in Case 1.

  \medskip

%%%%%%%%%%%%%%%%%%%%%%%%%%%%%%%%%%%%%%%%%%%%%%%%%%%%%%%%%%%%%%%%%%%%%%%%%%%%

\noindent
{\bf Case 2ii}  For this case, the index $j_1$ is a right polymer node. We then have that $j_1 - 1$ is a left polymer node index and it follows from the definition that 
\beq
| \sin (\theta_{v} - \theta_{j_1-1}) | 
> \frac {1}{10\ell}.
\eeq
We define two vectors 
\beq
\eta^{(1)}
&:=&
\frac { \xi_{j_1-1} \xi' - \xi'_{j_1-1} \xi
}
{
( |\xi_{j_1-1}|^2 + |\xi'_{j_1-1}|^2 )^{1/2}
}, ~~~\mbox{for indices} ~~~~ k \in [ 1, \ldots, j_1-1] ,
\eeq
and
\beq
\eta^{(2)}
&:=&
\frac {
\xi_{j_1} \xi' - \xi'_{j_1} \xi
}
{
( |\xi_{j_1}|^2 + |\xi'_{j_1}|^2 )^{1/2}
}, ~~~\mbox{for indices}  ~~~~ k \in [ j_1, \ldots, \ell]. 
\eeq
Note that $\eta^{(1)}_{j_1-1} = 0 $ and $\eta^{(2)}_{j_1} = 0$. Furthermore, we have the lower bounds
\beq
\| \eta^{(1)} \|
&\ge&
| \eta^{(1)}_v |
=
( |\xi_v|^2 + |\xi'_v|^2 )^{1/2} 
| \sin (\theta_v - \theta_{j_1}) |
\ge
\frac {1}{\sqrt{\ell}}
\Bigl(
| \sin (\theta_v - \theta_{v_1}) |
-
| \sin (\theta_{v_1} - \theta_{j_1}) |
\Bigr)
\\
& \ge &
\frac {1}{\sqrt{\ell}}
\left(
\frac {1}{\ell} - \frac {1}{10 \ell}
\right)
\ge
\frac {C_1}{\ell^{3/2}}
\\
\| \eta^{(2)} \|
&\ge&
| \eta^{(2)}_{v_1} |
=
( |\xi_{v_1}|^2 + |\xi'_{v_1}|^2 )^{1/2} 
| \sin (\theta_{v_1} - \theta_{j_1-1}) |
\ge
\frac {1}{\sqrt{\ell}}
\Bigl( | \sin (\theta_{v_1} - \theta_{v}) | - | \sin (\theta_{v} - \theta_{j_1-1}) |
\Bigr)    \\
& \ge &
\frac {1}{\ell}
\left(
\frac {1}{\ell} - \frac {1}{10 \ell}
\right)
\ge  \frac {C_2}{\ell^{2}}.
\eeq
Consequently, the vectors $\eta^{(k)}$ are approximately eigenvectors for $H^{(k)}$, $k=1,2$, and the proof continues as in Case 1.   
\end{proof}

%%%%%%%%%%%%%%%%%%%%%%%%%%%%%%%%%%%%%%%%%%%%%%%%%%%%%%%%%%%%%%%%%%%%%%%%%%%%%%%
\section{Transport and delocalization}\label{app:deloc1}
\setcounter{equation}{0}

We note that if an initial state is concentrated in an energy interval $I$ containing the set $\mathcal{C}$ of critical energies, then there is dynamical delocalization. That is, the wave packet $e^{-i Ht} P(I) \delta_0$ exhibits nontrivial transport in the sense of \eqref{deloc1}. Let $\{ \delta_j \}_{j \in Z}$ denote the canonical orthonormal basis of $\ell^2 (\Z)$.  The authors of \cite{jss} proved dynamical delocalization in the sense of \eqref{deloc1} for the AB RPM and initial state $\delta_0$. In this appendix, we show that the initial state $\delta_0$ may be replaced by $P_H(I) \delta_j$, for any $j \in \Z$, and for any interval $I \subset \Sigma$ containing all the critical energies. This simple result relies on the complementary fact that if $J \subset \Sigma$ is an interval disjoint from the finite set of critical energies $\mathcal{C} \cap J = \emptyset$, then we have dynamical localization for the initial state $P_H(J) \delta_j$ as proven for the AB RPM, and other models, in \cite[Theorem 2.3]{DSS}.  

\begin{proposition}\label{prop:loc_deloc1}
Let  $I \subset \Sigma$ be any union of intervals satisfying $\mathcal{C} \subset I$.  Then, for any $\alpha > 0$, there is a constant $C_\alpha(I) > 0$ so that for any $j \in \Z$, 
\bel{deloc_app1}
\frac{1}{T} \int_0^T \langle P(I) \delta_j,  e^{i H_\omega t}|X|^q e^{- i H_\omega t} P(I) \delta_j \rangle ~dt \geq C_\alpha (I) T^{q - \frac{1}{2} - \alpha}, a.s.
\ee
\end{proposition}

\begin{proof}
We first consider $j=0$. Since $P(I) + P(I^c) = {\bf 1}$, we have
\bea\label{eq:spect_decomp1}
\langle P(I) \delta_0,  e^{i H_\omega t}|X|^q e^{- i H_\omega t} P(I) \delta_0 \rangle 
 & = & \langle \delta_0,  e^{i H_\omega t}|X|^q e^{- i H_\omega t} \delta_0 \rangle 
  \nonumber \\
   & & - 2 \Re \langle P(I^c) \delta_0,  e^{i H_\omega t}|X|^q e^{- i H_\omega t} \delta_0 \rangle  \nonumber \\
    & & + \langle  P(I^c)\delta_0,  e^{i H_\omega t}|X|^q e^{- i H_\omega t} 
    P(I^c)\delta_0 \rangle  \nonumber   \\
     & =: &  I - II + III.
     \eea
 Dynamical localization in $I^c$ implies there is a constant $C_{III} >0$ so that 
 \bel{eq:dyn_loc1}
 \langle P(I^c) \delta_0,  e^{i H_\omega t}|X|^q e^{- i H_\omega t} P(I^c) \delta_0  \rangle    \leq C_{III} . 
 \ee     
 To bound $II$, the Cauchy-Schwarz inequality gives
 \bea\label{eq:dyn_loc2}
 | \langle P(I^c) \delta_0,  e^{i H_\omega t}|X|^q e^{- i H_\omega t} \delta_0 \rangle  |  & \leq &  \||X|^q  e^{-i H_\omega t} P(I^c) \delta_0 \| ~ \| e^{- i H_\omega t} \delta_0 \|  \nonumber \\
  & \leq & C_{II} , 
 \eea
 as follows from dynamical localization.  From \eqref{eq:dyn_loc1},  \eqref{eq:dyn_loc2}, the lower bound estimate \eqref{deloc1} of \cite{jss}, and the equality \eqref{eq:spect_decomp1}, we obtain
 \bel{eq:lower_bd1}
 M_q(T) \geq \widetilde{C}_\alpha T^{q - \frac{1}{2} - \alpha}, a.s. ,
 \ee
 for a possibly different constant. The result follows for arbitrary $j \in \Z$ from the unitarity of the translation group, the relation $\delta_j = U_j \delta_0$, and the covariance relation of $H_\omega$ with respect to this group. 
 \end{proof}

 %%%%%%%%%%%%%%%%%%%%%%%%%%%%%%%%%%%%%%%
\section{Sharpness of the transition}\label{app:sharp1}
\setcounter{equation}{0}

In this appendix, we provide the details of the results mentioned in Remark \ref{rmk:sharp1}.
These remarks  concern  the transition in the unfolded LES as the centering energy $E_0$ passes from a noncritical energy to a critical energy. 
 In the absence of Assumption \ref{assump:dos1}, we prove that the unfolded LES changes discontinuously from an infinitely divisible point process at noncritical energies to a non-infinitely divisible point process (a clock process) at a critical energy $E_0$. We prove this by showing that certain limiting random variables, constructed from the unfolded LES at noncritical energies, are infinitely divisible whereas the clock process is not infinitely divisible.

\begin{defn}\label{defn:idrv1}
A random variable
$X$ is infinitely divisible if and only if for any $J \in {\N}$, 
there exist iid  random variables $\{ Y_{N, j} \}_{j=1}^J$
so that
\beq
X  &  \stackrel{d}{=}  &
\sum_{j=1}^J
Y_{N, j}.
\eeq
\end{defn}

We characterize the random variables associated with a clock process. The following theorem characterizes bounded infinitely divisible random variables.

\begin{theorem}\label{thm:bddRV} \cite[Corollary 3]{bs60} %
A bounded random variable is infinitely divisible if and only if it is constant almost surely.
\end{theorem}

The proof of Theorem \ref{thm:bddRV} is based on a careful analysis of the L\'evy-Khintchine formula \cite[Theorem 1.2.1]{applebaum1} that provides a representation of the characteristic function of an infinitely divisible random variable. 

We first show that, 
any limit point of LES at critical energies are not infinitely divisible, under a mild condition. 
\begin{theorem}\label{thm:clock_notID1}
Suppose that a critical energy 
$E_c \in \mathcal{C}$
satisfies 

(i)
$| \langle e^{ 2 i \eta_{\pm} } \rangle | < 1$, 
and
(ii)
$0 < N(E_c) < 1$.\\
Let 
$\xi$
be an accumulation point of 
$\{ \xi_L \}$
centered at
$E_c$.
Then
$\xi$
is not infinitely divisible.
\end{theorem}
\begin{proof}
The condition (i) 
and Proposition \ref{prop:rel_prufer_conv1}
yield
$\Psi_L (x) \to x$
a.s., so together with the condition (ii), Corollary \ref{rmk:clock_assump2}
implies that an accumulation point 
$\xi$ 
of 
$\xi_L$
is a clock process with a probability measure 
$\mu$
on 
$[0,1)$.
By the condition (i),
$\mu$
is not a delta measure at a single point. 
Thus 
we can find an interval
$I \subset [0, 1)$
such that
$\xi (I)$
is a non-constant bounded random variable.
Therefore, 
by Theorem \ref{thm:bddRV}, 
$\xi(I)$,
and hence
$\xi$ 
itself, 
is not infinitely divisible.
\end{proof}
%

%

%\end{follow}
%
%Let 
%$I = [0, 1/2)$, 
%and let 
%$\xi \sim \text{clock} (\text{unif}[0,1))$.
%%
%Then 
%$\xi (I)$
%is not infinitely divisible.

\medskip

In order to show there is a transition, we need to characterize the limit points of the sequence $\{ \xi_L(I) \}$ for non-critical energies.  
We still
need a condition on the regularity of IDS, but this condition seems to be  reasonable.
%

%
%%%%%%%%%
\begin{assumption}
\label{assumpt:IDS2}
%{\bf Assumption 2}
%
IDS
$N$
and its inverse
$N^{-1}$
satisfy the H\"older continuity in the following sense near the reference energy 
$E_0$
for some 
$\rho_1, \rho_2 > 0$ : 
\beq
0 < 
\lim_{ \delta \downarrow 0}
\inf_{
\substack{ 
|E_1 - E_c| < \delta, 
\\
|E_2 - E_c| < \delta, 
\\
E_1 \ne E_2
}
}
\frac {
| N(E_1) - N(E_2) |
}
{
|E_1 - E_2|^{1/\rho_2}
}, 
\quad
\lim_{ \delta \downarrow 0}
\sup_{
\substack{ 
|E_1 - E_c| < \delta, 
\\
|E_2 - E_c| < \delta, 
\\
E_1 \ne E_2
}
}
\frac {
| N(E_1) - N(E_2) |
}
{
|E_1 - E_2|^{\rho_1}
}
< \infty.
\eeq

\end{assumption}
%%%
%

%
%%%%%%%%%
\begin{theorem}
\label{thm:infinite_divisibility}
%{\bf Theorem (infinite divisibility)}
%
Assume Assumption 2.
Let 
$\xi$
be an accumulation point of the LES
$(\xi_L)_L$
centered at 
$E_0 \in \Sigma \setminus \mathcal{C}$.
Then 
for each bounded interval
$I (\subset {\bf R})$, 
$\xi (I)$
is infinitely divisible.
\end{theorem}
%%%
%

%
Therefore, 
Theorems 
\ref{thm:clock_notID1} and \ref{thm:infinite_divisibility} 
show the sharpness of the transition.
In what follows, 
we give proof of Theorem \ref{thm:infinite_divisibility} for completeness, although the idea of which is similar to that for Theorem \ref{thm:Poisson}.
We recall the notations.\\
$\Lambda_L$ : box of size 
$L$, \\
$I_{\Lambda_L}
:=
N^{-1}
\left(
N(E_0) + \frac {I}{| \Lambda_L |}
\right)$, 
where 
$I$ : bounded interval,
\\
$\ell := L^{\beta}$, 
$\ell' := L^{\beta'}$, 
$0 < \beta' < \beta < 1$,
\\
$\Lambda_{\ell} (\gamma)$ : 
a box of size
$\ell$
centered at 
$\gamma$.
Then we decompose
$
\Lambda_L
=
\bigcup_j
\Lambda_{\ell} (\gamma_j) 
\cup
\Upsilon
$
where we define  
$\Upsilon$
to be the set 
$\Lambda_L \setminus \bigcup_j \Lambda_{\ell}(\gamma_j)$
enlarged by a length 
$\ell'$.
Let 
\beq
\xi_L (I)
&=&
\#
\left\{
\text{ eigenvalues of 
$H_{\Lambda}$
in 
$I_{\Lambda_L}$ }
\right\}
\\
\eta_{j, L} (I)
&=&
\#
\left\{
\text{ eigenvalues of 
$H_{\Lambda_{\ell} (\gamma_j)}$
in 
$I_{\Lambda_L}$ }
\right\}.
\eeq
The
following lemma is an analogue to the decomposition theorem (Theorem \ref{thm:decomposition}), but under much less requirements.
%
%%%%%
\begin{lemma}
\label{lemma:Y}
%{\bf Lemma Y}
%
Let 
\beq
\mathcal{ Y}_{\Lambda_L}
&:=&
\left\{
\omega \in \Omega
\, \middle| \,
\text{ COL's of 
$H_{\Lambda_L}$ 
in
$I_{\Lambda_L}$ 
are not in 
$\Upsilon$ }
\right\}.
\eeq
Then we can choose constants  
$\beta, \beta' > 0$
such that 
\beq
\Pp( \mathcal{ Y}_{\Lambda_L} )
\ge
1 - C L^{- \alpha}
\eeq
for some 
$\alpha > 0$.
\end{lemma}
%%%%%
%
\begin{proof}
Proof 
of lemma \ref{lemma:Y} is similar to that of Theorem \ref{thm:decomposition}, but since we do not use Minami's estimate, we do not need the condition that 
$\rho_1 \cdot \rho_2 > \frac 23$.
In fact, 
by the argument in the proof of Theorem \ref{thm:decomposition}, we have
\beq
&&
\Pp
\left(
\mbox{ 
$H_{\Lambda_L}$
has eigenvalues in 
$I_{\Lambda_L}$
centered in 
$\Upsilon$
}
; \mathcal{ U}_{\Lambda} 
\right)
\le
C L^{- \alpha}
\quad
\text{for some}
\;
\alpha.
\eeq
To have 
$\alpha > 0$, 
we need
$\rho_1 \cdot \rho_2
>
1 - \beta$.
We can choose 
$\beta > 0$
sufficiently close to 
$1$
if necessary such that this condition is satisfied.
Hence
\beq
\Pp
(\mathcal{ Y}_{\Lambda_L}^c \cap \mathcal{ U}_{\Lambda_L})
\le
C L^{- \alpha}.
\eeq
Since
$\Pp( \mathcal{ U}_{\Lambda_L}^c ) \le L^{ - p' }$, 
and since 
$p'$
can be taken arbitrary large, we have
\beq
\Pp
( \mathcal{ Y}^c_{\Lambda_L})
&\le&
\Pp
( \mathcal{ Y}_{\Lambda_L}^c \cap \mathcal{ U}_{\Lambda_L} )
+
\Pp
( \mathcal{ U}^c_{\Lambda_L})
\le
L^{- \alpha} + L^{- p'}
\le
C L^{- \alpha}.
\eeq
\QED
\end{proof}
In the following lemma, 
we compare the large system 
$H_{\Lambda_L}$
with direct sum of subsystems
$\bigoplus_j H_{ \Lambda_{\ell} (\gamma_j) }$.
%
%%%%%%%%%
\begin{lemma}
\label{lemma:eta}
%{\bf Lemma Eta}
%
Let 
$I (\subset \R)$
be a interval.
Then 
there is a set of configurations 
$\Omega_{\Lambda_L}$
with 
$\Pp(\Omega_{\Lambda_L}^c)
\le
L^{- \alpha}$ 
such that 
\beq
\xi_L (I)
&=&
\sum_j 
\eta_{j, L} (I), 
\quad
\omega \in \Omega_{\Lambda_L}.
\eeq
\end{lemma}
%%%
%

%
\begin{proof}
In the event 
$\mathcal{ U}_{\Lambda_L} \cap \mathcal{ Y}_{\Lambda_L}$, 
let 
$\varphi$
be a normalized eigenvector of 
$H_{\Lambda_L}$
in 
$I_{\Lambda_L}$
and let 
$x_{\varphi}$
be its COL.
Since
$\omega \in \mathcal{ Y}_{\Lambda}$, 
$x_{\varphi} \in \Lambda_{\ell - \ell'} (\gamma_j)$
for some 
$\gamma_j$.
Then
\beq
| \varphi (x) |
\le
C 
e^{ - c( \ell' )^{\xi} }
=
C
e^{- c L^{\beta' \xi} }, 
\quad
x \in \partial \Lambda_{\ell}(\gamma_j).
\eeq
Then
$\varphi |_{\Lambda_{\ell}(\gamma_j)}$
becomes an approximate eigenvector of 
$H_{\Lambda_{\ell}(\gamma_j)}$
so that 
$H_{\Lambda_{\ell}(\gamma_j)}$
has an eigenvalue in 
\beq
\widetilde{I}_{\Lambda_L}
:=
I_{\Lambda_L}
+
[ -e^{- c L^{\beta' \xi} }, e^{- c L^{\beta' \xi} } ].
\eeq
By Wegner's estimate(Corollary \ref{corollary:wegner_v2}), 
\beq
\Pp
\left(
\exists j 
\text{ s.t. }
H_{\Lambda_{\ell}(\gamma_j)}
\text{ has an eigenvalue in  }
\widetilde{I}_{\Lambda_L} \setminus I_{\Lambda_L}
\right)
& \le &
C_W
L^{1 - \beta}
\left(
e^{- c \rho_1 L^{\beta' \xi} }
+
e^{- c L^{\beta'}}
\right)
\\
&\le&
e^{- c' \rho_1 L^{\beta' \xi} }, 
\eeq
for some 
$0 < c' < c$. 
We have used the fact that, since
$0 < \xi < 1$, 
2nd term in RHS is negligible.
Therefore
\beq
\mathcal{ W}_{\Lambda_L}
&:=&
\left\{
\text{ 
$H_{\Lambda_{\ell}(\gamma_j)}$
does not have eigenvalue in 
$\widetilde{I}_{\Lambda_L} \setminus I_{\Lambda_L}$
for any 
$j$
%and any
%$k$
}
\right\}
\eeq
satisfies
$
\Pp(
\mathcal{ W}_{\Lambda_L}^c \cap \mathcal{ U}_{\Lambda_L} \cap \mathcal{ Y}_{\Lambda_L}
)
\le  
e^{- c' \rho_1 L^{\beta' \xi} }.
$
On the other hand, we have shown that
$\Pp(\mathcal{ U}_{\Lambda_L}^c)
\le
C L^{- p'}$, 
and 
$
\Pp(\mathcal{ Y}_{\Lambda_L}^c)
\le
C L^{- \alpha}.
$
Therefore
\beq
\Pp
( \mathcal{ W}_{\Lambda_L}^c)
& \le &
\Pp
(
\mathcal{ W}_{\Lambda_L}^c \cap \mathcal{ U}_{\Lambda_L} \cap \mathcal{ Y}_{\Lambda_L}
)
+
\Pp
(\mathcal{ U}_{\Lambda_L}^c)
+
\Pp
(\mathcal{ Y}_{\Lambda_L}^c)
\le  C L^{- \alpha}.
\eeq
Now, on the event
$\mathcal{ W}_{\Lambda_L} \cap \mathcal{ U}_{\Lambda_L} \cap \mathcal{ Y}_{\Lambda_L}$, 
for each eigenvalue 
$E_n (\Lambda_L)$
of 
$H_{\Lambda_L}$ 
in 
$I_{\Lambda_L}$, 
we can find 
$j = j(n)$
and corresponding eigenvalue 
$E_{m(n)} (\Lambda_{\ell}(\gamma_j))$
of 
$H_{\Lambda_{\ell}(\gamma_j)}$
in 
$I_{\Lambda_L}$.
Since
eigenfunctions of 
$H_{\Lambda_L}$
are orthogonal each other, this correspondence is one to one
\cite[Lemma A.4]{nakano1}. 
Therefore,
\begin{equation}
\xi_L (I_k)
\le 
\sum_j
\eta_{j, L} (I), 
\quad
\omega \in 
\mathcal{ W}_{\Lambda_L} \cap \mathcal{ U}_{\Lambda_L} \cap \mathcal{ Y}_{\Lambda_L}
\label{eq:eta1}
%\quad\cdots (1)
\end{equation}
which gives the upper bound.
For the lower bound, we set 
\beq
\widetilde{\mathcal{ U}}_{\Lambda_L}
&:=&
\bigcap_{j}
\mathcal{ U}_{ \Lambda_{\ell} (\gamma_j) }.
\\
\widetilde{\mathcal{ W}}_{\Lambda_L}
& := &
\left\{
\text{ 
$H_{\Lambda_L}$
does not have an eigenvalue in 
$\widetilde{I}_{\Lambda_L} 
\setminus
I_{\Lambda_L}$ 
}
\right\}.
\eeq
Then 
an analogous argument as in upper bound yields
\begin{equation}
\sum_j
\eta_{j, L} (I)
\le 
\xi_L (I), 
\quad
\omega \in 
\widetilde{\mathcal{ W}}_{\Lambda_L} 
\cap 
\widetilde{\mathcal{ U}}_{\Lambda_L}.
%\cap 
%{\cal Z}_{\Lambda}
\label{eq:eta2}
%\quad\cdots (2)
\end{equation}
By (\ref{eq:eta1}), (\ref{eq:eta2}), we have
\beq
\xi_L (I)
& = &
\sum_j
\eta_{j, L} (I), 
\quad
\omega \in 
\Omega_{\Lambda_L}
:=
\mathcal{ W}_{\Lambda_L} \cap \mathcal{ U}_{\Lambda_L} \cap \mathcal{ Y}_{\Lambda_L}
\cap 
\widetilde{\mathcal{ W}}_{\Lambda_L} 
\cap 
\widetilde{\mathcal{ U}}_{\Lambda_L} 
\eeq
with 
$
\Pp
\left(
\Omega_{ \Lambda_L }^c
\right)
\le  C L^{- \alpha}.
$

\end{proof}
\noindent
{\it Proof of Theorem \ref{thm:infinite_divisibility}}\\
Let 
$I (\subset \R)$
be an interval, and let 
$K \in \N$.
It suffices to show 
the following fact : 
There exist
$K$
i.i.d. random varibles
$(\zeta_k (I))_{k=1}^K$
such that 
\beq
\xi(I)
\stackrel{d}{=}
\sum_{k=1}^K
\zeta_k (I).
\eeq
We decompose
$\Lambda_L$
into 
$K$
disjoint boxes
$J_1, J_2, \cdots, J_K$
: 
$\Lambda_L = \bigcup_{k=1}^K J_k$.
Here 
we adjust the value of 
$\beta = \beta_L$
and the definition of 
$\Upsilon$
slightly if necessary such that each
$J_k$'s
contains equal number of 
$\Lambda_{\ell}(\gamma_j)$'s.
And 
we also decompose the sum of 
$\eta_{j, L}$'s
into sum of 
$K$ 
independent ones.
\beq
\sum_{j=1}^{N_L}
\eta_{j, L} (I)
&=&
\sum_{k=1}^K
\left(
\sum_{j \,:\, \Lambda_{\ell}(\gamma_j) \subset J_k}
\eta_{j, L} (I)
\right).
\eeq
Let 
$H_k 
:=
H |_{J_k}$, 
and let 
$\zeta_k^{(L)}$
be the corresponding point process.
\beq
\zeta_k^{(L)}
&:=&
\sum_{j}
\delta_{
L (N (E_j(J_k)) - N(E_0) )
}
\eeq
where 
$\{ E_j (J_k) \}_j$
are the eigenvalues of 
$H_k$.
By Lemma \ref{lemma:eta}, 
we have
\beq
\xi_L (I)
&=&
\sum_{j=1}
\eta_{j, L} (I), 
\quad
\omega \in \Omega_{ \Lambda_L }, 
\\
\zeta_k^{(L)} (I)
& = &
\sum_{j \,:\, \Lambda_{\ell}(\gamma_j) \subset J_k}
\eta_{j, L} (I), 
\quad
\omega \in 
\Omega_{J_k}
, 
\quad
k = 1, 2, \cdots, K, 
\eeq
with 
$
\Pp
\left(
\Omega_L^c
\cup
\bigcup_{k=1}^K
\Omega_{J_k}^c
\right)
\le  C L^{- \alpha}.
$
Therefore, 
we have
\begin{equation}
\xi_L (I)
-
\sum_{k=1}^K
\zeta_k^{(L)}(I)
\stackrel{P}{\to}
0, 
\quad
L \to \infty.
\label{eq:infinite_divisibility1}
%\quad\cdots (1)
\end{equation}
Since
$\xi$
is an accumulation point of 
$\xi_L$, 
there is a subsequence 
$\{ L(i) \}_i$
such that
$\xi_{L(i)} 
\stackrel{d}{\to}
\xi$.
Since 
$\{ \zeta_k^{(L(i))} (I)  \}_{i=1}^{\infty}$
is tight for any 
$k=1, \cdots, K$, 
they are jointly tight.
Thus
we can further find a subsequence 
$\{ L(i_{\ell}) \}_{\ell}$
of
$\{ L(i) \}_i$ 
and random variables
$\zeta_k (I)$, 
$k=1, 2, \cdots, K$, 
such that we jointly have
\begin{equation}
\xi_{L(i_{\ell})} (I)
\stackrel{d}{\to}
\xi(I), 
\quad
\sum_{k=1}^K
\zeta_k^{(L(i_{\ell}))}(I)
\stackrel{d}{\to}
\sum_{k=1}^K
\zeta_k(I), 
\quad
\ell \to \infty.
\label{eq:infinite_divisibility2}
%\quad\cdots (2)
\end{equation}
Here 
we use the following fact on the sequence of random variables : 
$X_n \stackrel{d}{\to} X$, 
and 
$Y_n \stackrel{P}{\to} 0$
yield
$X_n + Y_n
\stackrel{d}{\to}
X$.
Therefore, (\ref{eq:infinite_divisibility1}), (\ref{eq:infinite_divisibility2})
yield
\beq
\xi (I)
\stackrel{d}{=}
\sum_k \zeta_k (I).
\eeq
\QED\\
%

%%%%%%%%%%%%%%%%%%%%%%%%%%%%%%%%%%%%%%%%%%%%%%%%%%%%%%%%%%%%%%%%%%%%%%%%%%%%%%

\end{appendices}

%%%%%%%%%%%%%%%%%%%%%%%%%%%%%%%%%%%%%%%%%%%%%%%%%%%%%%%%%%%%%%%%%%%%%%%%%%%%%%%

\end{document}